\documentclass[10pt,letterpaper]{amsart}
\usepackage{amssymb}
\usepackage[dvips]{epsfig}
\usepackage{booktabs}
\usepackage{multirow}
\usepackage{xltabular}
\usepackage{tabularx}
\usepackage{makecell} 
\usepackage{ragged2e} 
\usepackage{graphicx}
\usepackage{color}
\usepackage{amsmath}
\allowdisplaybreaks
\usepackage{amssymb}
\usepackage{amsfonts}
\usepackage{geometry}
\usepackage{amsthm}
\usepackage[normalem]{ulem}    
\usepackage{cite}
\usepackage{tikz}

\newcommand{\Z}{\mathbb{Z}}

\newcommand{\R}{\mathbb{R}}
\newcommand{\T}{\mathbb{T}}
\newcommand{\E}{\mathbb{E}}

\newcommand{\bvepsilon}{\bar{\varepsilon}}
\newcommand{\spec}{{\rm Spec}}
\newcommand{\hatepsilon}{\widehat{\varepsilon}}

\renewcommand{\P}{\mathbb{P}}

\newcommand{\mcA}{\mathcal{A}}

\newcommand{\mcC}{\mathcal{C}}
\newcommand{\mcE}{\mathcal{E}}
\newcommand{\mcF}{\mathcal{F}}
\newcommand{\mcH}{\mathcal{H}}

\newcommand{\mcS}{\mathcal{S}}
\newcommand{\mcI}{\mathcal{I}}
\newcommand{\mcR}{\mathcal{R}}

\newcommand{\mcO}{\mathcal{O}}
\newcommand{\mcG}{\mathcal{G}}

\newcommand{\mcN}{\mathcal{N}}

\newcommand{\whT}{\widehat{T}}
\newcommand{\whx}{\widehat{\xi}}
\newcommand{\wta}{\widetilde{a}}
\newcommand{\trim}{{\rm Trim}}
\newcommand{\whP}{\widehat{P}}

\renewcommand{\Re}{{\rm Re}}
\newcommand{\dist}{{\rm dist}}

\newcommand{\supp}{{\rm supp}}
\newcommand{\proj}{{\rm Proj}}

\newtheorem{thm}{Theorem}[section]
\newtheorem{lem}[thm]{Lemma}

\newtheorem{rmk}{\bf Remark}[section]

\theoremstyle{definition}

\numberwithin{equation}{section}
\usepackage[colorlinks=true,pdfstartview=FitV,linkcolor=magenta,citecolor=cyan]{hyperref}
\usepackage{bm}
\begin{document}
\title[Long-range Bernoulli model]{On localization for  the alloy-type Anderson-Bernoulli model with long-range hopping}

\author[Liu]{Shihe Liu}
\address[S. Liu] {School of Mathematical Sciences,
Peking University,
Beijing 100871,
China}
\email{2301110021@stu.pku.edu.cn}

\author[Shi]{Yunfeng Shi}
\address[Y. Shi] {School of Mathematics,
Sichuan University,
Chengdu 610064,
China}
\email{yunfengshi@scu.edu.cn}

\author[Zhang]{Zhifei Zhang}
\address[Z. Zhang] {School of Mathematical Sciences,
Peking University,
Beijing 100871,
China}
\email{zfzhang@math.pku.edu.cn}

\begin{abstract}
In this paper, we prove the Anderson localization near the spectral edge  for some {\it alloy-type}  Anderson-Bernoulli model  on $\Z^d$ with exponential long-range hopping. This extends the  work of Bourgain [{\it Geometric Aspects of Functional Analysis}, LNM 1850: 77--99,
2004], in which  he pioneered a novel  multi-scale analysis to treat  Bernoulli random variables. Our proof  is mainly based on  Bourgain's method. However,  to establish the initial scales Green's function estimates,  we  adapt  the approach   of Klopp [{\it Comm. Math. Phys}, Vol. 232, 125--155, 2002],  which is based  on the Floquet-Bloch theory and a certain quantitative uncertainty principle.   Our proof also applies to an analogues model on $\R^d.$
\end{abstract}

\maketitle

\tableofcontents

\section{Introduction and main results}
The study of Anderson localization (i.e., pure point spectrum with exponentially decaying eigenfunctions)   for Schr\"odinger operator on $\Z^d$ with  i.i.d. random potentials  has attracted great attention over the years.   For multi-dimensional operators with {\it continuous}   (e.g., H\"older continuous or even absolutely continuous) random potentials, localization can be established via either multi-scale analysis (MSA) method \cite{FS83} or fractional moment method (FMM) \cite{AM93}, cf. \cite{Kir08, AW15} for more results. Nevertheless,  if one  tries to prove localization for random Schr\"odinger operators with {\it singular}  potentials,  such as the Bernoulli one, known as the Anderson-Bernoulli model,  there comes essential difficulty: the absence of a priori  Wegner estimate required in the traditional MSA scheme \footnote{The proof of localization via  FMM  even requires the  absolute continuity of the random potential.}. In the one dimensional Anderson-Bernoulli models case,  this difficulty can be resolved via the transfer matrix type method, such as the Furstenberg-LePage approach \cite{CKM87},  cf. e.g., \cite{SVW98, DSS02, JZ19} for more results in this case.   However, the  transfer matrix formalism may not be available in higher dimensions, and the proof of localization for multi-dimensional Anderson-Bernoulli models  becomes significantly challenging.

Actually, Bourgain  \cite{Bou04} made the first important contribution toward the proof of   Anderson  localization for multi-dimensional Anderson-Bernoulli type models. He considered certain  {\it alloy-type} Anderson-Bernoulli model  on $\Z^d$, i.e., the single-site  random potential is $$D_{n}(\varepsilon)=\sum\limits_{m} 2^{-|m-n|}\varepsilon_m, \ \varepsilon =\{\varepsilon_m\}\in\{\pm 1\}^{\Z^d}$$  and established localization near the spectral edge.  In this remarkable work, Bourgain developed a novel MSA scheme,  which   used    the {\it free sites argument}  together with Boolean functions  analysis  to  obtain  the  Wegner estimate along  the iterations.  In fact, the variable $\varepsilon_n$ on  {\it free sites} can be made  continuous  without affecting the Green’s function estimates at that scale, which  allows   the application of the first-order eigenvalue  variation.  In addition,  for  the proof of the Wegner estimate, the non-vanishing correction coefficient  of  $2^{-|n|}$ ($n\in\Z^d$)  provides the {\it transversality}  condition  required  in the probabilistic Lojasiewicz inequality (cf. Lemma \ref{dislem}),  and  plays an essential role there. Later, the method of \cite{Bou04} was largely extended by Bourgain-Kenig \cite{BK05}, in which they  achieved the breakthrough and proved Anderson localization near the spectral  edge  for the standard Anderson-Bernoulli model on $\R^d$ via,  particularly,  introducing a refined version of the unique continuation principle, cf. \cite{AGKW, GK13} for  more results. Since the work  \cite{BK05} does  not dispose of a discrete version of unique continuation principle, the case of the version of the Anderson-Bernoulli model on $\Z^d$ (for $d\geq 2$) remains unsettled until the recent important work of  Ding-Smart \cite{DS20}. In \cite{DS20}, the authors  proved the Anderson localization near the spectral edge for the standard Anderson-Bernoulli model (i.e., $2^{-|m-n|}\to \delta_{mn}$ in  $D_n(\varepsilon)$) on $\Z^2$, and among others, they established  a new probabilistic version of  unique continuation result related to Buhovsky et al. \cite{BLMS}.  Recently, Li-Zhang \cite{LZ22}  extended the work of \cite{DS20} to the $\Z^3$ case, and to the best of our knowledge,  the problem   remains open for  $\Z^d, d\geq 4$.  We also mention the work  of Imbrie \cite{Imb21}, where the localization has been proved for random Schr\"odinger operators on $\Z^d$ with single-site potential having a discrete distribution taking $N$ values, with $N$ large. 
 
In this paper, we aim to extend the work of Bourgain \cite{Bou04}  to the exponential long-range hopping case and prove  Anderson localization near the spectral edge.  Indeed,  there has been a lot of research  on localization  for  operators on  $\Z^d$ with  long-range hopping and {\it continuous}  random potentials, cf. e.g.,  \cite{SS89, Wan91, AM93, Kle93, Gri94, JM99, Klopp02, Shi21, SWY25}. As the localization phenomenon  is  universe,  it is reasonable to expect that it should occur  for Bernoulli type potentials, which is one of our main motivations.   The main scheme of our proof is definitely adapted from Bourgain \cite{Bou04}.  However,  to handle long-range hopping in the  initial scales Green's function estimates,   we use  some  ideas  of Klopp \cite{Klopp02} based on the Floquet-Bloch theory and a certain quantitative uncertainty principle.   Besides,  we give both a clarification and streamlining of Bourgain's approach \cite{Bou04} via  several important technical improvements.   Our proof also applies  to an analogues model on $\R^d.$  We want to remark  that,   since the  transfer matrix formalism is currently only available for operators with a finite-range (i.e., Laplacian)  hopping,  our localization result is even new in the one-dimension case. 

\subsection{Main results}
Denote   $|n|_1=\sum\limits_{i=1}^d|n_i|$ and  $|n|=\max\limits_{1\leq i\leq d}|n_i|$  for $n=(n_1,\cdots,n_d)\in\Z^d$. 

 In this paper,  we   study the  long-range model   on $\Z^d$ 
\[H(\varepsilon)=T+D(\varepsilon),\]
where  $D(\varepsilon)={\rm diag}\{D_n(\varepsilon)\}$  is the {\it alloy-type}  Bernoulli potential, with the single-site elements given by   
\[D(\varepsilon)_n=\lambda \sum_{m\in \Z^d} 2^{-|n-m|} \varepsilon_m,\ \lambda>0.\]
Here we assume that $\varepsilon=\{\varepsilon_n\}_{n\in \Z^d}\in \{\pm 1\}^{\Z^d}$ are the i.i.d. Bernoulli random variables obeying  
\[\mathbb{P}(\varepsilon_n=\pm 1)=\frac{1}{2}.\]
 The long-range hopping  $T$  is   a Toeplitz   operator   with  $T(n,n')=T(n-n')$.  Under the  Fourier transform
\begin{equation}\label{Fourier transform lattice}
  \mathcal{F}:\ell^2(\Z^d)\longrightarrow L^2(\T^d),\ (\mathcal{F}a)(x)=\sum_{n\in \Z^d}a_n e^{2\pi i n\cdot x}, 
\end{equation}
$T$ has  the  symbol of 
\[\widehat{T}(x)=\sum_{n\in\Z^d}T(n)e^{2\pi i n\cdot x}, \ x\in \T^d.\]
Throughout this paper, we assume $T$ satisfies   the following assumptions:
\begin{itemize}
  \item[\textbf{(A1)}] $\widehat{T}(x)$ is bounded and real-valued, so that $T$ is bounded and self-adjoint; 
  \item[\textbf{(A2)}] $|T(n)|\leq e^{-c|n|}$ for some $c>0$;
  \item[\textbf{(A3)}] Denote  \begin{align}\label{Mdefn}
M=\max_{x\in \T^d} \widehat{T}(x).
\end{align}
We assume  that 
\[\widehat{T}^{-1}(\{M\})=\{\theta_1,\theta_2,\cdots,\theta_J\} \subset \T^d, \]
and there exists  some constant $\Theta>0$ such that  for all $x\in\T^d,$
\[M-\whT(x) \geq \Theta \min_{1\leq j\leq J}\|x-\theta_j\|_{\T^d}^2.\]
\end{itemize}
\begin{rmk}
Since we study localization near the spectral edges (related to the Lifshitz tails argument), the non-degeneracy  assumption  {\bf (A3)} is reasonable. Indeed, this assumption was  needed  even in the continuous random potentials case, cf. e.g.,  \cite{Klopp98,Klopp02,GRM22}.
\end{rmk}

Clearly, the discrete Laplacian $T=\Delta$ defined by  $\Delta(n,n')=\delta_{|n-n'|_1,1}$ satisfies all the above assumptions.  

It is easy to see that for a.e. $\varepsilon,$
 \[ \sup \sigma(H(\varepsilon))=M+\lambda \sum_{m\in\Z^d}2^{-|m|}:= E^*,\] 
 where $\sigma(\cdot)$ denotes the spectrum of an operator. 
In \cite{Bou04}, Bourgain proved that for  a.e. $\varepsilon$,  $H_0(\varepsilon)=\Delta+D(\varepsilon)$ exhibits  Anderson localization near the spectral edge.   Our main result below is an extension of the work \cite{Bou04} to the long-range hopping setting. More precisely, we have 
\begin{thm}\label{Main}
Under the assumptions of  \textbf{(A1)$\sim$(A3)}, for any $\lambda >0$, there exists  a small $\delta=\delta (\lambda, T, d)>0$ such that, for a.e. $\varepsilon$, $H(\varepsilon)$ exhibits Anderson localization on  $[E^*-\delta,E^*]$.
\end{thm}
\begin{rmk}
We also prove a similar localization result  for an analogues model on $\R^d$ (cf. Appendix \ref{conmod}). 
\end{rmk}

The proof of  Theorem \ref{Main} is based on  the  following multi-scale analysis type  Green's function estimates.  So denote 
\[G_N(E;\varepsilon)=\left(R_{\Lambda_N} (H(\varepsilon)-E+io)R_{\Lambda_N}\right)^{-1},\]
where  $R_{\Lambda}$ is  the restriction  operator  on  $\Lambda\in \Z^d$ and  $$\Lambda_N=[-N,N]^d,\ \Lambda_N(k)=\Lambda_N+k\ {\rm for}\ k\in \Z^d.$$  For simplicity,  we write  $R_N=R_{\Lambda_N}$.

Denote by $\| \cdot\|$ the operator norm. We have 
\begin{thm}\label{Green function estimates}
Under the assumptions of Theorem \ref{Main},  there exist  some constant $\gamma>0,c>0$ such that, for each $E\in [E^*-\delta,E^*]$ and  $N\gg1$, there is a set $\Omega_N(E)\subset\{\pm 1\}^{\Z^d}$ satisfying 
\begin{equation}\label{N scale bad event prob}
  \mathbb{P}(\Omega_N(E))<e^{-c\frac{(\log N)^2}{\log \log N}},
\end{equation}
so that the following holds true. For  $\varepsilon\notin\Omega_N(E),$
we have
\begin{align}
 \label{Green L2 norm} \| G_N(E;\varepsilon)\|&<e^{N^{\frac{9}{10}}},\\
\label{Green off-diagonal decay}
  |G_N(E;\varepsilon)(n,n')|&<e^{-\gamma|n-n'|} \ {\rm for} \ |n-n'|>\frac{N}{10}. 
\end{align}
\end{thm}
 \begin{rmk}
 As we will see in the proof,  this theorem holds for $N\geq N_0\gg1$ with 
 $ \delta\sim  (\log N_0)^{-10^3}.$ We also have 
 $$\gamma\sim (\log N_0)^{-2\times 10^3}.$$
 So  $\gamma\sim \delta^2\to 0$ as $\delta\to 0.$
 \end{rmk}

Theorem \ref{Green function estimates} will be proven inductively on the scale $N$. In this process, as in \cite{Bou04}, additional Green's function estimates involving {\it continuous variables}  are required. Define  $H(r)$ to be  the continuous  extension of $H(\varepsilon)$ from $\varepsilon\in \{\pm 1\}^{\Z^d}$ to $r\in [-1,1]^{\Z^d}$. 
We have 
\begin{thm}\label{Green function estimates, continuous version}
For each  $E\in [E^*-\delta,E^*]$ and  $N\gg1$, the extended Green's function 
\begin{equation*}
  G'_{N}(E;t,\varepsilon):= G_N(E;r_0=t,r_j=\varepsilon_j (j\neq 0))
\end{equation*}
satisfies  \eqref{Green L2 norm} and \eqref{Green off-diagonal decay} for  any $ t\in [-1,1]$ and $\varepsilon$ outside a  set  $\Omega'_N(E)\subset \{\pm 1\}^{\Z^d\setminus\{0\}}$ with  
\begin{equation}\label{N scale bad event prob, continue version}
  \mathbb{P}(\Omega'_N(E))<\frac{1}{100}. 
\end{equation}
\end{thm}
\begin{rmk}
This theorem is designed to  perform the {\it free sites argument}, and also will be proven  via the induction on scales $N$. Note that in  the estimate of  \eqref{N scale bad event prob, continue version}, there is no ``propagation of smallness'' phenomenon. 
\end{rmk}

\subsection{New ingredients  of the proof}

The main scheme of our proof  is  taken  from Bourgain's work \cite{Bou04}, which developed a novel MSA scheme to establish Green's function estimates and  thus localization.  Compared with \cite{Bou04}, our main new ingredients can be summarized as follows. 

\begin{itemize}
\item To control  the probability  that the Green's function has ``good''  estimates at the initial scales, Bourgain \cite{Bou04}  employed  the fact that the free Laplacian can be decomposed into a combination of shifts in each direction, namely, 
\[\Delta =\sum_{1\leq i\leq d} (S_{\delta_i}+S_{-\delta_i}),\ S_n f(m)=f(m-n) \ {\rm for}\  \forall f\in \ell^2(\Z^d).\]
 The key point  is that in the regime of ``bad''  Green's function estimates,  one can find a vector  (nearly) belonging  to the eigenspace of the edge  spectrum of  both $\Delta$ and  $D(\varepsilon)$. The  character of the  eigenspace of $\Delta$ ensures that 
\begin{equation}\label{shift invariant}
  \| \xi -S_{\delta_i}\xi\|\ll 1,
\end{equation}
namely, $\xi$ is nearly  shift-invariant. This  forces  the Bernoulli potential to have some correlation  property, and thus leads to  the desired probability estimate. 

However, in the long-range hopping model, \eqref{shift invariant} has to  be replaced by a more essential fact:  the support of Fourier transform of $\xi$ nearly concentrates on  the maxima of $\whT(x)$. Unfortunately, the {\it non-uniqueness}  of the maxima of $\whT(x)$  poses the key challenge: it is hard to find such a shift-invariant property  of the form 
\[\| \xi-S_n \xi  \|\ll 1.\]
To resolve this difficulty, inspired by the idea of  Klopp   \cite{Klopp98, Klopp02}, we adapt a quantitative version of  the uncertainty principle as a replacement of \eqref{shift invariant}, together with  the Floquet-Bloch theory on some periodic approximation of  the restricted  random  operator (cf. Section \ref{section 2}).  In the continuum model, since  one can obtain the approximation of the identity via the dilation of a bump function, the concentration of $\whx(x)$ can be directly transferred  to the potential  correction via  some geometric projections (cf. Appendix \ref{conmod}).  We mention that in the final probability estimate, we  introduce the Dudley's estimate on sub-Gaussian random variables. 

\item Even for the large scales, we provide more detailed analysis of Bourgain's approach, such as the trim arguments on $\Omega_N(E)$ (cf. Subection \ref{trim}), a key coupling lemma (cf. Lemma \ref{MSA}) involving  the iteration of resolvent identities, and the energy-free Green's function estimates (cf. Section \ref{elimation}). 

\end{itemize}

\subsection{Sructure of the paper and the notation}
The paper is  organized as follows. In \S \ref{section 2}, we prove the Green's function estimates at  initial scales. In \S \ref{LGFS}, we use multi-scale analysis to establish the Green's function estimates for all scales, thus complete the proof of Theorems \ref{Green function estimates}, \ref{Green function estimates, continuous version}. In \S \ref{elimation}, we prove our main theorem (i.e., Theorem \ref{Main}) on the localization. In Appendix \ref{conmod},  we discuss an analogues model on $\R^d$ following  \cite{Bou04}. 
\begin{itemize}
\item In this paper, the notation $B=\mcO(A)$ means that $c_1 A\leq B\leq c_2 A$ ($A>0,B>0$) for some absolute  positive  constants  $c_1,c_2$. 
\item We denote $A\lesssim B$ if there is a constant $c>0$ independent of $A, B$ so that $A\leq cB$. By $A\sim B$, we mean $A\lesssim B$ and $B\lesssim A.$  We denote $A\ll B$   if $A$ is much smaller than $B$.
\item We denote by  $\lfloor \cdot\rfloor$   the integer part of a real number. 
\item We let $(\cdot)^c$ be the complement of a set.
\item The notation $\langle\cdot\rangle$ represents the standard inner product on $\Z^d, \R^d, \ell^2, L^2$ etc.  In general,  $\|\cdot\|$ denotes the norm induced by this inner product, or the operator norm (for an operator). 
\end{itemize}

\section{Green's function estimates at the initial scales}\label{section 2}
In this section, we will establish large deviation estimates for Green's function estimates at the initial scales. Our approach 
uses  the Floquet-Bloch theory and quantitative uncertainty principle in the spirit of Lifschitz tails estimates of Klopp \cite{Klopp02},  which  allow us to transfer estimates of Green's functions to those of random variables.  As a result, we can obtain  the probabilistic  estimates.

First, we consider the more general potential 
\[D(\varepsilon)_n=\lambda \sum_{m\in\Z^d}A_{n-m}\varepsilon_m,\ A_m\geq 0,\ A_0 >0 , \ E^*=M+\lambda\sum_{m\in\Z^d} A_m <\infty \]
(This class  definitely contains  the standard Bernoulli potential with  $A_m=\delta_{m,0}$,  and  our  model with   $A_m=2^{-|m|}$).
 
Next,  for fixed energy $E\in [E^*-\delta,E^*]$ ($\delta>0$),  we decompose  
\[H(\varepsilon)-E=(T+1)-(E+1-D(\varepsilon)).\]
Denote  by  $H_{N_0},T_{N_0},$ and $D_{N_0}$ the corresponding restrictions  on $\Lambda_{N_0}, N_0>0$. Recalling \eqref{Mdefn} and denoting
\[m=\min_{x\in \T^d}\whT(x).\]
Without loss of generality, we can assume $M+1>|m+1|$, otherwise we take a dilation of $T$ by $\kappa T$ with $0<\kappa \ll 1$. Observe that  
\begin{align*}
\| T+1\| &=\| \whT +1 \|_{L^{\infty}(\T^d)}=M+1,\\
E^*+1+\lambda\sum_{m\in \Z^d}A_m&\geq E+1-D(\varepsilon)_n\geq E^*+1-D(\varepsilon)_n-\delta\\
&>M+1-\delta \ {\rm for}\ \forall n\in \Z^d.
\end{align*}
So 
\begin{equation}\label{potential norm}
  \| (E+1-D(\varepsilon))^{-1}\| <\frac{1}{M+1-\delta}. 
\end{equation}
By a Neumann expansion argument,   we get 
\begin{align}\label{initial Neumann}
  G_{N_0}(E)=(H_{N_0}-E)^{-1}=(D_{N_0}(\varepsilon)-E-1)^{-1}\sum_{s\geq 0}(-1)^s ((T_{N_0}+1)(D_{N_0}(\varepsilon)-E-1)^{-1})^s.
\end{align}

We have the following main theorem of this section. 

\begin{thm}\label{initial scale}
For $N_0\geq C(d,\lambda,M,J,\Theta,|T(0)|)\gg1$ and $\delta=(\log N_0)^{-10^3}$, there is  some  $\Omega_{N_0}\subset \{\pm 1\}^{\Z^d}$ \text{independent of $E$} such that
\begin{equation*}
  \mathbb{P}(\Omega_{N_0})\leq e^{-(\log N_0)^3 }, 
\end{equation*}
and  for all $\varepsilon\notin \Omega_{N_0}$,  $E\in [E^*-\delta,E^*]$,  both \eqref{Green L2 norm} and \eqref{Green off-diagonal decay}  hold  true with  $N=N_0$ and  $\gamma=\gamma_0= \frac{1}{(\log N_0)^{2\times 10^3}}$. 
\end{thm}
\begin{rmk}
As we will see below, $N_0$ is not well defined in the whole interval  of length  $\sim e^{{\delta^{-10^{-3}}}}$,  but  only in  a ``dense''  subset (cf. \eqref{parameter condition 1}), which suffices for the MSA induction. 
\end{rmk}

Recalling \eqref{initial Neumann},  we assume
\begin{align}\label{convass}
  \| (T_{N_0}+1)(E+1-D_{N_0})^{-1}  (T_{N_0}+1)(E+1-D_{N_0})^{-1} \| >1-\delta. 
  \end{align}
  By \eqref{potential norm}, we have
  \begin{equation}\label{eq1}
      \| (T_{N_0}+1)(E+1-D_{N_0})^{-1}  (T_{N_0}+1)\| > (1-\delta)(M+1-\delta). 
  \end{equation} 
  Hence by the self-adjointness, there exists  $\xi \in \ell^2(\Z^d)$ such that $\supp (\xi) \subset \Lambda_{N_0},\|\xi \|_2=1$ (here ${\rm supp}(\cdot)$ denotes the support) and 
  \begin{align}\label{eq2}
    \frac{1}{M+1-\delta} \| (T_{N_0}+1)\xi\|^2 &\geq \langle \xi,(T_{N_0}+1)(E+1-D_{N_0}(\varepsilon))^{-1} (T_{N_0}+1)\xi\rangle \\
   \notag        &>(1-\delta)(M+1-\delta),
  \end{align}
  which shows 
  \begin{align*}
    \|(T+1)\xi\|^2 &\geq  \| (T_{N_0}+1)\xi \|^2 >(1-\delta)(M+1-\delta)^2 \\
                   &\geq (M+1)^2 -\mcO( \delta). 
  \end{align*}
  Now let $\delta\ll \eta\ll 1$ ($\eta$ will  be determined below). Then 
  \begin{align*}
    \|(T+1)\xi \|^2 &= \| (\whT +1)\cdot \widehat{\xi} \|_{L^2(\T^d)} \\
      &=\left(\int_{\{\whT > M-\eta\}}+\int_{\{\whT \leq  M-\eta\}} \right) (|\whT(x)+1|^2\cdot |\whx(x)|^2) dx \\
      &\leq (M+1)^2 \int_{\{\whT > M-\eta\}} |\whx(x)|^2 dx + (M+1-\eta)^2 \int_{\{\whT \leq  M-\eta\}} |\whx(x)|^2 dx \\
      &=(M+1)^2 -[(M+1)^2-(M+1-\eta)^2]\int_{\{\whT \leq  M-\eta\}} |\whx(x)|^2 dx. 
  \end{align*}
  This gives  
  \begin{align}\label{concentration of whx}
    \int_{\{\whT \leq  M-\eta\}} |\whx(x)|^2 dx & \leq \frac{\mcO(\delta)}{(M+1)^2-(M+1-\eta)^2} =\mcO(\frac{\delta}{\eta}). 
  \end{align}
  We emphasis that \eqref{concentration of whx} reveals the fact  that the Fourier transformation of $\xi$  concentrates  near the maxima of the symbol $\whT(x)$.  Applying   \eqref{concentration of whx}  implies 
  \begin{align*}
    \| (M+1)\xi -(T_{N_0}+1)\xi \|^2 &\leq \| (M+1)\xi -(T+1)\xi \|^2 \\
         &=\int_{\T^d}|\whT(x)-M|^2\cdot |\whx(x)|^2 dx\\
         &=\left(\int_{\{\whT > M-\eta\}}+\int_{\{\whT \leq  M-\eta\}} \right) (|\whT(x)-M|^2\cdot |\whx(x)|^2) dx\\
         &\leq \eta^2 \int_{\{\whT > M-\eta\}} |\whx(x)|^2 dx +(m+M)^2 \int_{\{\whT \leq M-\eta\}} |\whx(x)|^2 dx \\
         & \leq \eta^2 +\mcO(\frac{\delta}{\eta}). 
  \end{align*}
  Concerning the optimality,   we can choose $\eta\sim \delta^{\frac{1}{3}}$ and hence 
  \begin{equation}\label{eq3}
    \| (M+1)\xi -(T_{N_0}+1)\xi \|\leq \| (M+1)\xi -(T+1)\xi \| =\mcO(\delta^{\frac{1}{3}}). 
  \end{equation}
  Combining \eqref{eq3} and  \eqref{eq2} yields  
  \begin{align}\label{eq4}
    (M+1)^2 \langle (E+1-D(\varepsilon))\xi,\xi\rangle &\geq  \langle \xi,(T_{N_0}+1)(E+1-D_{N_0}(\varepsilon))^{-1} (T_{N_0}+1)\xi\rangle  \\
   \notag            &\ \ \  -\mcO(\| (M+1)\xi -(T_{N_0}+1)\xi \|) \\
         \notag   & >(1-\delta)(M+1-\delta) -\mcO(\delta^{\frac{1}{3}})\\
       \notag     &\geq M+1-\mcO(\delta^{\frac{1}{3}}). 
  \end{align}
  Then estimate \eqref{eq4} gives  
  \begin{align}\label{eq5}
    \| (E+1-D(\varepsilon))^{-1} (M+1)\xi -\xi\|^2 & =\| (E+1-D(\varepsilon))^{-1}(M+1)\xi  \|^2+\|\xi\|^2 \\
         \notag & \ \ \ -2\langle (E+1-D(\varepsilon))^{-1} (M+1)\xi,\xi\rangle  \\
            \notag & \leq \frac{M+1}{M+1-\delta}+1-2 \left( 1 -\frac{\mcO(\delta^\frac{1}{3})}{M+1} \right) \\
            \notag & =  \mcO(\delta^{\frac{1}{3}}), 
  \end{align}
  and 
  \begin{align}\label{eq6}
    \|(M+1)\xi-(E+1-D(\varepsilon))\xi\| & \leq \| (E+1-D(\varepsilon)) \| \cdot \| (E+1-D(\varepsilon))^{-1} (M+1)\xi -\xi\| \\
                     \notag       & \leq C  (\delta^{\frac{1}{3}} )^{\frac{1}{2}} =\mcO(\delta^{\frac{1}{6}}). 
  \end{align}
  Combining  \eqref{eq3} and \eqref{eq6} implies  
  \begin{equation}\label{nearly eigenfunction}
    \| (H(\varepsilon)-E)\xi \| =\| (T+1)\xi -(E+1-D(\varepsilon))\xi \| = \mcO(\delta^{\frac{1}{6}}). 
  \end{equation}

At this stage, we have shown that \eqref{convass} leads to the existence of approximate eigenvalue  $E$ of $H(\varepsilon).$ In the following, we will use   Floquet-Bloch theory to extract further information from $E, \xi.$
 
  \subsection{The periodic approximation and Floquet-Bloch decomposition}
  Recall that  $\xi$ (resp. $E$) is the approximate eigenfunction (resp. eigenvalue) and  $\supp(\xi)\subset \Lambda_{N_0}$.  
  This  indicates that  the periodic approximation of  $H(\varepsilon)$   (with potential restricted to $\Lambda_{N_0}$) has an eigenvalue close to $E$. From this perspective, the  periodic approximation technique and the Floquet-Bloch theory  (as in \cite{Klopp98,Klopp02}) can be applied here.   

  Some basic facts about Floquet-Bloch theory can be found  in  the Appendix \ref{Floquet-B}. To avoid  confusion, it is important to note that for an operator,   we use   $(\cdot)_{N}$ to represent  its  finite volume restriction and  $(\cdot)^N$ to represent its periodic extension. 

  Now  we define the  periodic extension of $D_{N_0}(\varepsilon)$ on $\Lambda_{N_0}$  to be  $\widetilde{D}^{N_0}(\varepsilon)$  with 
  \[\widetilde{D}^{N_0}(\varepsilon)_n=D(\varepsilon)_{n'}\]
for some  $n'\in \Lambda_{N_0}$ and  $n-n'\in [(2N_0+1)\Z]^d$. 
We denote further  
  \[H^{N_0}(\varepsilon)=T+\widetilde{D}^{N_0}(\varepsilon),\]
 and making $H^{N_0}(\varepsilon)$ a periodic Schr\"odinger operator.  From  
  \[\supp(\xi)\subset \Lambda_{N_0},\ \supp(D(\varepsilon)-\widetilde{D}^{N_0}(\varepsilon))\subset \Z^d\setminus \Lambda_{N_0},\]
it follows that  $H(\varepsilon)\xi=H^{N_0}(\varepsilon) \xi $. Thus, applying \eqref{nearly eigenfunction} yields 
  \[\|(H^{N_0}(\varepsilon)-E)\xi \|=\mcO(\delta^{\frac16}),\]
  and hence  $[E-\mcO(\delta^{\frac16}),E+\mcO(\delta^{\frac16})]$ contains some spectrum of $H^{N_0}(\varepsilon)$.

Since $\| \widetilde{D}^{N_0}(\varepsilon)\| \leq \|D(\varepsilon)\|\leq \lambda\sum_{m\in\Z^d}A_m$, we get  
  \[\|H^{N_0}(\varepsilon)\|\leq E^* .\]
  This together with $E\in [E^*-\delta,E^*]$ shows   
  \begin{equation}\label{periodic spectrum}
    \sigma(H^{N_0}(\varepsilon))\cap [E^*-\delta-\mcO(\delta^{\frac{1}{6}}),E^*]\neq \emptyset,
  \end{equation}
  where $\sigma(\cdot)$ denotes the spectrum.  So,
   \begin{equation}\label{positive periodic spectrum}
    \sigma(E^*-H^{N_0}(\varepsilon))\cap [0,\mcO(\delta^{\frac{1}{6}})]\neq \emptyset. 
  \end{equation}

  In view of \eqref{positive periodic spectrum}, we have  
  \[\widetilde{H}^{N_0}(\varepsilon)=E^*-H^{N_0}(\varepsilon) =(M-T)+(\lambda \sum_{m\in \Z^d} A_m -\widetilde{D}^{N_0} (\varepsilon)).\]
  Denote $h(x)=M-\whT(x)$ and $h_k(x)=(U\mathcal F_{N_0}^{-1}h)_k(x)$ (cf. \eqref{Floquet transform}). By \eqref{Floquet matrix}, the fiber matrix  at the Floquet quasi-momentum $x\in (\frac{\T}{2N_0+1})^d$ is given by 
  \begin{equation}
    M^{N_0}_{\varepsilon}(x)=(h_{k-j}(x))_{\Lambda_{N_0}\times\Lambda_{N_0}}+(\lambda\sum_{m\in \Z^d}A_m-D_{N_0}(\varepsilon)) :=P+ (\lambda\sum_{m\in \Z^d}A_m-D_{N_0}(\varepsilon)).
  \end{equation}
  Combining  \eqref{positive periodic spectrum} and \eqref{Floquet spectrum} implies  immediately  that there is some  $x$ such that 
  \begin{equation}\label{degenerate small spectrum}
    \sigma( M^{N_0}_{\varepsilon}(x))\cap [0,\mcO(\delta^{\frac{1}{6}})]\neq \emptyset. 
  \end{equation}
  $\qquad$\\
  For this  $x$ satisfying \eqref{degenerate small spectrum}, there exists some vector $a=a(x)\in \ell^2(\Lambda_{N_0}),\|a\|_{\ell^2(\Lambda_{N_0})}=1$ such that 
  \begin{align*}
    0 & \leq \langle a, M^{N_0}_{\varepsilon}(x) a\rangle_{\ell^2(\Lambda_{N_0})}= \langle a, P a\rangle_{\ell^2(\Lambda_{N_0})} +  \langle a, (\lambda\sum_{m\in \Z^d}A_m-D_{N_0}(\varepsilon)) a\rangle_{\ell^2(\Lambda_{N_0})}= \mcO(\delta^{\frac{1}{6}}). 
  \end{align*}
  
 Since  $h(x)\geq 0$ and  \eqref{Floquet basis}, we know that $(h_{k-j}(x))_{\Lambda_{N_0}\times \Lambda_{N_0}}$ is a  positive operator.   Obviously, $(\lambda\sum_{m\in \Z^d}A_m-D_{N_0}(\varepsilon))$ is also positive. Hence,
    \begin{align}
       \label{positive 1}   & 0\leq \langle a, P a\rangle_{\ell^2(\Lambda_{N_0})}  = \mcO(\delta^{\frac16}), \\
        \label{positive 2}  & 0\leq \langle a, (\lambda\sum_{m\in \Z^d}A_m-D_{N_0}(\varepsilon)) a\rangle_{\ell^2(\Lambda_{N_0})}  = \mcO(\delta^{\frac16}).
    \end{align}
  Form now on, we always use  the orthonormal Floquet basis (cf. \eqref{Floquet basis}) $\{\beta_s(x)\}_{s\in \Lambda_{N_0}}$ to represent vectors in $\ell^2(\Lambda_{N_0})$. For example, the coordinate representation of $a$ is
  \[a=\sum_{k\in \Lambda_{N_0}}a_k \beta_k(x) = (a_k)_{k\in \Lambda_{N_0}},\]
  and for the modified canonical basis $\{v_l\}_{l\in \Lambda_{N_0}}$ in \eqref{modified canonical basis}, 
  \[v_l=\sum_{k\in \Lambda_{N_0}} \frac{1}{(2N_0+1)^{\frac{d}{2}}}e^{2\pi i \frac{k}{2N_0+1}\cdot l}\cdot \beta_k(x) =\left(\frac{1}{(2N_0+1)^{\frac{d}{2}}}e^{2\pi i \frac{k}{2N_0+1}\cdot l} \right)_{k\in\Lambda_{N_0}}.\]
  \\
  $\qquad$\\

 Now,  take $L'<L,K<K'$ to be some positive integers  much smaller than  $N_0$ (will be determined below)  such that 
  \begin{equation}\label{parameter condition 1}
    (2N_0+1)=(2L+1)(2K+1)=(2L'+1)(2K'+1)
  \end{equation}
 (This leads to  that $N_0$ cannot be arbitrarily chosen in an interval as remarked previously). 
  Recall that the assumption \textbf{(A3)} ensures that 
  \[h^{-1}(\{0\})=\{\theta_1,\cdots,\theta_J\}\subset \T^d ,\ h(x)\geq \Theta \min_{1\leq j\leq J}|x-\theta_j|^2.\]
  For $1\leq j\leq J$, let $k_j\in \Z^d$ be  the integer part of $(2N_0+1)\theta_j$, namely, 
  \[k_j-(2N_0+1)\theta_j\in [-\frac{1}{2},\frac{1}{2}). \]
  Under the basis $\{\beta_k\}_{k\in \Lambda_{N_0}}$, define 
  \begin{equation}
    (a^j)_k=\left\{ 
      \begin{aligned}
        &a_k,\ {\rm if} \ |k-k_j|\leq K,  \\
        &0,\ {\rm else}. 
      \end{aligned}\right.
  \end{equation} 
  Then $\supp(a^j)$ are disjoint. In fact,  for  $j\neq j'$, we obtain 
  \begin{align*}
      |k_k-k_{j'}|& \geq |(2N_0+1)(\theta_j-\theta_{j'})|-|k_j-(2N_0+1)\theta_j|-|k_{j'}-(2N_0+1)\theta_{j'}|\\ 
          &\geq \min_{1\leq j\neq j'\leq J}|\theta_j-\theta_{j'}|\cdot (2N_0+1)-1 \gtrsim N_0\gg K. 
  \end{align*}
  Now consider the vector 
  \begin{equation}
    a-\sum_{1\leq j\leq J}a^j,
  \end{equation}
  which is supported on 
  \begin{align*}
      \{k:\ |k-k_j|>K,\forall 1\leq j\leq J \} &\subset  \{k:\ |k-(2N_0+1)\theta_j|>K-\frac{1}{2},\forall 1\leq j\leq J \}\\  &\subset \{k:\ |\frac{k}{2N_0+1}-\theta_j|>\frac{K-\frac{1}{2}}{2N_0+1},\forall 1\leq j\leq J\} \\
         &\subset \{k:\ |\frac{k}{2N_0+1}-\theta_j|>\frac{1}{3(2L+1)},\forall 1\leq j\leq J\}\\
         &:=\mathcal K.
  \end{align*}
  As  $x\in [-\frac{1}{2(2N_0+1)},\frac{1}{2(2N_0+1)}]^d$ has been   fixed  to satisfying  \eqref{degenerate small spectrum},  we have  for any $k\in \mathcal K,$
  \[\min_{1\leq j\leq J }|x+\frac{k}{2N_0+1}-\theta_j|>\frac{1}{3(2L+1)}-\frac{1}{2(2N_0+1)} \geq \frac{1}{6(2L+1)},\]
  and by the assumption \textbf{(A3)}, 
  \begin{equation}\label{2.21}
    E_k(x)=h(x+\frac{k}{2N_0+1}) \geq \Theta\cdot \left(\frac{1}{6(L+1)}\right)^2 \ {\rm for}\  \forall k\in \mathcal K. 
  \end{equation}
  Since  the Floquet basis $\{\beta_k\}_{k\in \Lambda_{N_0}}$ diagonalizes $P$,  we have by combining  \eqref{positive 1} and \eqref{2.21}   that 
  \begin{align*}
      \mcO(\delta^{\frac16})&= \langle a, P a\rangle_{\ell^2(\Lambda_{N_0})}  \\
           &\geq \left\langle (a-\sum_{1\leq j\leq J}a^j) ,P  (a-\sum_{1\leq j\leq J}a^j)\right\rangle_{\ell^2(\Lambda_{N_0})} \\
           &\geq \min_{k\in \mathcal K} E_k(x)\cdot  \left\|  a-\sum_{1\leq j\leq J}a^j \right\|^2 \\
           &\geq \Theta \cdot\left(\frac{1}{6(2L+1)}\right)^2 \cdot  \left\|  a-\sum_{1\leq j\leq J}a^j \right\|^2,
  \end{align*} 
which implies 
  \begin{equation}\label{concentrate of a}
      \left\|  a-\sum_{1\leq j\leq J}a^j \right\| \leq \mcO(\delta^{\frac{1}{12}} L ). 
  \end{equation}
  The above argument   just reveals that $a$ concentrates near $\{k_j\}_{1\leq j\leq J}$. 
  
  Next, rewrite \eqref{positive 2} as 
  \begin{equation}\label{positive 2 re}
    \langle a,D_{N_0}(\varepsilon)a\rangle_{\ell^2(\Lambda_{N_0})} \geq \lambda\sum_{m\in \Z^d}A_m-\mcO(\delta^{\frac16}).
  \end{equation}
  Plugging \eqref{concentrate of a} into \eqref{positive 2 re} yields  
  \begin{align}
  \notag  \langle \sum_{1\leq j\leq J}a^j ,D_{N_0}(\varepsilon) \sum_{1\leq j\leq J}a^j \rangle_{\ell^2(\Lambda_{N_0})}&\geq \langle a,D_{N_0}(\varepsilon)a\rangle_{\ell^2(\Lambda_{N_0})} -\mcO\left(\left\| a-\sum_{1\leq j\leq J}a^j \right\|\right) \\
  \notag          &\geq \lambda \sum_{m\in \Z^d}A_m - \mcO(\delta^{\frac{1}{12}} L ) \\
   \label{2.24}           &\geq \frac{2}{3} \lambda \sum_{m\in \Z^d}A_m. 
  \end{align}
  The  inequality \eqref{2.24} requires  
  \begin{equation}\label{parameter condition 2}
    \delta^{\frac{1}{12}} L \leq c({\lambda,M,\Theta}) \ll1. 
  \end{equation}
  
  Finally, express  $a^j$  in  the modified canonical basis $\{v_l\}_{l\in \Lambda_{N_0}}$ in \eqref{modified canonical basis}, namely, 
  \[a^j=\sum_{l\in \Lambda_{N_0}} \langle v_l,a^j\rangle v_l.\]
  Then 
  \begin{align}\label{2.26}
   &\ \ \  \langle \sum_{1\leq j\leq J}a^j ,D_{N_0}(\varepsilon) \sum_{1\leq j\leq J}a^j \rangle _{\ell^2(\Lambda_{N_0})} \\
   \notag & =\sum_{j=1}^{J} \sum_{l\in \Lambda_{N_0}} D(\varepsilon)_l\cdot |\langle v_l,a^j\rangle|^2 
                                      +2\Re \left( \sum_{1\leq j< j'\leq J} \sum_{l\in \Lambda_{N_0}}D(\varepsilon)_l \cdot \langle v_l,a^j\rangle \cdot \overline{\langle v_l,a^{j'}\rangle}    \right). 
  \end{align}
  As $\supp(a^j)= \Lambda_K(k_j)$ (under the Floquet basis),  we can construct another  $\widetilde{a}^j$ by changing the center $k_j$ of the support  to   the origin, namely,  
  \begin{equation}\label{tilde a}
    \widetilde{a}^j=\sum_{k\in \Lambda_K}a_{k+k_j} \cdot \beta_k(x). 
  \end{equation}
  Direct computations  show  that 
  \begin{align*}
    \langle v_l,\widetilde{a}^j\rangle & =\sum_{k\in \Lambda_K} a_{k+k_j} \langle v_l,\beta_k\rangle \\
                          & =\sum_{k\in \Lambda_K} a_{k+k_j} \left( \frac{1}{(2N_0+1)^{\frac{d}{2}}} e^{-2\pi i\frac{k}{2N_0+1}\cdot l}\right) \\
                          & = e^{2\pi i \frac{k_j}{2N_0+1}\cdot l} \sum_{k\in \Lambda_K} a_{k+k_j} \left( \frac{1}{(2N_0+1)^{\frac{d}{2}}} e^{-2\pi i\frac{k+k_j}{2N_0+1}\cdot l}\right) \\
                          & = e^{2\pi i \frac{k_j}{2N_0+1} \cdot l} \sum_{k\in \Lambda_K} a_{k+k_j} \langle v_l,\beta_{k+k_j}\rangle  \\
                          & = e^{2\pi i \frac{k_j}{2N_0+1}\cdot l } \langle v_l,a^j\rangle. 
  \end{align*}
  Hence, we get 
  \begin{equation}\label{a and tilde a rotation}
    \langle v_l,a^j\rangle = e^{-2\pi i \frac{k_j}{2N_0+1}\cdot l } \langle v_l,\widetilde{a}^j\rangle. 
  \end{equation}
  Substituting  \eqref{a and tilde a rotation} in \eqref{2.26} and using   \eqref{2.24} imply 
    \begin{align}\label{2.29}
   \sum_{j=1}^{J} \sum_{l\in \Lambda_{N_0}} D(\varepsilon)_l\cdot |\langle v_l,\widetilde{a}^j\rangle|^2  & +2\Re \left( \sum_{1\leq j< j'\leq J} \sum_{l\in \Lambda_{N_0}} e^{-2\pi i \frac{k_j-k_{j'}}{2N_0+1}\cdot l}D(\varepsilon)_l \cdot \langle v_l,\widetilde{a}^j\rangle  \cdot \overline{\langle v_l,\wta^{j'}\rangle}    \right) \\
      \notag     &\geq \frac{2}{3}\lambda \sum_{m\in \Z^d}A_m. 
  \end{align}
  \begin{rmk}
   In \eqref{a and tilde a rotation} and \eqref{2.29}, shifting $k_j$ to the origin already reveals that  the existence of multiple maxima  causes  the  term  $e^{-2\pi i \frac{k_j-k_{j'}}{2N_0+1}\cdot l}$ in the summation. This phenomenon will also be observed in the continuum model.
  \end{rmk}

  \subsection{Application of quantitative uncertainty principle} What we have done  is, as in \eqref{2.29},  expressing  \eqref{degenerate small spectrum} in a different form related to some information about $\ell^2(\Lambda_{N_0})$-vectors supported near the origin. This would allow   us  to apply  the quantitative  uncertainty principle of  \cite{Klopp02} (cf. the Appendix \ref{UP} for details). 
  
  From \eqref{discrete Fourier transform} and identifying $\Lambda_{N_0}$ with  $\Z^d_{2N_0+1}$,  we have
  \[\langle v_l,\wta^j\rangle=(\mcF_{N_0} \wta^j)_l,\]
  where $\mcF_{N_0}$ denotes the discrete Fourier transformation on $\Z^d_{2N_0+1}.$
  Recall that we have  \eqref{parameter condition 1}.  From $\supp(\wta^j)\in \Lambda_K,1\leq j\leq J$ and  applying  Lemma \ref{discrete UP}, we get some  $b^j\in \ell^2(\Lambda_{N_0})$ such that
  \begin{enumerate}
    \item $\|\wta^j - b^j\|_{\ell^2(\Lambda_{N_0})} \leq C \|\wta^j\|_{\ell^2(\Lambda_{N_0})} =\mcO(K/K')$;

    \item For $l' \in \Lambda_{L'}$ and $k' \in \Lambda_{K'}$, we have $\langle v_{l'+k'(2L'+1)},b^j\rangle  = \langle v_{k'(2L'+1)},b^j\rangle$;

    \item $\|\wta^j\|_{\ell^2(\Lambda_{N_0})} = \|b^j \|_{\ell^2(\Lambda_{N_0})}$.
    
\end{enumerate}
Substituting  $\wta^j$ by $b^j$ in \eqref{2.29} implies  
  \begin{align}\label{2.30}
   \sum_{j=1}^{J} \sum_{l\in \Lambda_{N_0}} D(\varepsilon)_l\cdot |\langle v_l,b^j\rangle |^2  & +2\Re \left( \sum_{1\leq j< j'\leq J} \sum_{l\in \Lambda_{N_0}} e^{-2\pi i \frac{k_j-k_{j'}}{2N_0+1}\cdot l}D(\varepsilon)_l \cdot \langle v_l,b^j\rangle  \cdot \overline{\langle v_l,b^{j'}\rangle}    \right) \\
      \notag     &\geq \frac{2}{3}\lambda \sum_{m\in \Z^d}A_m-C(J)\cdot \| D(\varepsilon) \|_{\ell^{\infty}(\Z^d)} \cdot \|\wta^j - b^j\|_{\ell^2(\Lambda_{N_0})} \\
      \notag &\geq \frac{2}{3}\lambda \sum_{m\in \Z^d}A_m- \mcO(K/K')
  \end{align}
Moreover,  on the left hand side of \eqref{2.30}, writing  the summation index  uniquely as 
\[\Lambda_{N_0}\ni l=l'+k'(2L'+1)\in \Lambda_{N_0},l'\in \Lambda_{L'},k'\in \Lambda_{K'}\]
and applying  property (2) of $b^j,1\leq j\leq J$ yield 
\begin{align}\label{2.31}
  & \sum_{j=1}^{J} \sum_{k'\in \Lambda_{K'}}  \mcS(j,j,k') \cdot (2L'+1)^d |\langle v_{k'(2L'+1)},b^j\rangle |^2 \\
   \nonumber +&\sum_{1\leq j< j'\leq J} \sum_{k'\in \Lambda_{K'}} 2\Re \left( \mcS(j,j',k')  \cdot e^{-2\pi i \frac{k_j-k_{j'}}{2K'+1}\cdot k'} \langle v_{k'(2L'+1)},b^j\rangle \cdot \overline{\langle v_{k'(2L'+1)},b^{j'}\rangle } \right)  \\
      \notag &\geq \frac{2}{3}\lambda \sum_{m\in \Z^d}A_m- \mcO(K/K'),
  \end{align}
where 
\begin{equation}\label{mcS summation}
  \mcS(j,j',k')= \frac{1}{(2L'+1)^d}\sum_{l'\in \Lambda_{L'}} e^{-2\pi i \frac{k_j-k_{j'}}{2N_0+1}\cdot l'}D(\varepsilon)_{l'+k'(2L'+1)}. 
\end{equation}
In addition,  the left hand side (LHS) of \eqref{2.31} can be controlled: 
\begin{equation}\label{2.33}
  {\rm LHS \ of \ \eqref{2.31}}\leq \sum_{1\le j,j'\leq J} \sum_{k'\in \Lambda_{K'}} |\mcS(j,j',k')|(2L'+1)^d \langle v_{k'(2L'+1)},b^j\rangle \cdot \overline{\langle v_{k'(2L'+1)},b^{j'}\rangle}|. 
\end{equation}
Now,  recalling $|\frac{k_j}{2N_0+1}-\theta_j|\leq \frac{1}{2(2N_0+1)}$,  we can define
\begin{equation}\label{mcN summation}
  \mcN(j,j',k')=\frac{1}{(2L'+1)^d}\sum_{l'\in \Lambda_{L'}} e^{-2\pi i (\theta_j-\theta_{j'})\cdot l'}D(\varepsilon)_{l'+k'(2L'+1)}
\end{equation}
and thus,
\begin{align}\label{2.35}
  |\mcS(j,j',k')-\mcN(j,j',k')| & \leq \| D(\varepsilon) \|_{\ell^{\infty}(\Z^d)} \frac{1}{(2L'+1)^d} \sum_{l'\in \Lambda_{L'}}|l'|\cdot \left |\frac{k_j-k_{j'}}{2N_0+1}-(\theta_j-\theta_{j'})\right| \\
                  \notag           & = \mcO(\frac{L'}{2N_0+1}) = \mcO(\frac{1}{K'}).
\end{align}
Combining \eqref{2.31}, \eqref{2.33} and \eqref{2.35} shows  
\begin{align*}
  &\sum_{1\le j,j'\le J} \sum_{k'\in \Lambda_{K'}}  |\mcN(j,j',k')|(2L'+1)^d |\langle v_{k'(2L'+1)},b^j\rangle \cdot \overline{\langle v_{k'(2L'+1)},b^{j'}\rangle}| \\
   &\geq \frac{2}{3}\lambda \sum_{m\in \Z^d}A_m- \mcO(K/K') -\mcO(\frac{1}{K'})\sup_{1\leq j,j'\leq J} \sum_{k'\in \Lambda_{K'}}  (2L'+1)^d |\langle v_{k'(2L'+1)},b^j\rangle \cdot \overline{\langle v_{k'(2L'+1)},b^{j'}\rangle}|.   
\end{align*}
Using  properties  $(2)$ and $(3)$ of $b^j$ implies 
\begin{align}\label{weight}
 &\ \ \  \sup_{1\leq j,j'\leq J} \sum_{k'\in \Lambda_{K'}}  (2L'+1)^d\cdot |\langle v_{k'(2L'+1)},b^j\rangle \cdot \overline{\langle v_{k'(2L'+1)},b^{j'}\rangle}|\\
  \notag  & = \sup_{1\leq j,j'\leq J} \sum_{l\in \Lambda_{N_0}}   |\langle v_l,b^j\rangle \cdot \overline{\langle v_l,b^{j'}\rangle}|\\
  \notag  &\leq \sup_{1\leq j,j'\leq J} \|b^j\|_{\ell^2(\Lambda_{N_0})}\cdot \|b^{j'}\|_{\ell^2(\Lambda_{N_0})} \\
   \notag & =\sup_{1\leq j,j'\leq J} \|\wta^j\|_{\ell^2(\Lambda_{N_0})}\cdot \|\wta^{j'}\|_{\ell^2(\Lambda_{N_0})} \\
   \notag &\leq 1.
\end{align}
Thus,
\begin{align}\label{2.37}
    &\ \ \ \sum_{1\le j,j'\le J}  \sum_{k'\in \Lambda_{K'}}  |\mcN(j,j',k')|(2L'+1)^d\cdot  |\langle v_{k'(2L'+1)},b^j\rangle \cdot \overline{\langle v_{k'(2L'+1)},b^{j'}\rangle}| \\
    \notag & \geq \frac{2}{3}\lambda \sum_{m\in \Z^d}A_m- \mcO(K/K') -\mcO(\frac{1}{K'}) \geq \frac{1}{2}\lambda \sum_{m\in \Z^d}A_m,
\end{align}
where for the  last inequality, it requires that 
\begin{equation}\label{parameter condition 3}
  K/K'\le c(\lambda,M,J)\ll 1.
\end{equation}
What's more, from  the proof of \eqref{weight},  it follows that 
\begin{equation}\label{2.39}
      \sum_{1\le j,j'\le J} \sum_{k'\in \Lambda_{K'}}  (2L'+1)^d\cdot |\langle v_{k'(2L'+1)},b^j\rangle  \cdot \overline{\langle v_{k'(2L'+1)},b^{j'}\rangle}| \leq J^2. 
\end{equation}
Then \eqref{2.39} together with \eqref{2.37}  indicates  that, there exists at least one pair $(j,j',k')$ such that 
\begin{equation}\label{2.40}
  |\mcN(j,j',k')|\geq (\frac{1}{2}\lambda \sum_{m\in \Z^d}A_m)/J^2 = \frac{1}{2J^2}\cdot \lambda \sum_{m\in \Z^d}A_m. 
\end{equation}
Hence, we can define 
\begin{equation}\label{event initial scale}
  \Omega_{N_0}=\left\{ \varepsilon:\  \sup_{1\leq j,j'\leq J \atop k'\in \Lambda_{K'}}|\mcN(j,j',k')|\geq \frac{1}{2J^2}\cdot \lambda \sum_{m\in \Z^d}A_m   \right\},
\end{equation}
which is {\it independent of $E$}  for $E\in[E^*-\delta, E^*].$

\subsection{Probabilistic part:  Dudley's $L^{\psi_2}$-estimate}
At this stage, we have shown
\begin{equation*}
  \{ \varepsilon: \ \| (T_{N_0}+1)(E+1-D_{N_0})^{-1}  (T_{N_0}+1)(E+1-D_{N_0})^{-1} \| >1-\delta \} \subset \Omega_{N_0}, 
\end{equation*}
where $\Omega_{N_0}$ is defined by \eqref{event initial scale}. 
We can further obtain 
\begin{align*}
  |\mcN(j,j',k')| &=\left| \frac{1}{(2L'+1)^d}\sum_{l'\in \Lambda_{L'}} e^{-2\pi i (\theta_j-\theta_{j'})\cdot l'}D(\varepsilon)_{l'+k'(2L'+1)} \right| \\
    &=\left| \frac{\lambda}{(2L'+1)^d}\sum_{l'\in \Lambda_{L'}} e^{-2\pi i (\theta_j-\theta_{j'})\cdot l'}\left( \sum_{m\in \Z^d}A_m\cdot \varepsilon_{l'+k'(2L'+1)-m}\right) \right| \\
        &=\left| \frac{\lambda}{(2L'+1)^d} \sum_{m\in \Z^d}A_{k'(2L'+1)-m} \left(\sum_{l'\in \Lambda_{L'}} e^{-2\pi i (\theta_j-\theta_{j'})\cdot l'} \cdot \varepsilon_{l'+m}\right) \right|. \\
\end{align*}
By Dudley's $L^{\psi_2}$-estimate (i.e.,  Theorem \ref{Dudley inequality}), we obtain 
\begin{align}\label{2.42}
  \left\|  \sup_{1\leq j,j'\leq J \atop k'\in \Lambda_K'}|\mcN(j,j',k')|  \right\|_{\psi_2} & \lesssim \sqrt{\log(J^2 \cdot \#\Lambda_{K'})} \cdot \frac{\lambda}{(2L'+1)^d}\\
          \notag &\ \ \  \times \sup_{1\leq j,j'\leq J \atop k'\in \Lambda_K'}  \sum_{m\in \Z^d}A_{k'(2L'+1)-m} \left\| \sum_{l'\in \Lambda_{L'}} e^{-2\pi i (\theta_j-\theta_{j'})\cdot l'} \cdot \varepsilon_{l'+m}\right\|_{\psi_2} \\
           \notag &\le C(J)   \lambda  (\sum_{m\in \Z^d}A_m ) \sqrt{\log K'} (2L'+1)^{-\frac{d}{2}},  
\end{align}
where in the last inequality, we  use  Theorem \ref{sub orthogonal} for  zero-mean  random variables  $\varepsilon_n, n\in \Z^d$.  As a result, applying the Chernoff estimate (cf.  Theorem \ref{Chernoff bound})  implies 
\begin{align}\label{2.43}
  \mathbb{P}(\Omega_{N_0}) 
           & \leq  2e^{-c\frac{(2L'+1)^d}{\log K'}}, \ c=c(J)>0.
\end{align}

  \subsection{Determination of  the parameters}
  Summarizing  all the conditions on the parameters, i.e., \eqref{parameter condition 1}, \eqref{parameter condition 2} and \eqref{parameter condition 3},  shows 
  \begin{enumerate}
    \item $2N_0+1=(2L+1)(2K+1)=(2L'+1)(2k'+1)$;
    \item $K/K'\le c(\lambda,M,J)\ll 1$;
    \item $\delta^{\frac{1}{12}} L\le c(\lambda,M,\Theta)\ll 1$. 
  \end{enumerate} 
  To satisfy those conditions, we can  take 
  \begin{itemize}
    \item $L=\lfloor \delta^{-\frac{1}{24}}\rfloor \ \Rightarrow \delta^{\frac{1}{12}} L \sim  \delta^{\frac{1}{24}} \le c(\lambda,M,\Theta)\ll1$;
    \item $L'=\lfloor\delta^{-\frac{1}{48}} \rfloor \Rightarrow \frac{K}{K'}\sim \frac{L'}{L}\sim \delta^{\frac{1}{48}}\le c({\lambda,M,J})\ll 1$;
    \item $K\sim\frac{N_0}{L}\sim N_0 \delta^{\frac{1}{24}}, \ K'\sim N_0 \delta^{\frac{1}{48}}$;
    \item $\delta=(\log N_0)^{-10^3}$. 
  \end{itemize}
  With the above chosen  parameters, it suffices to ensure  $N_0\ge C(\lambda,M,J,\Theta) \gg1$, and \eqref{2.43} becomes 
  \begin{align}
    \mathbb{P}(\Omega_{N_0}) \leq 2e^{ -c(J)\frac{\delta^{-\frac{1}{48}}}{\log N_0-\frac{1}{48}|\log\delta| }}\leq e ^{ -(\log N_0)^3 }. 
  \end{align}

  \subsection{Proof of Theorem \ref{Green function estimates}:  Initial scales case}
    We have already proven  that for $\varepsilon\notin \Omega_{N_0}$,  
    \[\| (T_{N_0}+1)(E+1-D_{N_0})^{-1}  (T_{N_0}+1)(E+1-D_{N_0})^{-1} \| \leq 1-\delta.\]
    Hence, using the Neumann series  expansion \eqref{initial Neumann} gives  
    \begin{align}\label{L2 norm initial scale}
      \| G_{N_0} (E)\|&\leq \left(\|(E+1-D_{N_0})^{-1}\|+\|(E+1-D_{N_0})^{-1} (T_{N_0}+1)(E+1-D_{N_0})^{-1} \| \right)\cdot \sum_{s\geq 0}(1-\delta)^s \\
              \notag     & \leq \left( \frac{1}{M+1-\delta}+\frac{M+1}{(M+1-\delta)^2}\right)\cdot \delta^{-1} \\
                \notag   &\lesssim  \delta^{-1} =(\log N_0)^{ 10^3}    \ll  e^{N_0^{\frac{9}{10}}}. 
    \end{align}
    This establishes \eqref{Green L2 norm} for $N=N_0$.  

    Next, for $|n-n'|\geq \frac{N_0}{10}$,  we have  
    \begin{align*}
      |G_{N_0}(E;\varepsilon)(n,n')| &\leq \sum_{s\geq 0} \left| (D_{N_0}(\varepsilon)-E-1)^{-1} ((T_{N_0}+1)(D_{N_0}(\varepsilon)-E-1)^{-1})^s (n,n') \right|\\
           &= (\sum_{s< A}+\sum_{s \geq  A})\cdots. 
    \end{align*}
    For the $s\geq A$ part,  we use the  $\ell^2$-operator norm (as in  \eqref{L2 norm initial scale}) to get 
    \begin{align}\label{>A}
      \sum_{s\geq A}\cdots & \leq \left( \frac{1}{M+1-\delta}+\frac{M+1}{(M+1-\delta)^2}\right) \cdot \sum_{s\geq A}(1-\delta)^s \\
     \notag    &\lesssim \frac{1}{\delta} (1-\delta)^A \leq \frac{1}{\delta} e^{-\delta A}.
    \end{align}
    For the $s< A$ part,  we use the decay assumption \textbf{(A2)} of $T$ to get 
    \[|(T_{N_0}+1)(m,m')|\leq (|T(0)|+1)e^{-c|m-m'|}.\] 
    So,
    \begin{align*}
       &\ \ \  \left| (D_{N_0}(\varepsilon)-E-1)^{-1}   ((T_{N_0}+1)(D_{N_0}(\varepsilon)-E-1)^{-1})^s (n,n') \right|  \\
          & \leq (\frac{1}{M+1-\delta})^{s+1}  \sum_{n_1,n_2, \cdots,n_{s-1}\in \Lambda_{N_0}}  |(T_{N_0}+1)(n,n_1)|\cdot |(T_{N_0}+1)(n_1,n_2)|\cdots |(T_{N_0}+1)(n_{s-1},n')| \\
          &\leq  (\frac{1}{M+1-\delta})^{s+1} (|T(0)|+1)^s  \sum_{n_1,n_2 \cdots,n_{s-1}\in \Lambda_{N_0}}  e^{ -c|n-n_1|-c|n_1-n_2| \cdots -c|n_{k-1}-n'|   } \\
          &\leq (\frac{|T(0)|+1}{M+1-\delta})^s  (\# \Lambda_{N_0})^{s-1} e^{-c|n-n'|} \\
          &\leq  \left(C(d,|T(0)|, M)N_0^d \right)^s e^{-c|n-n'|}. 
    \end{align*}
    Hence,
    \begin{align}\label{<A}
      \sum_{s< A} \cdots & \leq \sum_{s < A}  \left( C(d,|T(0)|,M) N_0^d\right)^s e^{-c|n-n'|} \\
      \notag &\leq \left( C(d,|T(0)|, M) N_0^d \right)^A e^{-c|n-n'|}.
    \end{align}
    Combining  \eqref{>A}, \eqref{<A} and setting  $A= \frac{N_0}{(\log N_0)^2}$ yield 
    \begin{align}
      |G_{N_0}(E;\varepsilon)(n,n')| &\lesssim \frac{1}{\delta} e^{-\delta A} + \left( C(d,|T(0)|, M) N_0^d\right)^A e^{-c|n-n'|} \\
       \notag &= e^{  -\delta A + 10^3  (\log\log N_0) } +e^{ -c|n-n'| + C(d,|T(0)|, M)\cdot A \log N_0     } \\
       \notag & \leq e^{  -\frac{2N_0}{(\log N_0)^{2\times 10^3}}   } \leq  e^{  -\frac{|n-n'|}{(\log N_0)^{2\times 10^4}} },
    \end{align}
     where in the last inequality, we use  $N_0\ge C({d,M,\lambda,J,\Theta,|T(0)|})\gg 1$ and  $\frac{N_0}{10}<|n-n'|\leq 2N_0$. We have established  \eqref{Green off-diagonal decay} with   $\gamma_0=\frac{1}{(\log N_0)^{2\times 10^4}}$.

    \subsection{Proof of  Theorem \ref{Green function estimates, continuous version}: Initial scales case}
    Indeed, we can prove a  refined version Theorem \ref{Green function estimates, continuous version} for the initial scales case: 
    \begin{thm}\label{initial scale, continuous}
   For $N_0\ge  C(d,\lambda,M,J,\Theta,|T(0)|)\gg 1$ and $\delta= (\log N_0)^{-10^3}$, there is  some $\Omega'_{N_0}$ \text{independent of $E\in[E^*-\delta, E^*]$} such that
   \begin{equation}
    \mathbb{P}(\Omega'_{N_0})\leq e^{-(\log N_0)^3}, 
   \end{equation}
and  for all $\varepsilon \notin \Omega'_{N_0}$, the conclusions  in Theorem \ref{initial scale} also hold  for $G_{N_0}'(E;t,\varepsilon)$ as in Theorem \ref{Green function estimates, continuous version}. 
   \end{thm}
\begin{proof}[Proof of Theorem \ref{initial scale, continuous}]
    Theorem \ref{initial scale, continuous} will be  proved using the so called {\it free site argument} originated from \cite{Bou04}.  Recall the Green's function $G_{N_0}'(E;t,\varepsilon)$ only changes  $\varepsilon_0\in\{\pm 1\}$ to  $t\in[-1, 1]$ in $G_{N_0}(E;\varepsilon)$. So, if we assume conversely that 
      \[\exists t\in [-1,1] \ {\rm s.t.,}\  \| (T_{N_0}+1)(E+1-D_{N_0})^{-1}  (T_{N_0}+1)(E+1-D_{N_0})^{-1} \| \biggl|_{\varepsilon_0=t}>1-\delta, \]
   then using  similar argument  leading to \eqref{2.40} as before  shows  that there exist  a  pair $(j,j',k')$ and  some $t\in [-1,1]$ such that 
    \begin{equation}\label{2.49}
        |\mcN(j,j',k')| \biggl|_{\varepsilon_0=t} \geq  \frac{1}{2J^2}\cdot \lambda \sum_{m\in \Z^d}A_m. 
    \end{equation} 
    However,  we have 
    \begin{align*}
   |\mcN(j,j',k')| \biggl|_{\varepsilon_0=t}   &=\left| \frac{\lambda}{(2L'+1)^d} \sum_{m\in \Z^d}A_{k'(2L'+1)-m} \left(\sum_{l'\in \Lambda_{L'}} e^{-2\pi i (\theta_j-\theta_{j'})\cdot l'} \cdot \varepsilon_{l'+m}\right) \biggl|_{\varepsilon_0=t} \right|  \\
        &\leq  \left| \frac{\lambda}{(2L'+1)^d} \sum_{m\in \Z^d}A_{k'(2L'+1)-m} \left(\sum_{l'\in \Lambda_{L'} \atop l'+m\neq 0 } e^{-2\pi i (\theta_j-\theta_{j'})\cdot l'} \cdot \varepsilon_{l'+m}\right)  \right| \\
         &\ \ \  +  \left| \frac{\lambda}{(2L'+1)^d} \sum_{m\in \Lambda_{L'}}A_{k'(2L'+1)-m}  e^{2\pi i (\theta_j-\theta_{j'})\cdot m} \cdot t \right| \\
         & \leq |\mcI(j,j',k')| +\frac{\lambda}{(2L'+1)^d} \sum_{m\in \Z^d} A_m,
   \end{align*}
  where 
   \begin{equation}\label{mcI summation}
    \mcI(j,j',k') =\frac{\lambda}{(2L'+1)^d} \sum_{m\in \Z^d}A_{k'(2L'+1)-m} \left(\sum_{l'\in \Lambda_{L'} \atop l'+m\neq 0 } e^{-2\pi i (\theta_j-\theta_{j'})\cdot l'} \cdot \varepsilon_{l'+m}\right). 
   \end{equation}
   Hence, using \eqref{2.49} implies  that there exists  a pair $(j,j',k')$ such that 
  \begin{equation*}
     |\mcI(j,j',k')| \geq  \frac{1}{2J^2}\cdot \lambda \sum_{m\in \Z^d}A_m - \frac{\lambda}{(2L'+1)^d} \sum_{m\in \Z^d} A_m \geq \frac{1}{4J^2}\cdot \lambda \sum_{m\in \Z^d}A_m,
  \end{equation*}
  where the last inequality  is  ensured by $L'=\lfloor \delta^{-\frac{1}{48}}\rfloor$ and $0<\delta\ll1$.  This allows us to  define 
\begin{equation}\label{event initial scale, continuous version}
  \Omega'_{N_0}=\left\{ \varepsilon:\  \sup_{1\leq j,j'\leq J \atop k'\in \Lambda_{K'}}|\mcI(j,j',k')|\geq \frac{1}{4J^2}\cdot \lambda \sum_{m\in \Z^d}A_m   \right\},
\end{equation}
which  is independent of $E\in[E^*-\delta, E^*].$
Again, from  similar arguments  leading to  \eqref{2.42} and \eqref{2.43}, it follows that  
 \begin{align*}
    \mathbb{P}(\Omega'_{N_0}) \leq  e^ { -(\log N_0)^3}. 
  \end{align*}

This completes the proof of Theorem \ref{initial scale, continuous}. 
 \end{proof}

\section{Green's function estimates:  the large scales case}\label{LGFS}
In this section,  we aim to establish the Green's function estimates for all scales, thereby completing  the proof of Theorem \ref{Green function estimates} and Theorem \ref{Green function estimates, continuous version}.  The main scheme is based on the MSA induction. However, the presence of singular Bernoulli potentials causes an essential difficulty:  an a priori Wegner estimate  is unavailable.  Such a problem has been resolved by Bourgain \cite{Bou04} via developing the {\it free sites argument}  together with a new distributional inequality(cf. Lemma \ref{dislem}, based on Boolean functions analysis and Sperner's lemma). We will follow the method of \cite{Bou04}. 

In the following, we first perform some trim surgeries  on  probabilistic  events, which allows us to handle the weak independence of $D_n(\varepsilon)$ and perform the {\it free sites argument}. Next, to establish off-diagonal decay estimates on Green's functions in the MSA scheme, we will prove a key coupling lemma. Finally, we complete the proof of our main theorems  by combining the initial scales Green's function estimates with the MSA induction schemes.

\subsection{Trim of  the probabilistic events}\label{trim}
The  key trim operations consist of the  following two perspectives.

\subsubsection{Weak independence of the random potential} Recall that the potential is 
\begin{equation}\label{alloy potential}
  D(\varepsilon)_n=\lambda \sum_{m\in \Z^d} 2^{-|n-m|} \varepsilon_m. 
\end{equation}
For the  usual Bernoulli potential $D(\varepsilon)_n\to\varepsilon_n$, the restricted operator $H_N(\varepsilon)=R_N H(\varepsilon) R_N$ only depends  on  $(\varepsilon_j)_{j\in \Lambda_N}$. Now, for any $n\in \Lambda_N$, if $\varepsilon,\varepsilon'\in \{\pm 1\}^{\Z^d}$ satisfy  
\[ (\varepsilon_j) _{j\in\Lambda_{\frac{11}{10}N}}= (\varepsilon'_j) _{j\in\Lambda_{\frac{11}{10}N}},\]
then 
\begin{align}\label{weak independence perturbation}
  |D(\varepsilon)_n-D(\varepsilon')_n| & =\lambda \left| \sum_{m\notin \Lambda_{\frac{11}{10}N}} 2^{-|n-m|} (\varepsilon_m-\varepsilon'_m) \right| \lesssim 2^{-\frac{N}{11}}. 
\end{align}
It is important that  both \eqref{Green L2 norm} and \eqref{Green off-diagonal decay} remain  essentially  preserved  under  a  $2^{-\frac{N}{11}}$-perturbation on potentials.  Indeed, by the Neumann series  argument, we have 
\begin{equation}\label{weak independence Neumann}
  G_N(E;\varepsilon') =G_N(E;\varepsilon) \sum_{s\geq 0} (-1)^s ((D_N(\varepsilon')-D_N(\varepsilon))\cdot G_N(E;\varepsilon))^s.
\end{equation}
Assume $\varepsilon\in  \Omega_N(E)^c$, i.e.,  $G_N(E;\varepsilon)$ satisfies  \eqref{Green L2 norm} and \eqref{Green off-diagonal decay}. Then \eqref{weak independence perturbation} and \eqref{weak independence Neumann} can ensure 
\begin{align}\label{weak independence L2 norm}
  \| G_N(E;\varepsilon') \| \leq \frac{\| G_N(E;\varepsilon)\|}{1-2^{-\frac{N}{11} }\| G_N(E;\varepsilon)\|} < 2 e^{N^{\frac{9}{10}}}.
\end{align}
On the other hand, combining \eqref{Green L2 norm} and \eqref{Green off-diagonal decay} gives \begin{equation*}
  |G_N(E;\varepsilon)(n,n')|< e^{ N^{\frac{9}{10}}+\gamma \frac{N}{10}-\gamma|n-n'|   }
\end{equation*}
for all $n,n'\in \Lambda_N$. This  implies that if $|n-n'|>\frac{N}{10}$, then 
\begin{align*} 
  &\ \ \ |G_N(E;\varepsilon')(n,n')| \\
 \notag &\leq e^{-\gamma|n-n'|}+ \sum_{s\geq 1} 2^{-\frac{N}{11} s} \sum_{n_1,n_2, \cdots, n_s\in \Lambda_N} |G_N(E;\varepsilon)(n,n_1)|\cdot |G_N(E;\varepsilon)(n_1,n_2)|\cdots |G_N(E;\varepsilon)(n_s,n')| \\
 \notag &\leq e^{-\gamma|n-n'|}+\sum_{s\geq 1} 2^{-\frac{N}{11} s} e^{(s+1)(N^{\frac{9}{10}}+\gamma\frac{N}{10})} \sum_{n_1,n_2, \cdots, n_s\in \Lambda_N} e^{-\gamma|n-n_1|\cdots-\gamma|n_s-n'|} \\
 \notag & \leq e^{-\gamma|n-n'|}\left( \sum_{s\geq 0} (\# \Lambda_N \cdot 2^{-\frac{N}{11}} \cdot e^{N^{\frac{9}{10}}+\frac{\gamma}{10}N})^s  \right). 
\end{align*}
As in the MSA  iteration $\frac{\gamma_0}{2}\leq \gamma\leq \gamma_0\ll \frac{\log 2}{11}$,  we obtain  
\[  \# \Lambda_N \cdot 2^{-\frac{N}{11}} \cdot e^{N^{\frac{9}{10}}+\frac{\gamma}{10}N} \ll 1,\]
and 
\begin{equation}\label{weak independence off-diagonal decay}
  |G_N(E;\varepsilon')(n,n')|< 2 e^{-\gamma|n-n'|}. 
\end{equation}

The above argument indicates that, if we trim the event $\Omega_N(E)$ as 
\begin{equation}\label{weak independence trim}
  \trim_1(\Omega_N(E)) := \left(\{\pm 1\}^{\Z^d\setminus \Lambda_{\frac{11}{10}N}}\times \proj_{\Lambda_{\frac{11}{10}N}}(\Omega_N(E)^c)\right)^c,  
\end{equation}
then  $\trim_1(\Omega_N(E)) \subset\Omega_N(E) $ and 
\[\mathbb{P}(\trim_1(\Omega_N(E)))\leq  \mathbb{P}(\Omega_N(E)) \leq e^{-c\frac{(\log N)^2}{\log \log N}}. \] 
Moreover,   for  $\varepsilon'$ outside the set of  \eqref{weak independence trim}, the Green's function $G_N(E;\varepsilon')$ satisfies \eqref{weak independence L2 norm} and \eqref{weak independence off-diagonal decay}. It is  remarkable that  the set of  \eqref{weak independence trim} only depends on  $(\varepsilon_j)_{j\in \Lambda_{\frac{11}{10}N}}$. 

The above trim operation reveals  the weak independence of the potential \eqref{alloy potential}, and we always denote it  by $\trim_1(\cdot)$.
\begin{rmk}\label{trim 1 remark}
Indeed, it's easy to see that as soon as $\varepsilon$ is outside the set of \eqref{weak independence trim}, the Green's function 
\[G_{N}(E;r_j=\varepsilon_j,j\in \Lambda_{\frac{11}{10}N};r_j=t_j,j\notin \Lambda_{\frac{11}{10}N})\ {\rm for}\ \forall t_j\in [-1,1],\]
ssatisfies \eqref{weak independence L2 norm} and \eqref{weak independence off-diagonal decay}.
\end{rmk}

\subsubsection{Concentration of measure.} Another observation is that the density of Bernoulli  random variables $\omega_n, {n\in\Z^d}$  highly concentrates  at  $\pm 1$. This implies that an event in $\{\pm 1\}^{\Z^d}$ with large probability  can be trimmed  {\it free} from   certain   sites (in $\Z^d$). For example, assume 
\[\Omega\subset \{\pm 1\}^{\Z^d},\ \mathbb{P}(\Omega)>1-\kappa, 0<\kappa\ll 1.\]
Fix $0\in \Z^d$ and consider  $\varepsilon_0$. Denote 
\begin{equation}
  A=\{(\varepsilon_j)_{j\neq 0} : \ (\varepsilon_0=1;\varepsilon_j,j\neq 0) \ {\rm or} \  (\varepsilon_0=-1;\varepsilon_j,j\neq 0) \notin \Omega \}. 
\end{equation}
Then   $  \frac{1}{2} \cdot \mathbb{P} (A)< \kappa$.  If we trim the complement $\Omega^c$ as 
\begin{equation}\label{concentration of mes trim}
  \trim_2(\Omega^c)=\{\pm 1\}^{\{0\}}\times A^c,
\end{equation} 
then the set of \eqref{concentration of mes trim} is free from site  $0$ and 
\begin{equation}\label{concentration of mes prob}
  \P(\trim_2(\Omega^c))<2\kappa , \left( \trim_2(\Omega^c) \right)^c\subset \Omega. 
\end{equation}
We will always denote by  $\trim_2(\cdot)$  this operation.\\

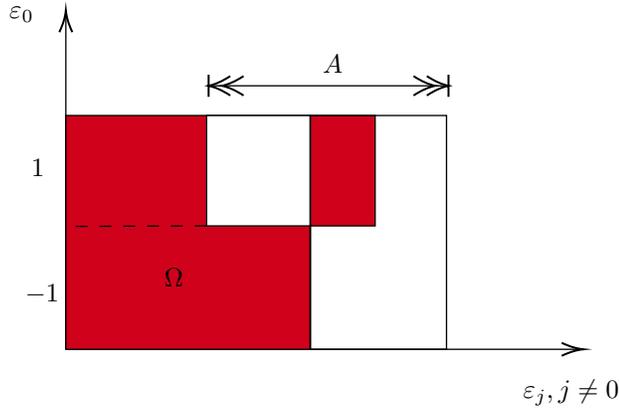
\begin{figure}[htbp]\label{trimfig}
  \centering

\tikzset{every picture/.style={line width=0.5pt}} 

\begin{tikzpicture}[x=0.75pt,y=0.75pt,yscale=-1,xscale=1]

\draw    (200.71,210.29) -- (200.71,41.29) ;
\draw [shift={(200.71,39.29)}, rotate = 90] [color={rgb, 255:red, 0; green, 0; blue, 0 }  ][line width=0.75]    (10.93,-3.29) .. controls (6.95,-1.4) and (3.31,-0.3) .. (0,0) .. controls (3.31,0.3) and (6.95,1.4) .. (10.93,3.29)   ;
\draw    (200.71,210.29) -- (459.71,210.29) ;
\draw [shift={(461.71,210.29)}, rotate = 180] [color={rgb, 255:red, 0; green, 0; blue, 0 }  ][line width=0.75]    (10.93,-3.29) .. controls (6.95,-1.4) and (3.31,-0.3) .. (0,0) .. controls (3.31,0.3) and (6.95,1.4) .. (10.93,3.29)   ;
\draw  [fill={rgb, 255:red, 208; green, 2; blue, 27 }  ,fill opacity=1 ] (200.71,92.29) -- (324,92.29) -- (324,210.29) -- (200.71,210.29) -- cycle ;
\draw  [fill={rgb, 255:red, 255; green, 255; blue, 255 }  ,fill opacity=1 ] (271.71,92.29) -- (324,92.29) -- (324,148) -- (271.71,148) -- cycle ;
\draw  [fill={rgb, 255:red, 208; green, 2; blue, 27 }  ,fill opacity=1 ] (324,92.29) -- (356.71,92.29) -- (356.71,148) -- (324,148) -- cycle ;
\draw    (324,148) -- (324,210.29) ;
\draw   (324,92.29) -- (392.71,92.29) -- (392.71,210.29) -- (324,210.29) -- cycle ;
\draw  [dash pattern={on 4.5pt off 4.5pt}]  (271.71,148) -- (200.71,148.29) ;
\draw    (272.71,77.29) -- (393.71,77.29) ;
\draw [shift={(393.71,77.29)}, rotate = 180] [color={rgb, 255:red, 0; green, 0; blue, 0 }  ][line width=0.75]    (0,5.59) -- (0,-5.59)(17.64,-4.9) .. controls (13.66,-2.3) and (10.02,-0.67) .. (6.71,0) .. controls (10.02,0.67) and (13.66,2.3) .. (17.64,4.9)(10.93,-4.9) .. controls (6.95,-2.3) and (3.31,-0.67) .. (0,0) .. controls (3.31,0.67) and (6.95,2.3) .. (10.93,4.9)   ;
\draw [shift={(272.71,77.29)}, rotate = 0] [color={rgb, 255:red, 0; green, 0; blue, 0 }  ][line width=0.75]    (0,5.59) -- (0,-5.59)(17.64,-4.9) .. controls (13.66,-2.3) and (10.02,-0.67) .. (6.71,0) .. controls (10.02,0.67) and (13.66,2.3) .. (17.64,4.9)(10.93,-4.9) .. controls (6.95,-2.3) and (3.31,-0.67) .. (0,0) .. controls (3.31,0.67) and (6.95,2.3) .. (10.93,4.9)   ;

\draw (171,36) node [anchor=north west][inner sep=0.75pt]   [align=left] {$\varepsilon_0$};
\draw (430,224) node [anchor=north west][inner sep=0.75pt]   [align=left] {$\varepsilon_j,j\neq 0$};
\draw (182,113) node [anchor=north west][inner sep=0.75pt]   [align=left] {$1$};
\draw (179,176) node [anchor=north west][inner sep=0.75pt]   [align=left] {$-1$};
\draw (249,168) node [anchor=north west][inner sep=0.75pt]   [align=left] {$\Omega$};
\draw (329,60) node [anchor=north west][inner sep=0.75pt]   [align=left] {$A$};

\end{tikzpicture}


\end{figure}

Moreover, if one wants  to free sites in  $Q\subset \Z^d$ for the above  event $\Omega$, then  \eqref{concentration of mes prob} will become 
\[\P(\trim_2(\Omega^c))<2^{\# Q}\cdot \kappa, \]
which shows  that such operation is robust only when $\# Q\ll |\log \kappa|$. 

Now,  only consider $\varepsilon_0$ and set  $\Omega=\Omega_N(E)^c$ in the above argument. By \eqref{N scale bad event prob}, this gives us a set free from  $\varepsilon_0$ as 
\[\Omega_N(E)\subset \trim_2(\Omega_N(E)),\  \P(\trim_2(\Omega_N(E)))<2 e^{-c\frac{(\log N)^2}{\log \log N}}.\]
Outside  of  $\trim_2(\Omega_N(E))$, both \eqref{Green L2 norm} and \eqref{Green off-diagonal decay} hold true.

\begin{rmk}
  The same operation can also be applied to  $\Omega'_N(E)$ defined  in Theorem \ref{Green function estimates, continuous version}. Note  that Theorem \ref{Green function estimates, continuous version} already implies that $\Omega'_N(E)$ is free from $\varepsilon_0$. Therefore, only the operation $\trim_1(\cdot)$ plays a role for this set.
\end{rmk}

\subsection{A key coupling lemma}
Next, we prove a key coupling lemma,  which will be used repeatedly  in  the MSA  iteration.

We say an $N$-size block $\Lambda_N(k)$ is good if the Green's function restricted to it satisfies  \eqref{Green L2 norm}  (resp. \eqref{weak independence L2 norm}) and \eqref{Green off-diagonal decay}  (resp.  \eqref{weak independence off-diagonal decay}) with decay rate $\gamma=\gamma_N$.

Denote by $|\Lambda|$ the diameter (size) of $\Lambda\subset\Z^d$ induced by  the norm $|\cdot|$. We have 
\begin{lem}\label{MSA}
 Fix $N_1=N^{\frac{4}{3}},1\ll N\ll L_1 \leq \frac{1}{2} N_1$.  Let $\Lambda_0' \subset \Lambda_1' \subset \Lambda\subset \Z^d$  be three blocks  satisfying 
\[|\Lambda_0'|\sim N, \ |\Lambda_1'|\sim L_1, \ |\Lambda|\sim N_1,  \]
and $\Lambda_1'$ contains the  $\frac{L_1}{10}$-neighborhood of $\Lambda_0'$. Assume 
there is  a class of $N$-size good blocks  
\[\mathbb F=\{\Lambda':\ \Lambda'\subset \Lambda\}\]
satisfying 
\begin{itemize}
  \item $\mathbb F$ and $\Lambda_1'$ cover $\Lambda$, i.e.,  
       \[\Lambda=\Lambda_1'\cup \mcA, \ \mcA=\bigcup_{\Lambda'\in \mathbb F}\Lambda';\]
  \item For each $n\in \Lambda\setminus \Lambda_0'$, there is a $\Lambda'\in \mathbb F$ such that
      \begin{equation}\label{good block nhd}
          \Lambda_{\frac{N}{5}}(n)\cap \Lambda \subset \Lambda'.
      \end{equation}
\end{itemize} 
Assume further  for  $E\in\R$, 
  \begin{equation}\label{center block L2 estimate}
           \| G_{\Lambda_1'} (E)\| <e^{L_1^{\frac{9}{10}}}. 
  \end{equation}
Then 
\begin{align}
\label{annuls L2 norm}  \| G_{\mcA}(E)\| &<e^{2N^{\frac{9}{10}}}, \\
\label{annuls off-diagonal decay}  |G_{\mcA}(E)(x,y)| &<e^{-\frac{4}{5}\gamma_N |x-y|}\ {\rm for}\ \forall |x-y|\geq N. 
\end{align}
Moreover, $G_\Lambda(E)$ is a good $N_1$-size block. 
\end{lem}

\begin{rmk}\label{exp rate remark}
  As we will see later,  the decay rate $\gamma=\gamma_N$ in \eqref{Green off-diagonal decay}   varies  with the scale  $N$ in  the MSA  iteration. Nevertheless, we will eventually show
  \[\gamma_0\geq \gamma_N=\gamma_0 -C\cdot \left(  \sum_{k: \ N_0\leq N_0^{(\frac{4}{3})^k}\leq N} N_0^{-\frac{1}{10}(\frac{4}{3})^k} \right) \geq \gamma_0 -C \cdot N_0^{-\frac{1}{10}} \geq \frac{\gamma_0}{2}.\]
  The last inequality holds true since  \eqref{initial parameter}  can ensure    $\gamma_0\gg\frac{1}{N_0}$.
\end{rmk}

\begin{figure}[htbp]
  \centering 
 
\tikzset{every picture/.style={line width=0.15pt}} 

\begin{tikzpicture}[x=0.75pt,y=0.75pt,yscale=-1,xscale=1]

\draw   (206,48) -- (482.71,48) -- (482.71,324.71) -- (206,324.71) -- cycle ;
\draw  [fill={rgb, 255:red, 208; green, 2; blue, 27 }  ,fill opacity=1 ] (292,110) -- (328.71,110) -- (328.71,146.71) -- (292,146.71) -- cycle ;
\draw   (242,70) -- (383.71,70) -- (383.71,211.71) -- (242,211.71) -- cycle ;
\draw   (409,106) -- (445.71,106) -- (445.71,142.71) -- (409,142.71) -- cycle ;
\draw   (321,138) -- (357.71,138) -- (357.71,174.71) -- (321,174.71) -- cycle ;

\draw (410,115) node [anchor=north west][inner sep=0.75pt]   [align=left] {good};
\draw (323,148) node [anchor=north west][inner sep=0.75pt]   [align=left] {good};
\draw (469,334) node [anchor=north west][inner sep=0.75pt]   [align=left] {$\Lambda$};
\draw (348,222) node [anchor=north west][inner sep=0.75pt]   [align=left] {$\Lambda_1'$};
\draw (271,88) node [anchor=north west][inner sep=0.75pt]   [align=left] {\textcolor[rgb]{0.82,0.01,0.11}{$\Lambda_0'$}};
\end{tikzpicture}

\end{figure}

\begin{proof}[Proof of Lemma \ref{MSA}]
We omit the dependence on $E$ for simplicity.  For any $x,y\in \mcA$ ($x\notin \Lambda_0'$), denote by $B_x\in \mathbb F$  the $N$-size block satisfying \eqref{good block nhd}.  Recall  the  resolvent identity 
\begin{equation}\label{mcA resolvent identity}
  G_{\mcA}=G_{B_x}\oplus G_{\mcA\setminus B_x}-(G_{B_x}\oplus G_{\mcA\setminus B_x})\Gamma G_{\mcA}, 
\end{equation}  
where 
\[\Gamma=(R_{B_x}T R_{\mcA\setminus B_x})\oplus (R_{\mcA\setminus B_x}T R_{B_x})\]
is the connecting matrix. Thus, applying  assumption \textbf{(A2)} implies 
\begin{align}\label{3.16}
  \sum_{y} |G_{\mcA}(x,y)| &\leq  \sum_{y}|G_{B_x}(x,y)|\chi_{B_x}(y)+\sum_{y}\sum_{\omega\in B_x \atop \omega'\in \mcA\setminus B_x}|G_{B_x}(x,\omega)|\cdot e^{-c|\omega-\omega'|}\cdot |G_{\mcA}(\omega',y)|.
\end{align}
Note that in the summation, we have  
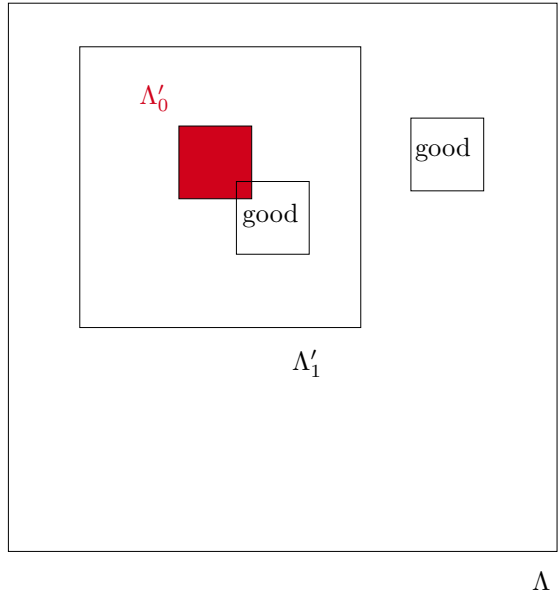
\begin{itemize}
  \item if $|\omega-x|\geq \frac{N}{10}$,  then \eqref{Green off-diagonal decay} gives $|G_{B_x}(x,\omega)|\leq e^{-\gamma_N |x-\omega|}$.
  \item if $|\omega-x|\leq \frac{N}{10}$, then \eqref{good block nhd} gives $|\omega-\omega'|\geq \frac{N}{10}\geq \frac{1}{2}|x-\omega'|.$
\end{itemize}
Hence, \eqref{3.16} can be further controlled via 
\begin{align*}
  \sum_{y} |G_{\mcA}(x,y)| &\leq \# B_x \cdot \| G_{B_x}\| + \left( \sup_{\omega'\in \mcA\setminus B_x} \sum_{y}|G_{\mcA}(\omega',y)| \right)\cdot (\# \Lambda)^2 \cdot (\| G_{B_x}\|e^{-c\frac{N}{10}}+e^{-\gamma_N \frac{N}{10}}) \\
                        &\leq (2N+1)^d e^{N^{\frac{9}{10}}}+ \frac{1}{2}  \left( \sup_{\omega'\in \mcA} \sum_{y}|G_{\mcA}(\omega',y)| \right),
\end{align*}
which  together with the Schur's  test gives
\begin{equation}\label{mcA L2 norm}
  \| G_{\mcA}\|\leq \left( \sup_{x\in \mcA} \sum_{y}|G_{\mcA}(x,y)| \right) \leq 2(2N+1)^{d}e^{N^{\frac{9}{10}}} \ll e^{2N^{\frac{9}{10}}}.
\end{equation}
Moreover,   by \eqref{mcA resolvent identity}, we have for  $|x-y|\geq N$, 
\begin{align}\label{matrix element estimate}
  |G_{\mcA}(x,y)| &\leq |G_{B_x}(x,y)| \chi_{B_x}(y) + \sum_{\omega\in B_x \atop \omega'\in \mcA\setminus B_x}|G_{B_x}(x,\omega)|\cdot e^{-c|\omega-\omega'|}\cdot |G_{\mcA}(\omega',y)| \\
  \notag   &\leq e^{-\gamma_N|x-y|} \chi_{B_x}(y)\\
  \notag&\ \ \ + \sum_{\omega'\in \mcA\setminus B_x} \left(\sum_{|\omega-x|\geq \frac{N}{10}}e^{-\gamma_N|x-\omega|-c|\omega-\omega'|}+\sum_{|\omega-x|\leq \frac{N}{10}} e^{N^{\frac{9}{10}}-\frac{c}{2}|x-\omega'|}\right) |G_{\mcA}(\omega',y)| \\
   \notag  &\leq e^{-\gamma_N|x-y|} \chi_{B_x}(y) + (\# \Lambda)^2 e^{-\gamma_N|x-x_1|}|G_{\mcA}(x_1,y)|, 
\end{align}
where $x_1\in \mcA\setminus B_x$ is a site  at which  $e^{-\gamma_N |x-\omega'|}|G_{\mcA}(\omega',y)|$ attains  its  maximum. 
The above estimate remains true  as long as $|x-y|>\frac{N}{10}$. Now, if  $|x_1-y|\leq \frac{N}{10}$, then $|x-x_1|\geq \frac{9}{10}|x-y|$, which   together with \eqref{mcA L2 norm} implies  
\begin{align}\label{3.18}
  |G_{\mcA}(x,y)| & \leq 2(2N_1+1)^d \exp\{2N^{\frac{9}{10}}-\frac{9}{10}\gamma_N |x-y'|\}  < e^{-\frac{4}{5}\gamma_N |x-y|}.
\end{align}
Otherwise,  if $|x_1-y|>\frac{N}{10}$,   we can iterate  the resolvent identity   and  perform  the same estimate for $|G(x_1,y)|$ to get $x_2, x_3,\cdots$.  This procedure can be  repeated  for  $s$ times, until $|x_s-y|\leq \frac{N}{10}$ or $s\sim |x-y|/(\frac{N}{10})$. Then one  obtains 
\begin{align}\label{3.19}
  |G_{\mcA}(x,y)| &\leq (2\# \Lambda)^s e^{2N^{\frac{9}{10}}-\frac{9}{10} \gamma_N|x-y|} \\
    \notag &\leq e^{-(\frac{9}{10}\gamma_N-2\frac{1}{N^{\frac{1}{10}}}-C\frac{\log N}{N})|x-y|} \\
    \notag &\leq e^{-\frac{4}{5}\gamma_N |x-y|}. 
\end{align}
In the estimates of  \eqref{3.18} and \eqref{3.19},  we use  (cf. Remark \ref{exp rate remark})  $\gamma_N\gtrsim \gamma_0\gg \frac{1}{N^C}$. Thus, we complete the proof of  \eqref{annuls L2 norm} and \eqref{annuls off-diagonal decay}. 

Now, we deal with $G_{\Lambda}$. Consider the  resolvent of identity
\begin{equation}\label{Lambda resolvent identity}
  G_{\Lambda}=G_{B}\oplus G_{\Lambda\setminus B}-(G_{B}\oplus G_{\Lambda\setminus B})\Gamma G_{\Lambda},
\end{equation}  
where we will take either $B\in \mathbb F$ or $B=\Lambda_1'$.  More precisely, 
\begin{itemize}
  \item \textbf{(Case 1:\ $x\notin \Lambda_0'$ or $y\notin \Lambda_0'$) }  \\
             Assume $x\notin \Lambda_0'$. Take $B=B_x\in \mathbb F$ satisfying  \eqref{good block nhd} and  perform similar estimate (for $x$)  leading to  \eqref{matrix element estimate}. Then we obtain 
             \begin{align}\label{3.22}
                |G_{\Lambda}(x,y)| &\leq e^{N^{\frac{9}{10}}}\chi_{B_x}(y)+ (\#\Lambda)^2 e^{-\gamma_N |x-\omega'|} |G_{\Lambda}(\omega',y)| \\
                     \notag & \leq e^{N^{\frac{9}{10}}}\chi_{B_x}(y)+ e^{-(\gamma_N-C\frac{\log N}{N}) |x-\omega'|} |G_{\Lambda}(\omega',y)| 
             \end{align} 
             for some $\omega'\notin B_x$ (and thus $|\omega'-x|\geq \frac{N}{5}$). We denote $\gamma'_{N}= \gamma_N-C\frac{\log N}{N}$.
    \item \textbf{(Case 2:\ $x,y \in \Lambda_0'$)}  \\
            In this case, we take $B=\Lambda_1'$. Then  applying  \eqref{Lambda resolvent identity} gives  
            \begin{equation*}
              |G_{\Lambda}(x,y)|\leq |G_{\Lambda_1'}(x,y)| \chi_{\Lambda_1'}(y) + \sum_{\omega\in \Lambda_1' \atop \omega'\in \Lambda\setminus \Lambda_1'}|G_{\Lambda_1'}(x,\omega)|\cdot e^{-c|\omega-\omega'|}\cdot |G_{\Lambda}(\omega',y)|.
            \end{equation*}
            Note that  $\omega'\notin \Lambda_1'\Rightarrow \omega'\notin \Lambda_0'$. Hence, applying \eqref{center block L2 estimate} enables  us to  get  
            \begin{align*}
                   |G_{\Lambda}(x,y)| & \leq  e^{L_1^{\frac{9}{10}}} + (\# \Lambda)^2 e^{L_1^{\frac{9}{10}}} |G_{\Lambda}(x_1,y)|. 
            \end{align*}
            As  $\Lambda_1'$ contains the  $\frac{L_1}{10}$-neighborhood of $\Lambda_0'$, we have $x_1 \notin \Lambda_0'$ and $|x_1-y|\geq \frac{L_1}{10}$. Applying  \eqref{3.22} for $|G_{\Lambda}(x_1,y)|$ repeatedly leads to  
            \begin{align*}
              |G_{\Lambda}(x,y)| &\leq  e^{L_1^{\frac{9}{10}}} + (\# \Lambda)^2 e^{L_1^{\frac{9}{10}}}\cdot  e^{-\gamma'_N |x_1-x_2|} |G_{\Lambda}(x_2,y)| \\
                  & \leq \cdots \\
                  & \leq e^{L_1^{\frac{9}{10}}}  + (\# \Lambda)^2 e^{L_1^{\frac{9}{10}}}\cdot e^{-\gamma'_N (|x_1-x_2|+\cdots+|x_s-x_{s+1}|)}  |G_{\Lambda}(x_{s+1},y)|. 
            \end{align*} 
   We stop the iterations  when $y\in B_{x_{s+1}}$ or $s\sim (\frac{L_1}{10})/(\frac{N}{5})$. Recall that during the iterations, we always have $|x_i-x_{i+1}|\geq \frac{N}{5}$. Thus, we can finally obtain   
            \begin{align}\label{3.23}
              |G_{\Lambda}(x,y)|\leq e^{L_1^{\frac{9}{10}}}  + (\# \Lambda)^2 e^{L_1^{\frac{9}{10}}-C\cdot \gamma'_N L_1} \cdot \| G_{\Lambda} \|.
            \end{align}
\end{itemize}
Combining \eqref{3.22} and \eqref{3.23} implies 
\begin{align}\label{3.24}
  |G_{\Lambda}(x,y)| \leq e^{L_1^{\frac{9}{10}}}+ e^{-\gamma'_N \frac{N}{5}}\|G_{\Lambda}\|.
\end{align}
Introducing the  Hilbert-Schmidt norm in  \eqref{3.24} gives  
\begin{align*}
  \| G_{\Lambda}\| &\leq \| G_{\Lambda}\|_{{\rm HS}} \leq \left(\sum_{x,y}(e^{L_1^{\frac{9}{10}}}+ e^{-\gamma'_N \frac{N}{5}}\|G_{\Lambda}\|)^2\right)^{\frac{1}{2}} \\
      &\leq \# \Lambda \cdot(e^{L_1^{\frac{9}{10}}}+ e^{-\gamma'_N \frac{N}{5}}\|G_{\Lambda}\| ) \leq \# \Lambda \cdot e^{L_1^{\frac{9}{10}}} +\frac{1}{2}\|G_{\Lambda}\|. 
\end{align*}
Thus,
\begin{equation}\label{3.25}
  \|G_{\Lambda}\|\leq 2(2N_1+1)^d e^{L_1^{\frac{9}{10}}}\ll e^{N_1^{\frac{9}{10}}}. 
\end{equation}
Moreover,  for the off-diagonal decay estimate of $G_{\Lambda}$, we assume  $|x-y|>\frac{N_1}{10}$. Then  $x\notin \Lambda_0'$ or $y\notin \Lambda_0'$. Again applying  \eqref{3.22} repeatedly gives 
\begin{align*}
  |G_{\Lambda}(x,y)|&\leq e^{-\gamma'_N|x-x_s|}|G_{\Lambda}(x_s,y)|. 
\end{align*}
We stop the iterations  until $s$ is large or $x_s\in \Lambda_0'$. If both $x_s,y\in \Lambda_0'$, we get $|G_{\Lambda}(x_s,y)|\leq \|G_{\Lambda}\|$. Otherwise,   if $x\in \Lambda_0'$ but $y\notin \Lambda_0'$,  then we  iterate  the resolvent identity  beginning with   $y$ (note that $G_\Lambda$ is self-adjoint), until $y_t\in \Lambda_0'$, to obtain 
\begin{align*}
    |G_{\Lambda}(x,y)|&\leq e^{-\gamma'_N(|x-x_s|+|y_t-y|)}|G_{\Lambda}(x_s,y_t)| \\
       &\leq e^{-\gamma'_N(|x-y|-|\Lambda_0'|)} \|G_{\Lambda}\|,
\end{align*}
which together with \eqref{3.25}  yields 
\begin{align}\label{3.26}
  |G_{\Lambda}(x,y)|\leq e^{-(\gamma'_N-C \frac{N}{N_1}-C\frac{1}{N_1^{\frac{1}{10}}})|x-y|}. 
\end{align}
Thus, we have 
\[\gamma_{N_1}=\gamma_N-C\cdot N^{-\frac{1}{10}}\geq \gamma'_N-C \frac{N}{N_1}-C\frac{1}{N_1^{\frac{1}{10}}},\]
which establishes  \eqref{Green off-diagonal decay} for $G_{\Lambda}$ with  rate $\gamma_{N_1}$.

We have completed the whole proof. 
\end{proof}

\begin{rmk}\label{loose small scale}
  From the proof below,  it is easy to  see that if we assume loosely $N\sim N_1^{\frac{3}{4}} $, for example, 
  \[(1-\frac{1}{100})N_1^{\frac{3}{4}}\leq N \leq (1+\frac{1}{100})N_1^{\frac{3}{4}},\]
 then the results and proofs remain essentially unchanged.
\end{rmk}

\begin{rmk}\label{off diagonal only good annuls}
  We emphasize that, from the proceeding proof, the restriction $L_1\leq \frac{1}{2}N_1$ only aims to ensure \eqref{3.25}.  Indeed,  one can  take  $\Lambda_1'=\Lambda$,  $L_1=N_1$,  and assume that \eqref{center block L2 estimate} holds true (which is  consistent with \eqref{3.25}): the off-diagonal decay estimate \eqref{3.26} remains true.
\end{rmk}

\begin{rmk}\label{not only one bad block}
   In Lemma \ref{MSA}, we may make a less restrictive assumption on the bad cube $\Lambda_0'$, namely, $\Lambda_0'$ may be replaced by a set $\mathbb B$ which is a union of   several bad $N$-blocks, for which   there is a collection $\mathbb F$ of a bounded number of $L_1$-blocks $\Lambda_1'$ satisfying \eqref{center block L2 estimate} and such that 
  \begin{itemize}
    \item distinct elements of $\mathbb F$ are at distance $\gtrsim L_1$;
    \item the $\frac{L_1}{10}$-neighborhood of $\mathbb B$  is  contained in $\cup_{\mathbb F}\Lambda_1'$.
  \end{itemize}
  Then the conclusion of Lemma \ref{MSA} remains true. In fact, we will only use this remark for the case that $\mathbb B$ is a union of $3$ bad $N$-blocks.
\end{rmk}

\subsection{Proof of Theorem \ref{Green function estimates} and Theorem \ref{Green function estimates, continuous version}}
In this part, we will finish the proof of Theorem \ref{Green function estimates} and Theorem \ref{Green function estimates, continuous version}, which is based on the MSA induction.

Note that we have proven the initial scales case in the Section \ref{section 2}.  We take $N_0\gg1$    so that  Theorem \ref{initial scale} and Theorem \ref{initial scale, continuous} hold true.  Indeed,  Theorem \ref{Green function estimates} and Theorem \ref{Green function estimates, continuous version} hold for  scales in $N_0\leq N\leq N_0^2$,  if  we let  
\begin{equation}\label{initial parameter}
  \delta=(\log N_0^2)^{-10^3}\sim (\log N_0)^{-10^3},\ \gamma_0=\frac{1}{(\log N_0^2)^{2\times 10^3}} 
\end{equation}
and  $E\in [E^*-\delta,E^*]$.  Recalling  the restriction \eqref{parameter condition 1},  we  take $L=\lfloor \delta^{-\frac{1}{24}}\rfloor$ and $L'=\lfloor \delta^{-\frac{1}{48}}\rfloor$. Denote 
\[{\rm Scale}_0:=\left\{N\in\Z:\ N_0\leq N\leq N_0^2, \ \frac{2N+1}{(2L'+1)(2L+1)}\in \Z_+  \right\}.\]
Then for $N\in {\rm Scale}_0$, \eqref{Green off-diagonal decay} holds true  for   $\gamma=\gamma_0$. Moreover, the probability estimate  \eqref{N scale bad event prob} (and \eqref{N scale bad event prob, continue version})  can  also be  ensured by 
\[\exp\{-(\log N)^3 \}\ll \exp \{-c\frac{(\log N)^2}{\log\log N}\},\ N_0\leq N\leq N_0^2.\]
It remains to establish \eqref{Green L2 norm} and \eqref{Green off-diagonal decay} for large scales $N\geq N_0^2$.  

We are now in a position to prove Theorem \ref{Green function estimates} and Theorem \ref{Green function estimates, continuous version} for scales $\geq N_1^2$, and we let 
\[N_1\geq N_0^2, \ N^{\frac{4}{3}}\sim N_1.\]
Assume Theorem \ref{Green function estimates} and Theorem \ref{Green function estimates, continuous version} hold true  for  scales belonging to ${\rm Scale}_0\cup[N_0^2, N_1)$. The assumption $N_1\geq N_0^2$ ensures that $N\geq N_0$,  so Green's function estimates hold for scale $N$.
\begin{rmk}\label{Dense initial scale}
  We want to remark that ${\rm Scale}_0$ is dense in $[N_0 ,N_0^2]$  because the distance between the two adjacent  elements in ${\rm Scale}_0$ is $(2L+1)(2L'+1)\sim\delta^{-\frac{1}{16}}\ll N_0^{0+}$. This suffices for the propagation of induction scales, i.e.,    all  $N_1\geq N_0^2.$  Indeed,   we can  find $N\in{\rm Scale}_0$ satisfying  $(1-\frac{1}{100})N_1^{\frac{3}{4}}\leq N\leq(1+\frac{1}{100})N_1^{\frac{3}{4}} $ so that 
  \[N=N_1^{\frac{3}{4}} \pm \mcO(\delta^{-\frac{1}{16}}),\ L_1=N_1^{\frac{15}{16}} \pm \mcO(\delta^{-\frac{1}{16}}).\]
  In this case,   the coupling  Lemma \ref{MSA} still works (cf. Remark \ref{loose small scale}).
\end{rmk}

First, we apply Theorem \ref{Green function estimates, continuous version} for scale $L_1\sim N^{\frac{5}{4}}\sim N_1^{\frac{15}{16}}\ll N_1$ so that for $\varepsilon$ outside of  $\Omega'_{L_1}(E)$, we have  
\[G'_{L_1}(E;t,\varepsilon)=G_{L_1}(E;r_0=t,r_j=\varepsilon_j (j\neq 0))\]
satisfies \eqref{Green L2 norm} and \eqref{Green off-diagonal decay} with $N=L_1$.  Moreover, we can  apply Theorem \ref{Green function estimates} at scale $N$  so that,  on ${\rm Trim}_1(\Omega_N)$ (only depends on sites  in $\frac{11}{10}N$-size block),   all $N$-size cubes satisfying 
\begin{equation}\label{3.27}
  \Lambda_N(k)\subset \Lambda_{N_1},0\notin \Lambda_{\frac{11}{10} N}(k)
\end{equation}
are  good.   More precisely,   we define  for $\Omega\subset \{\pm 1\}^{\Z^d}$, 
\[S_k\Omega=\{\varepsilon:\ S_k\varepsilon\in \Omega \},\ (S_k\varepsilon)_j=\varepsilon_{j-k},\]
and 
\begin{equation}\label{3.28}
  \Omega'_{N_1}(E)=\Omega'_{L_1}(E) \cup\left(\bigcup_{k \ {\rm statisfies \ } \eqref{3.27}} S_{-k}\trim_1(\Omega_N(E)) \right).
\end{equation}
By our construction \eqref{3.28}, $\Omega'_{N_1}(E)$ is free from $\varepsilon_0=t\in[-1, 1]$.  From  the induction assumptions,  it follows that  \begin{align}
  \P(\Omega'_{N_1}(E)) & \leq \P(\Omega'_{L_1}(E))  + CN^d \P(\Omega_N(E)) \\
 \notag     &\leq \P(\Omega'_{L_1}(E)) + e^{-(\log N)^{\frac{3}{2}}} \\
  \notag    &\leq \P(\Omega'_{L_1^{\frac{15}{16}}}(E)) +e^{-(\log L_1)^{\frac{3}{2}}}+e^{-(\log N)^{\frac{3}{2}}} \\
   \notag   &\cdots\\
    \notag  &\lesssim \P(\Omega'_{N_0}(E)) + \sum_{k\geq 1} e^{ -\log (N_0^{(\frac{16}{15})^k})^{\frac{3}{2}}} \lesssim \frac{1}{N_0} \ll \frac{1}{100}. 
\end{align}
Furthermore, outside of  $\Omega'_{N_1}(E)$,  the conditions  of  Lemma \ref{MSA} with $\Lambda_0'=\Lambda_{10N},\Lambda_1'=\Lambda_{L_1}, \Lambda=\Lambda_{N_1}$ are satisfied, which establishes the Green's function estimates at scale $N_1$.  We have proven  Theorem \ref{Green function estimates, continuous version} for the scale  $N_1$.  We should mention that in this proof, there is no need to  ``propagate  the randomness''.

Now, recall the arguments in Section \ref{trim}:  for each scale $N_0\leq L\leq N_1$,  we can trim  $\Omega_L(E)$ and $\Omega'_L(E)$ as 
\[\trim_1(\trim_2(\Omega_L(E))),\ \trim_1(\trim_2(\Omega'_L(E))).\]
For the remaining part,  we still use $\Omega_L(E),\Omega'_L(E)$ to denote the above two trimmed events for convenience.  That is to say, we can assume  $\Omega_L(E),\Omega'_L(E)$ only depend  on variables $(\varepsilon_j)_{j\in \Lambda_{\frac{11}{10}L}\setminus\{0\}}$. 

To establish Theorem \ref{Green function estimates} for the scale $N_1$ (we have to ``propagate  the randomness''),  we will apply the  {\bf free sites argument} together with  {\bf a distributional inequality} (cf. Lemma \ref{dislem}), which originate from \cite{Bou04} and significantly  extended later  in \cite{BK05}. Define 
\begin{equation}
  \mcS_0=\left\{ \frac{5}{4}rN:\ r\in\Z^d,|r|\leq \frac{4N_1}{5N} \right\}. 
\end{equation}
Now,  slightly adjust the elements of $\mcS_0$  near the boundary of $\Lambda_{N_1}$ 
 so that  each pair of  overlapped $N$-blocks with centers belonging to $\mathcal S$ (the adjustment of $\mathcal S_0$)  still has  a size of at least $\frac{N}{2}$.  
 Moreover, we have 
\begin{equation}
  \Lambda_{N_1}= \bigcup_{ k \in \mcS}\Lambda_{N}(k), 
\end{equation}
and  for $k, k'\in \mcS$,
\begin{align}
 &\label{disjoint far away} {\rm either \ } \Lambda_{N}(k)\cap \Lambda_{N}(k')\neq \emptyset \ {\rm or \ } \dist(\Lambda_N(k),\Lambda_N(k'))\geq \frac{N}{4},\\
\label{seperate cover}
  &\dist(k', \Lambda_{N}(k))\geq \frac{N}{5} \ {\rm if } \ k\neq k'. 
\end{align}
Indeed, if $\Lambda_{N}(k)\cap \Lambda_{N}(k')\neq \emptyset$, the size of their overlap is larger than $\frac{N}{2}$, which ensures \eqref{good block nhd}.
As  the adjustment from $\mcS_0$ to $\mcS$ only happens near the boundary of $\Lambda_{N_1}$  and is small,  we still use $\frac{5}{4}rN,r\in \Z^d\cap[-\frac{4N_1}{5N},\frac{4N_1}{5N}]^d$ to label the points in $\mcS$ for simplicity. {\it Now,  we want to ensure that  any two disjoint bad $N$-blocks (contained in $\Lambda_{N_1}$)  centering at $\mcS$ are not all bad, which requires removing more $\varepsilon$}. Indeed, for fixed $k,k'\in \mcS$, if $\Lambda_N(k)\cap \Lambda_N(k')=\emptyset$, then \eqref{disjoint far away} guarantees  that $S_{-k}\Omega_N(E)$ and $S_{-k'}\Omega_N(E)$ are independent. Hence
\begin{align*}
  \P ({ \rm Both  } \ \Lambda_N(k),\Lambda_N(k') \ {\rm are\ bad } )& \leq \P(\Omega_N(E))^2.  
\end{align*}
Considering  all possible $k,k'$, we obtain an event  $A_1$ with 
\begin{equation}\label{Prob removing 1}
  \P(A_1)\geq 1-(\# \mcS)^2\cdot \P(\Omega_N(E))^2, 
\end{equation}
such that for $\varepsilon\in A_1$ the following holds true:  There is some $r_0\in \Z^d$ such that if $k\in \mcS$ statisfies $\Lambda_{N}(k)\cap\Lambda_{10N}(\frac{5}{4}r_0 N)=\emptyset$, then $\Lambda_N(k)$ is good.  As we have already trimmed  the event and by \eqref{seperate cover},     $A_1$  depends only on $(\varepsilon_j)_{j\in \Z^d \setminus \mcS}$.\ \\
Next,  we will apply Theorem \ref{Green function estimates, continuous version}  at the scale $N$ to remove another probabilistic event as follows. For any $r\in \Z^d\cap[-\frac{4N_1}{5N},\frac{4N_1}{5N}]^d$, denote $J_r=\Lambda_N(\frac{5}{4}rN)$ for $\frac{5}{4}rN\in\mcS$. For any block $\mcI\subset \Z^d\cap[-\frac{4N_1}{5N},\frac{4N_1}{5N}]^d $ of  size $\frac{1}{3}(\log N)^2$, consider the event of 
\begin{equation}\label{in single I}
  {For \ all } \ r\in \mcI, \ G_{J_r}' \ { is \  bad  \ for}\    {\rm Theorem} \ \ref{Green function estimates, continuous version}.
\end{equation}
Then we know that there are $(\frac{1}{3}(\log N)^2)^d$ many mutually disjoint $J_r$ for $r\in \mcI$. 
\begin{figure}[htbp]
  \centering

\tikzset{every picture/.style={line width=0.15pt}} 

\begin{tikzpicture}[x=0.75pt,y=0.75pt,yscale=-0.5,xscale=0.5]

\draw  [fill={rgb, 255:red, 80; green, 227; blue, 194 }  ,fill opacity=1 ] (45.36,86.93) -- (202.07,86.93) -- (202.07,243.64) -- (45.36,243.64) -- cycle ;
\draw  [fill={rgb, 255:red, 80; green, 227; blue, 194 }  ,fill opacity=1 ] (249.36,86.93) -- (406.07,86.93) -- (406.07,243.64) -- (249.36,243.64) -- cycle ;
\draw  [fill={rgb, 255:red, 80; green, 227; blue, 194 }  ,fill opacity=1 ] (453.36,86.93) -- (610.07,86.93) -- (610.07,243.64) -- (453.36,243.64) -- cycle ;
\draw   (147.36,86.93) -- (304.07,86.93) -- (304.07,243.64) -- (147.36,243.64) -- cycle ;
\draw   (351.36,86.93) -- (508.07,86.93) -- (508.07,243.64) -- (351.36,243.64) -- cycle ;
\draw    (122.71,167.29) -- (20.71,167.29) ;
\draw [shift={(20.71,167.29)}, rotate = 225] [color={rgb, 255:red, 0; green, 0; blue, 0 }  ][line width=0.75]    (-5.59,0) -- (5.59,0)(0,5.59) -- (0,-5.59)   ;
\draw [shift={(122.71,167.29)}, rotate = 225] [color={rgb, 255:red, 0; green, 0; blue, 0 }  ][line width=0.75]    (-5.59,0) -- (5.59,0)(0,5.59) -- (0,-5.59)   ;
\draw    (224.71,167.29) -- (122.71,167.29) ;
\draw [shift={(122.71,167.29)}, rotate = 225] [color={rgb, 255:red, 0; green, 0; blue, 0 }  ][line width=0.75]    (-5.59,0) -- (5.59,0)(0,5.59) -- (0,-5.59)   ;
\draw [shift={(224.71,167.29)}, rotate = 225] [color={rgb, 255:red, 0; green, 0; blue, 0 }  ][line width=0.75]    (-5.59,0) -- (5.59,0)(0,5.59) -- (0,-5.59)   ;
\draw    (326.71,167.29) -- (224.71,167.29) ;
\draw [shift={(224.71,167.29)}, rotate = 225] [color={rgb, 255:red, 0; green, 0; blue, 0 }  ][line width=0.75]    (-5.59,0) -- (5.59,0)(0,5.59) -- (0,-5.59)   ;
\draw [shift={(326.71,167.29)}, rotate = 225] [color={rgb, 255:red, 0; green, 0; blue, 0 }  ][line width=0.75]    (-5.59,0) -- (5.59,0)(0,5.59) -- (0,-5.59)   ;
\draw    (428.71,167.29) -- (326.71,167.29) ;
\draw [shift={(428.71,167.29)}, rotate = 225] [color={rgb, 255:red, 0; green, 0; blue, 0 }  ][line width=0.75]    (-5.59,0) -- (5.59,0)(0,5.59) -- (0,-5.59)   ;
\draw    (530.71,167.29) -- (428.71,167.29) ;
\draw [shift={(530.71,167.29)}, rotate = 225] [color={rgb, 255:red, 0; green, 0; blue, 0 }  ][line width=0.75]    (-5.59,0) -- (5.59,0)(0,5.59) -- (0,-5.59)   ;
\draw    (632.71,167.29) -- (530.71,167.29) ;
\draw [shift={(632.71,167.29)}, rotate = 225] [color={rgb, 255:red, 0; green, 0; blue, 0 }  ][line width=0.75]    (-5.59,0) -- (5.59,0)(0,5.59) -- (0,-5.59)   ;

\draw (625,194) node [anchor=north west][inner sep=0.75pt]   [align=left] {$\mcI$};

\end{tikzpicture}

\caption{Disjoint $J_r,r\in \mcI$ along a line.}
\end{figure}
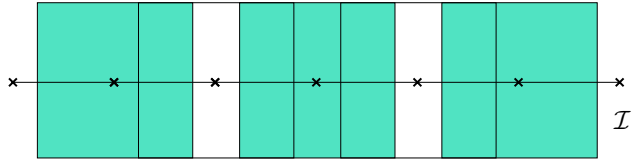\ \\ 
Actualy, \eqref{disjoint far away} shows that the $\frac{N}{10}$-neighborhood of those disjoint $J_r$ are still disjoint. This together with the independence  implies  
\begin{align}
  \P(\eqref{in single I}) & \leq \P({\rm There \ are} \ (\frac{1}{3}(\log N)^2)^d {\rm \ many \ disjoint} \ J_r {\rm \ fail\ for\ Theorem \ \ref{Green function estimates, continuous version}}.) \\
 \notag     &\leq \P(\Omega'_N(E))^{\frac{1}{9}(\log N)^{2d}} \leq (\frac{1}{50})^{\frac{1}{3^d}(\log N)^{2d}}. 
\end{align}
Counting  in all possible $\mcI$ (which is $\lesssim (\frac{4N_1}{5N})^d$ many) gives an  event $A_2$ (still only depends on $(\varepsilon_j)_{j\in \Z^d \setminus \mcS}$) with 
\begin{align}\label{removing prob 2}
  \P(A_2)& \lesssim (\frac{N_1}{N})^d (\frac{1}{50})^{\frac{1}{3^d}(\log N)^{2d}} \ll e^{-(\log N)^2}. 
\end{align}
We take 
\begin{equation*}
  A_3=A_1\setminus A_2. 
\end{equation*}
Then  $A_3$ only depends on $(\varepsilon_j)_{j\in \Z^d \setminus\mcS}$.  Combining  \eqref{Prob removing 1} and \eqref{removing prob 2}  yields 
\begin{equation}\label{remove prob 3}
  \P(A_3)\geq 1-(10N_1)^d \cdot e^ {-2c\frac{(\log N)^2}{\log\log N}} - e^{-(\log N)^2}. 
\end{equation}
Denote $\bar{\varepsilon}=(\varepsilon_j)_{j\in \Z^d \setminus\mcS}$. We can make a cylinder decomposition
\begin{equation}
  A_3=\bigcup_{\bar{\varepsilon}\in\proj_{\Z^d\setminus \mcS }A_1}\{\bar{\varepsilon}\}\times \{\pm 1\}^{\mcS}:=\bigcup_{\bar{\varepsilon}\in\proj_{\Z^d \setminus \mcS }A_1} T_{\bvepsilon}. 
\end{equation}
In summary, for  $\varepsilon$ in each cylinder $T_{\bvepsilon}= \{\bar{\varepsilon}\}\times \{\pm 1\}^{\mcS} $, we have shown 
\begin{itemize}
  \item $\exists r_0=r_0(\bvepsilon)$ so that,  if $k\in \mcS,\Lambda_{N}(k)\cap\Lambda_{10N}(\frac{5}{4}r_0 N)=\emptyset$, then $\Lambda_N(k)$ is good;
  \item  for $\forall \mcI$,   $\exists r\in \mcI$ so that,  the extended Green's function
         \begin{equation}\label{3.41}
              G_{J_r}(E;r_j=\varepsilon_j,j \in \Lambda_{\frac{11}{10}N}(\frac{5}{4}rN)\setminus\{\frac{5}{4}rN\};r_j=t_j,{\rm else})\ {\rm for}\ \forall t_j\in [-1,1]
         \end{equation} 
          is good (cf.  Remark \ref{trim 1 remark}, note that we use  $r_j\in[-1, 1]$ to indicate the  possible extension). 
\end{itemize}
Now,  pave $\Z^d\cap[-\frac{4N_1}{5N},\frac{4N_1}{5N}]^d $ with $\mcI$ of size $\frac{1}{3}(\log N)^2$, and pick one $r$ in each $\mcI$ so that,  \eqref{3.41} is good. This   leads to  a subset 
\[\mcR_0\subset \Z^d\cap[-\frac{4N_1}{5N},\frac{4N_1}{5N}]^d. \]
Let 
\begin{equation}\label{mcR construction}
  \mcR=\mcR_0\setminus ([-40,40]^d+r_0)
\end{equation}
and define   the set of {\bf  free sites} to be
\begin{equation}\label{set of free sites}  
  \mcS'=\{\frac{5}{4}rN:\ r\in \mcR \}\subset\mcS. 
\end{equation}
Obviously, both  $\mcR=\mcR(\bvepsilon)$ and $\mcS'=\mcS'(\bvepsilon)$  are  determined in each cylinder $T_{\bvepsilon}$. By above construction and Remark \ref{trim 1 remark},  we have \begin{itemize}
  \item $|r-r_0|>40$ for $r\in \mcR$ and thus,
\begin{equation}\label{mcS' far away}
  \dist(\mcS',\frac{5}{4}r_0 N)> 50 N; 
\end{equation}
  \item One can choose $r_l\in \mcR$ such that $|r_l-r_0|\sim l(\log N)^2$ with $1\leq l\leq \frac{N_1}{N(\log N)^3};$
  \item If $k\in \mcS,\Lambda_{N}(k)\cap\Lambda_{10N}(\frac{5}{4}r_0 N)=\emptyset$, the Green's function 
  \begin{equation}\label{3.45}
    G_{\Lambda_N(k)}(E;\bvepsilon;r_j=\varepsilon_j=\pm 1,j\in \mcS\setminus\mcS';r_j=t_j\in[-1,1],j\in \mcS')
  \end{equation}
  is good for all  possible $\varepsilon_j=\pm 1,j\in \mcS\setminus\mcS'$ and $t_j\in [-1,1]$.
\end{itemize} 
Hence, applying  Lemma \ref{MSA}  for 
\begin{equation}
 \mcA=\bigcup_{\Lambda_N(k)\cap \Lambda_{10N}(\frac{5}{4}r_0 N)=\emptyset } \Lambda_N(k)
\end{equation}
shows 
\begin{equation}
  G_{\mcA}= G_{\mcA}(E;\bvepsilon;r_j=\varepsilon_j=\pm 1,j\in \mcS\setminus\mcS';r_j=t_j\in[-1,1],j\in \mcS')
\end{equation}
satisfies \eqref{annuls L2 norm} and \eqref{annuls off-diagonal decay}.\\
Now fix $\widehat{\varepsilon}=(\varepsilon_j)_{j\in \mcS\setminus\mcS'}$ and further decompose
\[T_{\varepsilon}=\bigcup_{\hatepsilon\in \{\pm 1\}^{\mcS\setminus\mcS'}} \{(\bvepsilon,\hatepsilon)\}\times \{\pm 1\}^{\mcS'}  =\bigcup_{\hatepsilon\in \{\pm 1\}^{\mcS\setminus\mcS'}} T_{(\bvepsilon,\hatepsilon)}. \]
In each cylinder $T_{(\bvepsilon,\hatepsilon)}$,  since the self-adjoint operator 
\begin{equation}\label{3.47}
  H_{\Lambda_{N_1}}(\bvepsilon,\hatepsilon;r_j=t_j\in[-1,1],j\in \mcS')
\end{equation}
 analytically  depends on $t=(t_j)_{j\in \mcS'}\in [-1,1]^{\mcS'}$. The Kato-Rellich theorem enables  us to obtain the  continuous parameterizations of  the eigenvalue class of \eqref{3.47} as
\begin{equation}\label{eigenvalue class}
 \{E_{\tau}(t)\}_{\tau\in\Lambda_{N_1}},\ t\in [-1,1]^{\mcS'},
\end{equation}
where $E_{\tau}(t)$ is $C^1$ in each  $t_j\in[-1, 1]$ ($j\in\mcS'$). 
Denote by  $\xi_{\tau}(t)$ the corresponding normalized  eigenfunction of $E_{\tau}(t)$. Take one  $\mcE(t)\in \{E_{\tau}(t)\}$ with eigenfunction $\xi(t)=\{\xi_n(t)\}$. By first order eigenvalue variation formula and \eqref{alloy potential}, we obtain 
\begin{equation}\label{first order variation}
  \partial_{t_j}\mcE(t)=\langle \xi(t),\partial_{t_j}D_{N_1}(t) \xi(t)\rangle=\lambda\sum_{n\in \Lambda_{N_1}}2^{-|n-j|}|\xi(t)_n|^2
\end{equation}
and thus (by the  mean value theorem),
\begin{align}\label{mean value theorem}
  |\mcE(t)-\mcE(t')| & = |\langle(t-t'),\nabla \mcE(t+s(t'-t))\rangle| \ (0< s< 1) \\
  \notag  &\leq |t-t'|_{\infty}\cdot \sup_{t''=t+s(t'-t) \atop 0<s<1 } \left\{   \lambda \sum_{j\in \mcS' \atop n\in \Z^d} 2^{-|n-j|}\cdot |\xi(t'')_n|^2    \right\}. 
\end{align}
Now, assume additionally
\begin{equation}\label{energy closeness}
  |\mcE(t)-E|\leq N^{10}e^{-\gamma_0 N}. 
\end{equation}
Under $ N^{10}e^{-\gamma_0 N}$-perturbation,  using similar estimates in  \eqref{weak independence perturbation}$\sim$\eqref{weak independence off-diagonal decay} ensures  that the good estimates \eqref{annuls L2 norm} and \eqref{annuls off-diagonal decay} for 
\[G_{\mcA}(\mcE(t)) =G_{\mcA}(\mcE(t);\bvepsilon;r_j=\varepsilon_j=\pm 1,j\in \mcS\setminus\mcS';r_j=t_j\in[-1,1],j\in \mcS')\] 
are  essentially preserved (since $\gamma_0> \frac{1}{10}\cdot \frac{4}{5}\gamma_N$). Applying then  Poisson's formula  gives  
\begin{equation}\label{Poisson's formula}
  \xi(t) =-\left(G_{\mcA}(\mcE(t))\oplus G_{\Lambda_{N_1}}(\mcE(t))\right)\Gamma \xi(t). 
\end{equation}
 Therefore, we have the following cases: 
\begin{itemize}
  \item if $\dist(n,\frac{5}{4}r_0 N)<15N$, we have $|\xi_n|\leq \|\xi\|= 1;$
  \item if $\dist(n,\frac{5}{4}r_0 N)\geq 15N$, then 
  \begin{align*}
  n\in \mcA \ {\rm and} \ \dist(n,\Lambda_{N_1}\setminus\mcA)\geq \frac{4}{15} \dist(n,\frac{5}{4}r_0N) \geq  4N.\end{align*}
  In this case,  from \eqref{Poisson's formula}, \eqref{annuls L2 norm} and \eqref{annuls off-diagonal decay}, it follows that 
        \begin{align}
          |\xi_n|& \leq \sum_{\omega\in \mcA \atop \omega'\in \Lambda_{N_1}\setminus\mcA}|G_{\mcA}(n,\omega)|e^{-c|\omega-\omega'|}|\xi_{\omega'}| \\
               \notag &\lesssim N^d e^{2N^{\frac{9}{10}}-\frac{4}{5}\gamma_N |n-\omega'|} \ ({\rm for \ some } \ \omega'\notin \mcA) \\
               \notag &< e^{-\frac{3}{4}\gamma_N \dist(n,\partial_-\mcA)} \leq e^{-\frac{1}{5}\gamma_N \dist(n,\frac{5}{4}r_0 N)}. 
        \end{align}
\end{itemize}
Summarizing the  above estimates  concludes 
\begin{equation}\label{eigenfunction decay}
  |\xi(t)_n|\leq e^{-\frac{1}{5}\gamma_N (\dist(n,\frac{5}{4}r_0N)-15N)}\ {\rm for}\ \forall n. 
\end{equation}
Recalling  \eqref{mcS' far away},  if $\dist(n,\mcS')\leq N$, then  $\dist(n,\frac{5}{4}r_0 N)\geq 49N $ and thus by \eqref{eigenfunction decay},  
\begin{equation}
  |\xi(t)_n|\leq e^{ -6\gamma_N N}. 
\end{equation}
Using  \eqref{mean value theorem}, we have 
\begin{align}\label{3.57}
    \lambda \sum_{j\in \mcS' \atop n\in \Z^d} 2^{-|n-j|}\cdot |\xi(t)_n|^2  &\leq   \sum_{\dist(n,\mcS')\geq N }\sum_{j\in \mcS'} 2^{-|n-j|} +\sum_{\dist(n,\mcS')<N}\sum_{j\in\mcS'} 2^{-|n-j|}e^{-6\gamma_N N}\\
        \notag & \lesssim N^C 2^{-N}+ N^C e^{-6\gamma_N N} < e^{-5\gamma_N N}
\end{align} 
as long as $\mcE(t)$ satisfies \eqref{energy closeness}.\\
Now assume 
\begin{equation}\label{min close}
  \min_{t\in [-1,1]^{\mcS'}}|\mcE(t)-E|\leq e^{-\gamma_0 N}
\end{equation}
and the minimum attains at $t_0$. Pave $[-1,1]^{\mcS'}$ by $e^{-N^{1+}}$-size cubes and assume $t_0\in B$ ($B$ is $e^{-N^{1+}}$-size). Obviously,  one gets for all $t$,  
\begin{align}
    \lambda \sum_{j\in \mcS' \atop n\in \Z^d} 2^{-|n-j|}\cdot |\xi(t)_n|^2  &\leq    \sum_{j\in \mcS' \atop n\in \Z^d} 2^{-|n-j|} \lesssim \# \mcS'\lesssim N_1^d. 
\end{align} 
Thus, for all $t$ in cube $B$, we have 
\begin{align}
  |\mcE(t)-E|\leq |\mcE(t_0)-E|+|\mcE(t_0)-\mcE(t)|\leq e^{-\gamma_0 N}+C N_1^d e^{-N^{1+}}\ll N^{10} e^{-\gamma_N N}
\end{align}
satisfying \eqref{energy closeness}, and hence \eqref{3.57} holds. This gives us a better estimate that 
\begin{align}\label{3.61}
  |\mcE(t)-E|\leq |\mcE(t_0)-E|+|\mcE(t_0)-\mcE(t)|\leq e^{-\gamma_0 N}+ e^{-5\gamma_N N}|t-t_0|. 
\end{align}
We can propagate the above estimate  iteratively over the entire $[-1,1]^{\mcS'}$ via the $e^{-N^{1+}}$-scale covering, and finally get 
\[ |\mcE(t)-\mcE(t_0)|\leq e^{-\gamma_0 N}+ e^{-5\gamma_N N}|t-t_0| \leq 2e^{-\gamma_0 N} \ {\rm for}\ \forall t\in [-1,1]^{\mcS'}.\]
This indicates that  \eqref{min close} can imply  
\begin{equation}\label{max close}
    \max_{t\in [-1,1]^{\mcS'}}|\mcE(t)-E|\leq 2e^{-\gamma_0 N}. 
\end{equation}\ \\
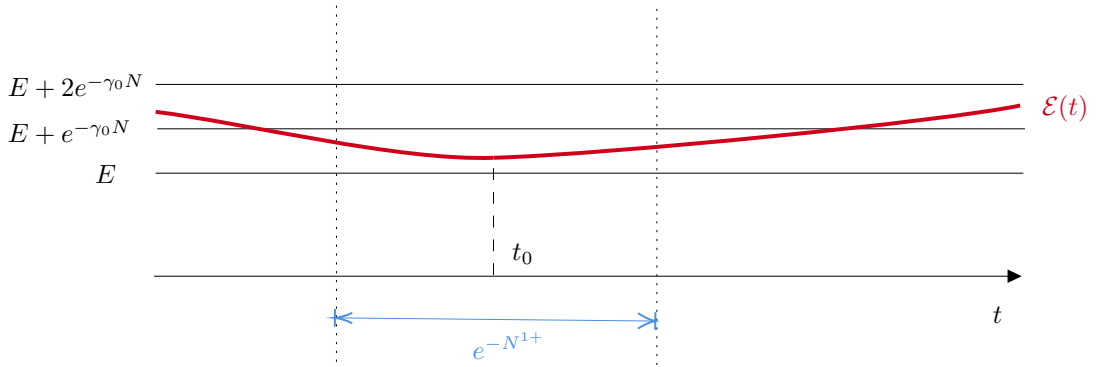
\begin{figure}[htbp]
  \centering

\tikzset{every picture/.style={line width=0.3pt}} 

\begin{tikzpicture}[x=0.75pt,y=0.75pt,yscale=-0.8,xscale=0.8]

\draw    (58.71,207) -- (602.71,207) ;
\draw [shift={(605.71,207)}, rotate = 180] [fill={rgb, 255:red, 0; green, 0; blue, 0 }  ][line width=0.08]  [draw opacity=0] (8.93,-4.29) -- (0,0) -- (8.93,4.29) -- cycle    ;
\draw    (59.71,142) -- (606.71,142) ;
\draw    (60.71,114) -- (607.71,114) ;
\draw  [dash pattern={on 0.84pt off 2.51pt}]  (173.71,36.29) -- (173.71,263.29) ;
\draw  [dash pattern={on 0.84pt off 2.51pt}]  (375.71,38.29) -- (375.71,265.29) ;
\draw    (59.71,86) -- (606.71,86) ;
\draw [color={rgb, 255:red, 74; green, 144; blue, 226 }  ,draw opacity=1 ]   (172.71,233.29) -- (374.71,235.29) ;
\draw [shift={(374.71,235.29)}, rotate = 180.57] [color={rgb, 255:red, 74; green, 144; blue, 226 }  ,draw opacity=1 ][line width=0.75]    (0,5.59) -- (0,-5.59)(10.93,-4.9) .. controls (6.95,-2.3) and (3.31,-0.67) .. (0,0) .. controls (3.31,0.67) and (6.95,2.3) .. (10.93,4.9)   ;
\draw [shift={(172.71,233.29)}, rotate = 0.57] [color={rgb, 255:red, 74; green, 144; blue, 226 }  ,draw opacity=1 ][line width=0.75]    (0,5.59) -- (0,-5.59)(10.93,-3.29) .. controls (6.95,-1.4) and (3.31,-0.3) .. (0,0) .. controls (3.31,0.3) and (6.95,1.4) .. (10.93,3.29)   ;
\draw  [dash pattern={on 4.5pt off 4.5pt}]  (272.71,206.29) -- (272.71,132.29) ;
\draw [color={rgb, 255:red, 208; green, 2; blue, 27 }  ,draw opacity=1 ][line width=1.5]    (59.71,103.29) .. controls (116.71,110.29) and (211.71,134.29) .. (272.71,132.29) ;
\draw [color={rgb, 255:red, 208; green, 2; blue, 27 }  ,draw opacity=1 ][line width=1.5]    (272.71,132.29) .. controls (343.71,130.29) and (557.71,109.29) .. (604.71,99.29) ;

\draw (586,225) node [anchor=north west][inner sep=0.75pt]   [align=left] {$t$};
\draw (20,135) node [anchor=north west][inner sep=0.75pt]   [align=left] {$E$};
\draw (258,241) node [anchor=north west][inner sep=0.75pt]  [color={rgb, 255:red, 74; green, 144; blue, 226 }  ,opacity=1 ] [align=left] {$e^{-N^{1+}}$};
\draw (-35,107) node [anchor=north west][inner sep=0.75pt]   [align=left] {$E+e^{-\gamma_0 N}$};
\draw (-35,78) node [anchor=north west][inner sep=0.75pt]   [align=left] {$E+2e^{-\gamma_0 N}$};
\draw (283,185) node [anchor=north west][inner sep=0.75pt]   [align=left] {$t_0$};
\draw (617,91) node [anchor=north west][inner sep=0.75pt]  [color={rgb, 255:red, 208; green, 2; blue, 27 }  ,opacity=1 ] [align=left] {$\mcE(t)$};

\end{tikzpicture}

\caption{The distance between $\mcE(t)$ and $E$.}
\end{figure}
$\qquad$\\
Now for all $t$, $\mcE(t)$ statisfies \eqref{energy closeness} and thus \eqref{eigenfunction decay} also holds uniformly about $t$. Recall that we have chosen  
\[r_l\in \mcR,\ |r_l-r_0|\sim l(\log N)^2,\ 1\leq l\leq \frac{N_1}{N(\log N)^3}.\]
Again,  fixing any  variables  $r_j=\varepsilon_j,j\in \mcS'\setminus\{\frac{5}{4}r_l N\}$, we  define the Boolean function \begin{equation*}
  f(\varepsilon'_l \big| 1\leq l\leq \frac{N_1}{N(\log N)^3} ) = \mcE((\bvepsilon,\hatepsilon);\varepsilon_j,j\in \mcS'\setminus\{\frac{5}{4}r_l N\};t_{\frac{5}{4}r_l N}=\varepsilon'_l). 
\end{equation*}
Consider the $l$-influence: 
\begin{align}\label{3.64}
  I_l&=f \big|^{\varepsilon'_l=1}_{\varepsilon'_l=-1} \\
  \notag & =\int_{-1}^{1}\partial_{t_{\frac{5}{4}r_l N}}\mcE(t_{\frac{5}{4}r_l N}=s)ds \\
  \notag &= 2\lambda\sum_{n\in \Z^d} 2^{-|n-\frac{5}{4}r_l N|}|\xi(t_{\frac{5}{4}r_l N}=s')_n|^2\\
  \notag &\overset{\eqref{eigenfunction decay}}{\lesssim} \sum_{|n-\frac{5}{4}r_l N|\geq \frac{1}{2}|r_l-r_0| N} 2^{-|n-\frac{5}{4}r_l N|}+\sum_{|n-\frac{5}{4}r_l N|< \frac{1}{2}|r_l-r_0| N} e^{-\frac{1}{5}\gamma_N (|n-\frac{5}{4}r_0N|-15N)}. 
\end{align}
Notice that $|n-\frac{5}{4}r_l N|< \frac{1}{2}|r_l-r_0| N$ ensures $|n-\frac{5}{4}r_0 N|\geq \frac{3}{4}|r_0-r_l|N\gg 15 N$. Thus, \eqref{3.64} becomes 
\begin{align}\label{influence upperbound}
  I_l& \leq \sum_{|n-\frac{5}{4}r_l N|\geq \frac{1}{2}|r_l-r_0| N} 2^{-|n-\frac{5}{4}r_l N|}+\sum_{|n-\frac{5}{4}r_l N|< \frac{1}{2}|r_l-r_0| N} \exp\{-\frac{1}{10}\gamma_N |r_l-r_0| N\} \\
  \notag &\leq e^{-\frac{1}{20}\gamma_N |r_0-r_l| N}< e^{-l\cdot C_1 \gamma_N (\log N)^2  N } :=b^{-l}, \ C_1>0. 
\end{align}
This gives an upper bound on  $l$-influence.  
Moreover, \eqref{eigenfunction decay} also tells us that 
\begin{align*}
  \sup_{t\in [-1,1]^{\mcS'}}\sum_{|n-\frac{5}{4}r_0 N|>20N}|\xi(t)_n|^2 & \lesssim \sum_{k\geq 20 N} k^{d-1}e^{-\frac{1}{5}\gamma_N(k-15N)} < e^{-\frac{1}{2}\gamma_N N},
\end{align*}
which implies  the concentration bound
\begin{equation}\label{concentration of eigenfunction}
  \inf_{t\in[-1,1]^{\mcS'}}\sum_{|n-\frac{5}{4}r_0 N|\leq 20N} |\xi(t)_n|^2 >\frac{1}{2}. 
\end{equation}
Thus,
\begin{align}\label{influence lowerbound}
  I_l&= 2\lambda\sum_{n\in \Z^d} 2^{-|n-\frac{5}{4}r_l N|}|\xi(t_{\frac{5}{4}r_l N}=s')_n|^2\\
  \notag &\geq {\lambda}\min_{|n-\frac{5}{4}r_0 N|\leq 20N} 2^{-|n-\frac{5}{4}r_l N|}\\
  \notag &\geq \frac{\lambda}{2} e^{-\frac{5}{4}|r_0-r_1|N-20N}\\
  \notag &>e^{-C_2l(\log N)^2 N}:= a^{-l},\ C_2>0. 
\end{align}
This gives a lower bound on  $l$-influence. Summarize the above  estimates as 
\[a^{-l}<I_l <b^{-l},\ 1\leq l\leq \frac{N_1}{N(\log N)^3}=m, \]
\[a=\exp\{C_2(\log N)^2 N\}, \  b=\exp\{C_1 \gamma_N (\log N)^2 N\}.\]
This can allow  us to apply the remarkable {\bf distributional inequality} of Bourgain \cite{Bou04}. More precisely, we have 
\begin{lem}[\cite{Bou04}, Lemma 2.1 and its Remark ]\label{dislem}
Let $f(\varepsilon_1,\cdots,\varepsilon_m)$ be a bounded function on $\{\pm 1\}^m$ and  denote $I_j=f|_{\varepsilon_j=1}-f|_{\varepsilon_j=-1}$ as the $j$-influence, which is a function of $\varepsilon_{j'}, j'\neq j$. Let $2<b<a, \frac{\log a}{\log b}\lesssim 1$ and 
${a^{-j}}\leq I_j\leq b^{-j},\ 1\leq j\leq m.$
Then for $\kappa>a^{-m},$
$$\sup_{E\in\R}{\mathbb P}(\varepsilon\in\{\pm 1\}^m:\ |f(\varepsilon)-E|<\kappa)<e^{-c(\log \frac{\log \kappa^{-1}}{\log a})^2},$$
where $c>0$ is some absolute constant.  If in addition, $1<\frac{\log a}{\log b}<K,$ then 
$$\sup_{E\in\R}{\mathbb P}(\varepsilon\in\{\pm 1\}^m:\ |f(\varepsilon)-E|<\kappa)<e^{-c\frac{(\log \frac{\log \kappa^{-1}}{\log a})^2}{\log K}}.$$
\end{lem}
\begin{rmk}
In the proof of this lemma, Bourgain used the Sperner's lemma. 
\end{rmk}
Now applying  Lemma \ref{dislem}  with  
\[\kappa=e^{-N\cdot N_1^{\frac{1}{10}}}=e^{-N_1^{\frac{17}{20}}}>a^{-m}=e^{-C_2\frac{N_1}{(\log N)}} \]
\[\frac{\log a}{\log b}=\frac{C_1}{C_2 \gamma_N}<C_3\gamma_0^{-1}:= K,  \]
concludes 
\begin{align}\label{prob for single eigenfunction}
  \P_{(\varepsilon'_l)}(|f-E|<e^{-N\cdot N_1^{\frac{1}{10}}}) & < e^{-C_4\frac{(\log \frac{\log \delta^{-1}} {\log a})^2}{\log K}}<e^{-C_5 (\log \frac{1}{\gamma_0})^{-1} (\log N)^2},
\end{align}
where $C_1,C_2,C_3,C_4,C_5$ are positive absolute constants. So, \eqref{prob for single eigenfunction} is the probabilistic estimate for one single parameterized function satisfying \eqref{max close}. Considering all possible $\mcE(t)\in \{E_{\tau}(t)\}$ satisfying \eqref{max close}, which is at most $C N_1^d$ many, one can remove a set of  probability
\begin{equation}\label{3.69}
  \P_{(\varepsilon'_l)}(\cdots)\lesssim N_1^d e^{-C_5 (\log \frac{1}{\gamma_0})^{-1} (\log N)^2 }<e^{-C (\log \frac{1}{\gamma_0})^{-1} (\log N)^2 }
\end{equation}
in each cylinder 
\[T_{(\bvepsilon,\hatepsilon),(\varepsilon_j)_{j\in \mcS'\setminus \{\frac{5}{4}r_l N\}}}=\{(\bvepsilon,\hatepsilon),(\varepsilon_j)_{j\in \mcS'\setminus \{\frac{5}{4}r_l N\}}\}\times\{\pm 1\}^{\{\frac{5}{4}r_l N\}}.\]
And, for $(\varepsilon'_l)$  not in the above  set, we have 
\begin{equation}
  \dist(\sigma(H_{N_1}),E)\geq e^{-N\cdot N_1^{\frac{1}{10}}}. 
\end{equation}
Since  $\hatepsilon, (\varepsilon_j)_{j\in \mcS'\setminus \{\frac{5}{4}r_l N\}}$ can be arbitrarily chosen, we in fact remove an event of  probability (on $(\varepsilon_j)_{j\in \mcS}$)  satisfying  \eqref{3.69} in each cylinder $T_{\bvepsilon}$. Finally, by taking account of  all $T_{\bvepsilon}\subset A_3$, removing  the above events   allows us to obtain a subset $\Omega_{N_1}(E)^c\subset A_3$ with 
\begin{equation}\label{3.71}
  \P(\dist(\sigma(H_{N_1}),E) \geq e^{-N\cdot N_1^{\frac{1}{10}}})\geq \P(\Omega_{N_1}(E)^c)\geq \P(A_3)\cdot(1-e^{-C (\log \frac{1}{\gamma_0})^{-1} (\log N)^2 }). 
\end{equation} 
Now for $\varepsilon\in \Omega_{N_1}(E)^c$, not only 
\begin{equation*}
  \| G_{N_1}(E;\varepsilon)\|\leq e^{N\cdot N_1^{\frac{1}{10}}}\ll e^{N_1^{\frac{9}{10}}}
\end{equation*}
is ensured, but also the Green's function of  \eqref{3.45} is good. This exactly enables us to apply Lemma \ref{MSA} together with Remark \ref{off diagonal only good annuls} to obtain 
\begin{equation}
    |G_{N_1}(E;\varepsilon)(n,n')|<e^{-\gamma_{N_1}|n-n'|} \ {\rm for} \ |n-n'|>\frac{N_1}{10}.
\end{equation}
Finally, recalling  \eqref{remove prob 3} and \eqref{3.71}, we have  the Wegner  type estimate 
\begin{align}\label{Wegner estimate}
     \P(\dist(\sigma(H_{N_1}),E)   < e^{-N\cdot N_1^{\frac{1}{10}}}) & \leq \P(\Omega_{N_1}(E)) \\
     \notag &\leq 1- \P(A_3)\cdot(1-e^{-C (\log \frac{1}{\gamma_0})^{-1} (\log N)^2 }) \\
     \notag & \leq (10N_1)^d \cdot e^{-2c\frac{(\log N)^2}{\log\log N}} + e^{-(\log N)^2}+e^{-C (\log \frac{1}{\gamma_0})^{-1} (\log N)^2 } \\
      \notag & \leq (10N_1)^d \cdot e^{-\frac{9}{8}c\frac{(\log N_1)^2}{\log\log N}} + e^{-(\log N)^2}+e^{-C (\log \frac{1}{\gamma_0})^{-1} (\log N)^2 } \\    
     \notag & \leq e^{-c\frac{(\log N_1)^2}{\log\log N_1}}, 
\end{align}
which  finishes  the proof of Theorem \ref{Green function estimates} at scale $N_1$.

\begin{rmk}\label{wegner independence}
Indeed, in the construction of  $\Omega_{N_1}(E)$, we focus   on   the $N$-size  blocks in $\Lambda_{N_1}$ and  random variables with indexes  in $\mcS,\mcS',\{\frac{5}{4}r_l N\} \subset \Lambda_{N_1}$. Thus, the event $\Omega_{N_1}(E)$ only depends on random variables with indexes  in 
\[\Lambda_{\frac{1}{10}N+N_1}\subset \Lambda_{\frac{11}{10}N_1},\]
and \eqref{Wegner estimate} indicates that 
\[\{ \dist(\sigma(H_{N_1}),E)   < e^{-N\cdot N_1^{\frac{1}{10}}} \}\subset \Omega_{N_1}(E).\] 
As a result,  if we define  
\begin{equation}
  \Sigma_N (E)=\{\forall t_j=\pm 1, \dist(\sigma(H_{N_1}(t_j,j\notin \Lambda_{\frac{11}{10}N_1};\varepsilon_j,j\in \Lambda_{\frac{11}{10}N_1})),E)   < e^{-N\cdot N_1^{\frac{1}{10}}}\}
\end{equation}
which is a subset of $\proj_{\Lambda_{\frac{11}{10}N_1}}(\Omega_{N_1}(E))\subset  \{\pm1\} ^{\Lambda_{\frac{11}{10}N_1}}$, it will have the same probability  estimate  bounded by $\P(\Omega_{N_1}(E))$ as that  in  \eqref{Wegner estimate}.
\end{rmk}

\section{Proof of Theorem \ref{Main}: Elimination of the energy}\label{elimation}
In this section, we will prove Theorem \ref{Main} via eliminating energy $E\in [E^*-\delta, E^*]$ appeared in 
Theorem \ref{Green function estimates} and Theorem \ref{Green function estimates, continuous version}. 
 Indeed,  in  Theorems  \ref{Green function estimates}, \ref{Green function estimates, continuous version}, the probabilistic estimates  \eqref{N scale bad event prob} and \eqref{N scale bad event prob, continue version} depend sensitively on $E$.  Once we eliminated the energy variables, the proof of localization follows immediately from the Shnol's theorem (cf. e.g., \cite{Kir08}),  which is based on the generalized eigenvalues (eigenfunctions) arguments. 
 

For fixed  $\delta$ as in  \eqref{initial parameter} and any $\varepsilon$, we say that a $N$-size block $\Lambda$ is $E$-bad, if the Green's function $G_{\Lambda}(E;\varepsilon)$ does not satisfy  \eqref{Green L2 norm} or  \eqref{Green off-diagonal decay}. Our main result of this section is  
\begin{thm}\label{Kirsch prob}
  Under the assumptions of Theorems  \ref{Green function estimates}, \ref{Green function estimates, continuous version} and  assuming  $N\geq N_0\gg1$,  we have for any  $\Lambda_N(k_1)$ and $\Lambda_N(k_2)$ satisfying  
  \begin{equation}\label{disjoint fixed N blocks}
     \dist(\Lambda_{N}(k_1),\Lambda_N(k_2))>\frac{N}{5}
  \end{equation}
 that 
  \begin{equation}\label{uniform prob}
    \P\left(\exists E\in [E^*-\frac{1}{2}\delta,E^*] \ {\rm s.t., \ both } \ \Lambda_N(k_1)\ {\rm and}\  \Lambda_N(k_2)\ {\rm are} \  E{\rm-bad}\right)< e^{-\tilde{c}\frac{(\log N)^2 }{\log\log N}},
  \end{equation}
 where   $\tilde{c}>0$ is some absolute constant.
\end{thm}
\begin{proof}[Proof of Theorem \ref{uniform prob}]
  First, we denote by  $\Omega_N(E), \Omega'_N(E)$   the trimmed sets in Theorems \ref{Green function estimates}, \ref{Green function estimates, continuous version} so that they depend only on random variables in $\Lambda_{\frac{11}{10}N}$.   Denote by   $\Omega_{N,k_1,k_2}$  the event in  \eqref{uniform prob}, which is also  trimmed  and depends only on random variables in $\Lambda_{\frac{11}{10}N}(k_1)\cup\Lambda_{\frac{11}{10}N}(k_2) $. 

  For the initial scales  $N_0\leq N\leq N_0^2$, recall that  both  $\Omega_N$ and $\Omega'_{N}$ are {\bf independent} of $E$ (cf. Theorems \ref{initial scale}, \ref{initial scale, continuous}).  As a result,  if $\varepsilon\notin \Omega_N$, then for all $E\in [E^*-\frac{1}{2}\delta,E^*]\subset [E^*-\delta,E^*]$,   the block $\Lambda_N$ is $E$-good.  Thus,
  \begin{align*}
    \P(\Omega_{N,k_1,k_2}) & \leq \P(S_{-k_1}\Omega_N \cap S_{-k_2}\Omega_N) \\
    \notag & =\P(\Omega_N)^2 < e^{-2c\frac{(\log N)^2}{\log\log N}},
  \end{align*}
  where in the above estimate, we  used \eqref{disjoint fixed N blocks} and the independence of trimmed events.  We only need to choose $0<\tilde{c}<\frac{1}{2} c$. 

  For large scale $N_1>N_0^2$, still take $N_1\sim N^{\frac{4}{3}}$. Recalling  Lemma \ref{MSA} and Remark \ref{off diagonal only good annuls}, to ensure the $N_1$-block $\Lambda_{N_1}$ is good,  it requires  that
  \begin{itemize}
    \item \textbf{($L^2$-norm estimate)} The Green's function satisfies
                             \[\|G_{N_1}(E)\|\leq e^{N_1^{\frac{9}{10}}},\]
         which is equivalent to
         \[\dist(\sigma(H_{N_1}(\varepsilon)),E)\geq e^{-N_1^{\frac{9}{10}}}.\]
         Hence by  \eqref{Wegner estimate},  one can get  
         \begin{equation*}
          A_{N_1} (E)=\left\{\dist(E,\sigma(H_{N_1}(\varepsilon))) \leq \frac{1}{10} e^{-N\cdot N_1^{\frac{1}{10}}}\right\}
         \end{equation*}
         such that 
         \begin{equation*}
          \P(A_{N_1}(E))\leq e^{-c\frac{(\log N)^2}{\log\log N}},
         \end{equation*}
           and for $\varepsilon\notin A_{N_1}(E),$ 
         \[\dist(E,\sigma(H_{N_1}(\varepsilon))) >\frac{1}{10} e^{-N\cdot N_1^{\frac{1}{10}}}\gg e^{-N_1^{\frac{9}{10}}}.\]
    \item \textbf{(Covered by good $N$-scale blocks)} By Remark \ref{not only one bad block}, one also needs  to ensure  that  except for  3 bad $N$-size blocks,   all points in $\Lambda_{N_1}$ can be covered by a $N$-size $E$-good  block as in  \eqref{good block nhd}. This can  be  satisfied if 
           \[\varepsilon\in B_{N_1}(E)^c:=\left\{{\rm No \ four}\ N{-\rm size} \ E{\rm-bad\ blocks\ in} \ \Lambda_{N_1}\ {\rm mutually\  satisfy\ \eqref{disjoint fixed N blocks}} \right\}.\]
           If $\varepsilon\in B_{N_1}(E)^c$,   all $E$-bad $N$-size blocks  will be well contained   in a block  (in $\Lambda_{N_1}$) of  size $10N$.   
  \end{itemize}
  Summarizing the above discussions implies 
  \[\Lambda_{N_1} {\rm is }\  E-{\rm bad} \Rightarrow A_{N_1}(E) \ {\rm or}\ B_{N_1}(E).\] 
  This leads to  
  \begin{align*}
    \P(\Omega_{N_1,k_1,k_2}) & \leq \P \left(\bigcup_{E\in [E^*-\frac{1}{2}\delta,E^*] }   S_{-k_1}(A_{N_1}(E)\cup B_{N_1}(E)) \cap S_{-k_2}(A_{N_1}(E)\cup B_{N_1}(E))\right) \\
                          \notag & \leq \P\left( \bigcup_{E\in [E^*-\frac{1}{2}\delta,E^*] }   (S_{-k_1}A_{N_1}(E) \cap S_{-k_2}A_{N_1}(E)) \right)\\
                          \notag&\ \ \   +\P\left( \bigcup_{E\in [E^*-\frac{1}{2}\delta,E^*] }   (S_{-k_1}B_{N_1}(E) \cup S_{-k_2}B_{N_1}(E))  \right).\\
  \end{align*}

  On one hand, for  
  \[\varepsilon\in \bigcup_{E\in [E^*-\frac{1}{2}\delta,E^*] }  ( S_{-k_1}B_{N_1}(E) \cup S_{-k_2}B_{N_1}(E)), \]
  we can get that, there exist  some $E\in [E^*-\frac{1}{2}\delta,E^*]$ and  at least four $N$-size blocks  
  \[\Lambda_{N}(k'_i)\subset \Lambda_{N_1}(k_1)\cup\Lambda_{N_1}(k_2),\ i=1,2,3,4\]
  satisfying \eqref{disjoint fixed N blocks} for $k'_i\neq k'_j$. Considering  all possible $(k'_i)_{i=1,2,3,4}$ and applying  \eqref{uniform prob} at scale $N$ yield  (we have  already trimmed  $\Omega_{N,k'_i,k'_j}$ depending only on random variables in $\Lambda_{\frac{11}{10}N}(k'_i)\cup\Lambda_{\frac{11}{10}N}(k'_j) $)
  \begin{align}\label{4.6}
    \P\left( \bigcup_{E\in [E^*-\frac{1}{2}\delta,E^*] }   (S_{-k_1}B_{N_1}(E) \cup S_{-k_2}B_{N_1}(E) ) \right) &\leq \# \{(k'_i)_{i=1,2,3,4}\}\cdot \P(\Omega_{N,k'_1,k'_2})^2\\
       \notag  &\lesssim N_1^{4d} e^{-2\tilde{c}\frac{(\log N)^2}{\log\log N}}.
  \end{align}

  On the other hand, for 
  \[\varepsilon\in\bigcup_{E\in [E^*-\frac{1}{2}\delta,E^*] }   (S_{-k_1}A_{N_1}(E) \cap S_{-k_2}A_{N_1}(E)),   \]
  there is some $E\in [E^*-\frac{1}{2}\delta,E^*]$ such that 
  \[\dist (E,\sigma(H_{\Lambda_{N_1}(k_1)}(\varepsilon))) \leq \frac{1}{10} e^{-N\cdot N_1^{\frac{1}{10}}},\ \dist (E,\sigma(H_{\Lambda_{N_1}(k_2)}(\varepsilon))) \leq \frac{1}{10} e^{-N\cdot N_1^{\frac{1}{10}}}.\]
  This implies  (since $\frac{1}{10}e^{-N\cdot N_1^{\frac{1}{10}}} \ll \frac{1}{4}\delta$)
  \begin{equation}
    \dist (\sigma(H_{\Lambda_{N_1}(k_1)}(\varepsilon)),\sigma(H_{\Lambda_{N_1}(k_2)}(\varepsilon)) \cap [E^*-\frac{3}{4}\delta,E^*]) \leq \frac{1}{5} e^{-N\cdot N_1^{\frac{1}{10}}}. 
  \end{equation}
  Moreover, we change the variables $\varepsilon_j$  with  $j\notin\Lambda_{\frac{11}{10}N_1}(k_1)$ to be arbitrary $t_j\in \{-1,1\}$. As in \eqref{weak independence perturbation}, this only causes  $2^{-\frac{N_1}{11}}$-perturbation of the potential and thus the spectrum of $H_{\Lambda_{N_1}(k_1)}$.   The same argument applies  to   $H_{\Lambda_{N_1}(k_2)}$. Then  (since $2^{-\frac{N_1}{11}}\ll \frac{1}{4}\delta$) for all $t_j\in \{-1,1\}$, the distance between   
   \[\spec_1(t;\varepsilon_j,j\in \Lambda_{\frac{11}{10}N_1}(k_1)):=\sigma(H_{\Lambda_{N_1}(k_1)}(r_j=t_j,j\notin\Lambda_{\frac{11}{10}N_1}(k_1);r_j= \varepsilon_j ,j\in \Lambda_{\frac{11}{10}N_1}(k_1)))\]
   and 
   \[\spec_2(\varepsilon_j,j\in \Lambda_{\frac{11}{10}N_1}(k_2)):=\sigma(H_{\Lambda_{N_1}(k_2)}(r_j=1,j\notin\Lambda_{\frac{11}{10}N_1}(k_2);r_j= \varepsilon_j ,j\in \Lambda_{\frac{11}{10}N_1}(k_2) )) \cap [E^*-\delta,E^*]\]
   is less than 
   \[ \frac{1}{5} e^{-N\cdot N_1^{\frac{1}{10}}} + 2\cdot 2^{-\frac{N_1}{11}} <  e^{-N\cdot N_1^{\frac{1}{10}}}.  \]
   Denote
   \[\bvepsilon_1= \varepsilon_j,j\in \Lambda_{\frac{11}{10}N_1}(k_1);\ \bvepsilon_2=\varepsilon_j,j\in \Lambda_{\frac{11}{10}N_1}(k_2).\]
   Considering  the conditional probability on $\bvepsilon_2$, we have that,  since $\bvepsilon_1$ and $\bvepsilon_2$ are independent,
   \begin{align}\label{4.8}
    \P(\forall t,\dist(\spec_1,\spec_2) <e^{-N\cdot N_1^{\frac{1}{10}}}) & =\E_{\bvepsilon_2}(\P(\cdots|\bvepsilon_2)). 
   \end{align}
   From  Remark \ref{wegner independence},  $\#\spec_2\leq \#\Lambda_{N_1}$ and $\spec_2\subset [E^*-\delta,E^*]$ (i.e., Theorem \ref{Green function estimates}  works), it follows that  for all   $\bvepsilon_2,$
   \begin{equation}\label{4.9}
    \P(\cdots|\bvepsilon_2) \lesssim N_1^d e^{-c\frac{(\log N)^2}{\log\log N}}.
   \end{equation}
Combining \eqref{4.8} and \eqref{4.9} gives 
   \begin{equation}\label{4.10}
    \P\left( \bigcup_{E\in [E^*-\frac{1}{2}\delta,E^*] }   (S_{-k_1}A_{N_1}(E) \cap S_{-k_2}A_{N_1}(E) ) \right) \lesssim  N_1^d e^{-c\frac{(\log N)^2}{\log\log N}}. 
   \end{equation}

   Finally, by \eqref{4.6} and \eqref{4.10}, we have 
   \begin{align*}
    \P(\Omega_{N_1,k_1,k_2})&\lesssim N_1^{4d} e^{-2\tilde{c}\frac{(\log N)^2}{\log\log N}}+  N_1^d e^{-c\frac{(\log N)^2}{\log\log N}} \\
      & < e^{-\tilde c\frac{(\log N_1)^2}{\log\log N_1}}, 
   \end{align*}
   where  we have  chosen  $0<\tilde{c}<\frac{1}{2}c$.
\end{proof}

\begin{proof}[Proof of Theorem \ref{Main}]
The Anderson localization (i.e., Theorem \ref{Main}) follows  from combining Theorem \ref{Kirsch prob} and the Shnol's theorem (cf. e.g., \cite[Theorem 9.13]{Kir08} for details). 
\end{proof}

\appendix
\section{A continuum model}\label{conmod}
In this section, similar to that in \cite[Section 4]{Bou04}, we consider the continuous analogues of the discrete model described  in the previous sections. Define on  $\R^d$  the Hamiltonian 
\begin{equation}\label{conmodel}
  H(\varepsilon)=P(i\partial)+V_{\varepsilon} (x), \ x\in\R^d
\end{equation}
with 
\begin{equation}
  V_{\varepsilon}(x)=\sum_{m\in \Z^d}\phi(x-m)\varepsilon_m
\end{equation}
and $\phi(x)$ is a function in $\R^d$ satisfying ($|x|:=\|x\|_{\infty}$)
\[\phi(x)\sim e^{-|x|};\ \widehat{\phi}(0)=1,\ \widehat{\phi}(n)=0\  {\rm for}\ \forall \ n \in \Z^d\setminus\{0\}.\]
Thus, by the Poisson's formula,  we have
\begin{equation*}
  \sum_{m\in \Z^d}\phi(x-m)=1,\forall \ x\in \R^d
\end{equation*}
The operator $P=P(i\partial)$ is the (pseudo-)differential operator 
$$\widehat{P(i\partial)f} (x)=\widehat P(x)\widehat f(x)$$
with its symbol $\widehat{P}(y)$  satisfying the following properties:
\begin{itemize}
  \item[\textbf{(P1)}] $\widehat{P}(x)\geq 0$ is real-valued; 
  \item[\textbf{(P2)}] $|\mcF^{-1}_{\R^d}\big((1+\widehat{P})^{-\frac{1}{2}})(x) |\lesssim e^{-c|x|} $; 
  \item[\textbf{(P3)}] $0$ is the minima  of $\widehat{P}(x)$ and 
  \[\widehat{P}^{-1}(\{0\})=\{y_1,y_2,\cdots,y_J\} \subset \R^d. \]  In addition, there exists a constant $D>0$ such that 
\begin{equation}\label{nondegenerate min}
  \whP(y) \geq D \min_{1\leq j\leq J}|y-y_j|^2. 
\end{equation}
\end{itemize}
In particular, the standard Laplacian $-\Delta$ satisfies all  the above properties. 

From now on, we use $\|\cdot\|$ to denote  both $\|\cdot\|_{L^2(\R^d)}$ and the  operator norm. 

Via the  standard argument (cf. e.g., \cite[Proposition 3.8]{Kir08}), we know that for a.e.  $\varepsilon,$
\[\inf \sigma(H(\varepsilon))=-1. \]
Therefore, we assume that $E\in [-1,-1+\delta]$ (lies in the edge of the spectrum). Write  the Green's function  as 
\begin{align*}
  G(E;\varepsilon)& =G(E+i o;\varepsilon) \\
    \notag & = (P+1)^{-\frac{1}{2}}\cdot [1+(P+1)^{-\frac{1}{2}} (V_\varepsilon-E-1) (P+1)^{-\frac{1}{2}}]\cdot (P+1)^{-\frac{1}{2}}
\end{align*}
and denote 
\begin{equation*}
  A_{\varepsilon}= (P+1)^{-\frac{1}{2}} (V_\varepsilon-E-1) (P+1)^{-\frac{1}{2}}. 
\end{equation*}

As done in  Section \ref{section 2}, we want to show that for $N_0\gg1 $ and $0<\delta\ll1$, with large probability we have 
\begin{equation}\label{aim}
   \| R_{N_0} A_{\varepsilon} R_{N_0} \|\leq 1-\delta,
\end{equation}
where $R_{N_0}$ is the restriction operator (i.e., with Dirichlet boundary condition) to $[-N_0,N_0]^d\subset \R^d$. 

Assuming that \eqref{aim} is not true,  one can find $\xi(x)\in C^{\infty}_c([-N_0,N_0]^d),\ \|\xi\|=1$ such that 
\begin{equation}\label{5.8}
  |\langle \xi,A_{\varepsilon}\xi\rangle|=|\langle (V_\varepsilon-E-1) (P+1)^{-\frac{1}{2}} \xi, (P+1)^{-\frac{1}{2}}\xi\rangle| >1-\delta.
\end{equation}
As $E+1\in[0,\delta]$, we have
\begin{equation*}
  \|V_\varepsilon-E-1\|\leq 1+\delta, 
\end{equation*}
which shows 
\begin{equation*}
  \|(P+1)^{-\frac{1}{2}} \xi\|^2>\frac{1-\delta}{1+\delta}=1-2\delta.  
\end{equation*}
Simultaneously, by the Plancherel theorem, we get 
\begin{align}\label{5.10}
  1-2\delta & <\|(P+1)^{-\frac{1}{2}} \xi\|^2  \\
    \notag & = \int_{\R^d}\frac{1}{\whP(\lambda)+1} |\whx(\lambda)|^2 d\lambda \\
    \notag &=\left(\int_{\{\whP(\lambda) < D\eta^2 \}}+\int_{\{\whP(\lambda) \geq  D\eta^2 \}}\right)\cdots \\
    \notag & \leq \int_{\{\whP(\lambda) < D\eta^2 \}} |\whx(\lambda)|^2 d\lambda+\frac{1}{1+D\eta^2}\int_{\{\whP(\lambda) \geq D\eta^2 \}}|\whx(\lambda)|^2 d\lambda \\
    \notag &=1-(1-\frac{1}{D\eta^2+1}) \int_{\{\whP(\lambda) \geq D\eta^2 \}}|\whx(\lambda)|^2 d\lambda,
\end{align}
where $\eta>0$ will be specified later.  Combining \eqref{nondegenerate min} and \eqref{5.10} gives 
\begin{equation}\label{5.12}
 \int_{\{\min_{1\leq j\leq J}|\lambda-y_j|\geq \eta \}}|\whx(\lambda)|^2 d\lambda \leq  \int_{\{\whP(\lambda) \geq D\eta^2 \}}|\whx(\lambda)|^2 d\lambda \leq \mcO(\frac{\delta}{\eta^2}). 
\end{equation}
Moreover,  we get 
\begin{align}
  \|(P+1)^{-\frac{1}{2}}\xi-\xi\|^2 & =\int_{\R^d} \left((\frac{1}{\whP(\lambda)+1})^{\frac{1}{2}}-1\right)^2|\whx(\lambda)|^2 d\lambda \\
       \notag &=\left(\int_{\{\whP(\lambda) < D\eta^2 \}}+\int_{\{\whP(\lambda) \geq  D\eta^2 \}}\right)\cdots \\
    \notag &\lesssim   \int_{\{\whP(\lambda) \geq D\eta^2 \}}|\whx(\lambda)|^2 d\lambda + \max_{\whP(\lambda)<D\eta^2} |(\frac{1}{\whP(\lambda)+1})^{\frac{1}{2}}-1|^2 \\
    \notag & \leq \mcO(\frac{\delta}{\eta^2})+\mcO(\eta^2).
\end{align}
Taking  $\eta=\delta^{\frac{1}{4}}$ implies 
\begin{equation}\label{5.14}
   \int_{\{\min_{1\leq j\leq J}|\lambda-y_j|\geq \delta^{\frac{1}{4}} \}}|\whx(\lambda)|^2 d\lambda= \mcO(\delta^{\frac{1}{2}})
\end{equation}
and
\begin{equation}\label{5.15}
\|(P+1)^{-\frac{1}{2}}\xi-\xi\| =\mcO(\delta^{\frac{1}{4}}).
\end{equation}
Combining  \eqref{5.15} and  \eqref{5.8} gives  
\begin{equation*}
  |\langle \xi,(V_\varepsilon-E-1)\xi\rangle|= 1-\mcO(\delta^{\frac{1}{4}}),
\end{equation*} 
and again by $E+1\in [0,\delta]$, 
\begin{equation}\label{5.17}
    |\langle \xi,V_\varepsilon\xi\rangle|= 1-\mcO(\delta^{\frac{1}{4}})-|E+1|= 1-\mcO(\delta^{\frac{1}{4}}). 
\end{equation}

Now,  we will use some geometric projection argument  to transfer the concentration of $\xi$, i.e.,  \eqref{5.14}, to the potential $V_{\varepsilon}(x)$. For this, we take $\psi(\lambda)$ to be a smooth bump function satisfying 
\[0\leq \psi\leq 1;\ \psi(\lambda)=1\ {\rm for}\ |\lambda|\leq 1;\ \psi(\lambda)=0\ {\rm for}\ |\lambda|>2\] 
and denote $\psi_{s}(\lambda)=\psi(s^{-1}\lambda)$. Note that 
\begin{equation*}
  |\langle \xi,V_{\varepsilon}\xi\rangle |=\left|\int_{\R^d}\int_{\R^d}\whx(x)\whx(y) \widehat{V_{\varepsilon}}(x-y)dxdy \right|. 
\end{equation*}
From \eqref{5.14},  it follows that the density $\whx(x)\whx(y)$  concentrates  on  
\begin{align}\label{concentrating set}
  \mcC &=\{\min_{1\leq j\leq J}|x-y_j|< \delta^{\frac{1}{4}} \} \times \{\min_{1\leq j\leq J}|y-y_j|< \delta^{\frac{1}{4}} \} \\
  \label{5.20} &:=\bigcup_{1\leq  i,j\leq J} B(y_i,\delta^{\frac{1}{4}})\times B(y_j,\delta^{\frac{1}{4}}). 
\end{align}
By choosing $\delta\ll 1$, we can ensure that \eqref{5.20} is a disjoint union. Denote by $L=\{(x,y)\in\R^{2d}:\ x+y=0\}$ the hyperplane in $\R^{2d}$, and  by $\proj_L$ the corresponding orthogonal projection. Let 
\begin{align}
  \mcG=\proj_L^{-1}(\proj_L(\mcC)) & = \bigcup_{i,j}\{|(x-y)-(y_i-y_j)|\lesssim \delta^{\frac{1}{4}} \} \\
    \label{5.22}  &\subset  \bigcup_{i,j}\{|(x-y)-(y_i-y_j)|< 2 \delta^{\frac{1}{5}} \}.  
\end{align}
Still, one can take $\delta\ll 1$  so  that \eqref{5.22} is a disjoint union.  
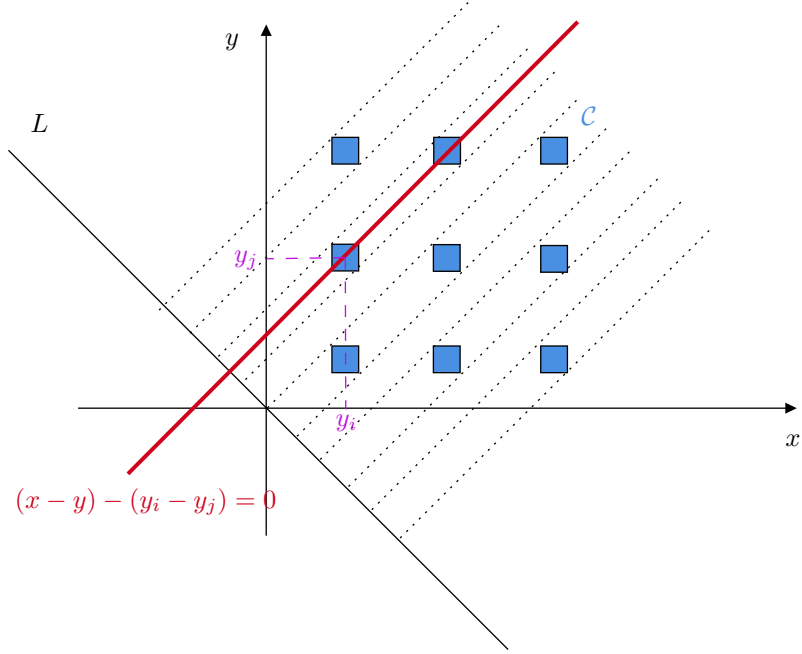
\begin{figure}[htbp]
  \centering

\tikzset{every picture/.style={line width=0.5pt}} 

\begin{tikzpicture}[x=0.75pt,y=0.75pt,yscale=-0.65,xscale=0.65]

\draw    (61,307.29) -- (615.71,307.29) ;
\draw [shift={(618.71,307.29)}, rotate = 180] [fill={rgb, 255:red, 0; green, 0; blue, 0 }  ][line width=0.08]  [draw opacity=0] (8.93,-4.29) -- (0,0) -- (8.93,4.29) -- cycle    ;
\draw    (207,12.29) -- (207,406.29) ;
\draw [shift={(207,9.29)}, rotate = 90] [fill={rgb, 255:red, 0; green, 0; blue, 0 }  ][line width=0.08]  [draw opacity=0] (8.93,-4.29) -- (0,0) -- (8.93,4.29) -- cycle    ;
\draw    (7,107) -- (394.71,494.71) ;
\draw  [fill={rgb, 255:red, 74; green, 144; blue, 226 }  ,fill opacity=1 ] (258,97) -- (278.71,97) -- (278.71,117.71) -- (258,117.71) -- cycle ;
\draw  [fill={rgb, 255:red, 74; green, 144; blue, 226 }  ,fill opacity=1 ] (258,180) -- (278.71,180) -- (278.71,200.71) -- (258,200.71) -- cycle ;
\draw  [fill={rgb, 255:red, 74; green, 144; blue, 226 }  ,fill opacity=1 ] (258,259) -- (278.71,259) -- (278.71,279.71) -- (258,279.71) -- cycle ;
\draw  [color={rgb, 255:red, 0; green, 0; blue, 0 }  ,draw opacity=1 ][fill={rgb, 255:red, 74; green, 144; blue, 226 }  ,fill opacity=1 ] (336.71,259) -- (357.43,259) -- (357.43,279.71) -- (336.71,279.71) -- cycle ;
\draw  [fill={rgb, 255:red, 74; green, 144; blue, 226 }  ,fill opacity=1 ] (336.71,180.29) -- (357.43,180.29) -- (357.43,201) -- (336.71,201) -- cycle ;
\draw  [color={rgb, 255:red, 0; green, 0; blue, 0 }  ,draw opacity=1 ][fill={rgb, 255:red, 74; green, 144; blue, 226 }  ,fill opacity=1 ] (337,97) -- (357.71,97) -- (357.71,117.71) -- (337,117.71) -- cycle ;
\draw  [fill={rgb, 255:red, 74; green, 144; blue, 226 }  ,fill opacity=1 ] (420,97) -- (440.71,97) -- (440.71,117.71) -- (420,117.71) -- cycle ;
\draw  [fill={rgb, 255:red, 74; green, 144; blue, 226 }  ,fill opacity=1 ] (420,259) -- (440.71,259) -- (440.71,279.71) -- (420,279.71) -- cycle ;
\draw  [fill={rgb, 255:red, 74; green, 144; blue, 226 }  ,fill opacity=1 ] (420,181) -- (440.71,181) -- (440.71,201.71) -- (420,201.71) -- cycle ;
\draw  [dash pattern={on 0.84pt off 2.51pt}]  (447.71,68.29) -- (206.71,308.29) ;
\draw  [dash pattern={on 0.84pt off 2.51pt}]  (488.71,108.29) -- (247.71,348.29) ;
\draw  [dash pattern={on 0.84pt off 2.51pt}]  (470.71,90.29) -- (229.71,330.29) ;
\draw  [dash pattern={on 0.84pt off 2.51pt}]  (510.71,129.29) -- (269.71,369.29) ;
\draw  [dash pattern={on 0.84pt off 2.51pt}]  (430.71,46.29) -- (420.19,56.77) -- (189.71,286.29) ;
\draw  [dash pattern={on 0.84pt off 2.51pt}]  (528.71,147.29) -- (287.71,387.29) ;
\draw  [dash pattern={on 0.84pt off 2.51pt}]  (550.71,169.29) -- (309.71,409.29) ;
\draw  [dash pattern={on 0.84pt off 2.51pt}]  (409.71,28.29) -- (168.71,268.29) ;
\draw  [dash pattern={on 0.84pt off 2.51pt}]  (387.71,10.29) -- (146.71,250.29) ;
\draw  [dash pattern={on 0.84pt off 2.51pt}]  (363.71,-7.71) -- (122.71,232.29) ;
\draw [color={rgb, 255:red, 208; green, 2; blue, 27 }  ,draw opacity=1 ][line width=1.5]    (448.71,7.29) -- (99.71,358.29) ;
\draw [color={rgb, 255:red, 189; green, 16; blue, 224 }  ,draw opacity=1 ] [dash pattern={on 4.5pt off 4.5pt}]  (268.36,190.36) -- (268.71,308.29) ;
\draw [color={rgb, 255:red, 189; green, 16; blue, 224 }  ,draw opacity=1 ] [dash pattern={on 4.5pt off 4.5pt}]  (268.36,190.36) -- (207.71,191.29) ;

\draw (608,326) node [anchor=north west][inner sep=0.75pt]   [align=left] {$x$};
\draw (173,14) node [anchor=north west][inner sep=0.75pt]   [align=left] {$y$};
\draw (22,77) node [anchor=north west][inner sep=0.75pt]   [align=left] {$L$};
\draw (449.71,71.29) node [anchor=north west][inner sep=0.75pt]  [color={rgb, 255:red, 74; green, 144; blue, 226 }  ,opacity=1 ] [align=left] {$\mcC$};
\draw (180,183) node [anchor=north west][inner sep=0.75pt]   [align=left] {\textcolor[rgb]{0.74,0.06,0.88}{$y_j$}};
\draw (259,310) node [anchor=north west][inner sep=0.75pt]   [align=left] {\textcolor[rgb]{0.74,0.06,0.88}{$y_i$}};
\draw (10,367) node [anchor=north west][inner sep=0.75pt]   [align=left] {\textcolor[rgb]{0.82,0.01,0.11}{$(x-y)-(y_i-y_j)=0$}};

\end{tikzpicture}
  
\end{figure}\ \\ 

Next,  let 
\[\Psi(\lambda)=\sum_{y_i-y_j:\ 1\leq i ,j\leq J } \psi_{\delta^{\frac{1}{5}} }(\lambda-(y_i-y_j))\]
which is supported  in 
\[\bigcup_{i,j}\{|\lambda-(y_i-y_j)|< 2 \delta^{\frac{1}{5}} \}.\]
Then one can decompose 
\begin{align}
   \notag \langle \xi,V_{\varepsilon}\xi\rangle& = \int_{\R^d}\int_{\R^d}\whx(x)\whx(y) \widehat{V_{\varepsilon}}(x-y)dxdy \\
       \label{5.23}       & = \int_{\R^d}\int_{\R^d}\whx(x)\whx(y) \widehat{V_{\varepsilon}}(x-y)\Psi(x-y)dxdy \\
        \label{5.24}     & \ \ \ + \int_{\R^d}\int_{\R^d}\whx(x)\whx(y) \widehat{V_{\varepsilon}}(x-y) (1-\Psi(x-y))dxdy.  
\end{align}
Direct computation shows that 
\begin{equation}\label{5.25}
  \eqref{5.23}=\left\langle\xi,\left(V_{\varepsilon}* \mcF_{\R^d}^{-1}  \left(\sum_{y_i-y_j}  \tau_{y_i-y_j} \psi_{\delta^{\frac{1}{5}}} \right) \right)\xi\right\rangle
\end{equation}
with $\tau_sf(\lambda)=f(\lambda-s)$ denoting the shift operator. Moreover, we have 
\begin{align}
\label{5.26}  \eqref{5.24} & = \int_{\R^d}\int_{\R^d}\whx(x)\chi_{\{\min|x-y_i|>\delta^{\frac{1}{4}}\}}(x) \cdot \whx(y) \widehat{V_{\varepsilon}}(x-y) (1-\Psi(x-y))dxdy  \\
   \label{5.27}     &\ \ \ + \int_{\R^d}\int_{\R^d}\whx(x) \chi_{\{\min|x-y_i|\leq \delta^{\frac{1}{4}}\}}(x) \cdot\whx(y) \widehat{V_{\varepsilon}}(x-y) (1-\Psi(x-y))dxdy.  
\end{align}
For \eqref{5.26}, applying the  Cauchy-Schwarz inequality gives
\begin{align}\label{5.28}
  |\eqref{5.26}| \leq \|\whx\|\cdot \|\whx \chi_{\{\min|x-y_i|>\delta^{\frac{1}{4}}\}} \| \cdot \|V_{\varepsilon}* (\delta_0 -\mcF^{-1}_{\R^d}\Psi)\|_{\infty},
\end{align}
where  ${\delta}_0$ is  the Dirac function. Moreover, for \eqref{5.27},   $1-\Psi(x-y)$ is  supported in\[\bigcap_{i,j}\{|(x-y)-(y_i-y_j)|> \delta^{\frac{1}{5}} \}.\]
Recalling $\min_{1\leq i\leq J}|x-y_i| \leq \delta^{\frac{1}{4}}$,  one gets 
\begin{equation}\label{5.29}
  \min_{1\leq i\leq J}|y-y_j|\geq \delta^{\frac{1}{5}}-\delta^{\frac{1}{4}}>\delta^{\frac{1}{4}}. 
\end{equation}
Hence, the valid integral region of \eqref{5.27} is contained in the set satisfying  \eqref{5.29}, and  so 
\begin{align}\label{5.30}
  |\eqref{5.27} | \leq \|\whx\|\cdot \|\whx \chi_{\{\min|y-y_j|>\delta^{\frac{1}{4}}\}} \| \cdot \|V_{\varepsilon}* (\delta_0 -\mcF^{-1}_{\R^d}\Psi)\|_{\infty}. 
\end{align}
Recalling  \eqref{5.14} and  using \eqref{5.28}, \eqref{5.30} together with  the Young's inequality,  we obtain 
\begin{align}
  |\eqref{5.24}| & \leq 2 \|\whx\|\cdot \|\whx \chi_{\{\min|y-y_j|>\delta^{\frac{1}{4}}\}} \| \cdot \|V_{\varepsilon}* (\delta_0 -\mcF^{-1}_{\R^d}\Psi)\|_{\infty} \\
  \notag & \lesssim \delta^{\frac{1}{4}}(\|V_\varepsilon\|_{\infty}+\|V_{\varepsilon}\|_{\infty} \cdot \|\mcF^{-1}_{\R^d}\Psi\|_1),
\end{align}
where $\|\cdot\|_1:=\|\cdot\|_{L^1}$. 
Direct computation shows that 
\begin{align*}
  \|\mcF^{-1}_{\R^d}\Psi\|_1 &\leq \sum_{y_i-y_j}\| \mcF^{-1}_{\R^d}(\tau_{y_i-y_j} \psi_{\delta^{\frac{1}{5}}}) \|_1 \\
  \notag   & =\sum_{y_i-y_j}\| \widehat{\psi}\|_1\leq J^2 | \widehat{\psi}\|_1 =C(J)<\infty. 
\end{align*} 
Thus,
\begin{equation}\label{5.33}
  |\eqref{5.24}| =\mcO(\delta^{\frac{1}{4}}). 
\end{equation}
Combining \eqref{5.17}, \eqref{5.25} and \eqref{5.33} gives  
\begin{equation}
  \sum_{1\leq i,j\leq J} \|V_{\varepsilon}* \mcF^{-1}_{\R^d} (\tau_{y_i-y_j}\psi_{\delta^{\frac{1}{5}}})\|_{L^{\infty}([-N_0,N_0]^d)}=1-\mcO(\delta^{\frac{1}{4}})> \frac{1}{2}.
\end{equation} 
This  concludes that  there are  some $i,j$ such that 
\begin{align}\label{5.35}
 &\ \ \  \|V_{\varepsilon}* \mcF^{-1}_{\R^d} (\tau_{y_i-y_j}\psi_{\delta^{\frac{1}{5}}})\|_{L^{\infty}([-N_0,N_0]^d)}\\
  \notag&= \|V_{\varepsilon}* (e^{2\pi i(y_i-y_j) \cdot }\mcF^{-1}_{\R^d} (\psi_{\delta^{\frac{1}{5}}})(\cdot))\|_{L^{\infty}([-N_0,N_0]^d)} >\frac{1}{2J^2}. 
\end{align}

Next,  by $\mcF^{-1}_{\R^d}\psi\in \mcS(\R^d)$ (the Schwarz space), we have 
\begin{align}
  \| \mcF^{-1}_{\R^d}(\psi_{\delta^{\frac{1}{5}}})- \tau_y\mcF^{-1}_{\R^d}(\psi_{\delta^{\frac{1}{5}}})\|_1 & =\int_{\R^d}|\mcF^{-1}_{\R^d}\psi(\lambda)-\mcF^{-1}_{\R^d}\psi(\lambda-\delta^{\frac{1}{5}}y) |d\lambda \\
             \notag & \lesssim |\delta^{\frac{1}{5}}y|\lesssim \delta^{\frac{1}{10}}
\end{align}
as long as $|y|\lesssim \delta^{-\frac{1}{10}}$. With this assumption and by the Young's inequality,  we have 
\begin{equation}\label{5.37}
  \|V_{\varepsilon}* (e^{2\pi i(y_i-y_j) \cdot }(\mcF^{-1}_{\R^d} (\psi_{\delta^{\frac{1}{5}}})-\tau_y\mcF^{-1}_{\R^d} (\psi_{\delta^{\frac{1}{5}}}))(\cdot))\|_{L^{\infty}([-N_0,N_0]^d)}  \lesssim \delta^{\frac{1}{10}}. 
\end{equation}
Take an  integer $R\sim \delta^{-\frac{1}{10}}$ and   $y=k=(k_s)\in \Z^d\cap[1,R]^d$ in \eqref{5.37}. One can obtain 
\begin{equation*}
    \|V_{\varepsilon}* (e^{2\pi i(y_i-y_j) \cdot }(\mcF^{-1}_{\R^d} (\psi_{\delta^{\frac{1}{5}}})-R^{-d}\sum_{1\leq k_s\leq R \atop s=1,\cdots ,d }\tau_k \mcF^{-1}_{\R^d} (\psi_{\delta^{\frac{1}{5}}}))(\cdot))\|_{L^{\infty}([-N_0,N_0]^d)}  \lesssim \delta^{\frac{1}{10}}
\end{equation*}
which together with \eqref{5.35} implies  
\begin{equation}\label{5.39}
  \|V_{\varepsilon}* (e^{2\pi i(y_i-y_j) \cdot }R^{-d}\sum_{1\leq k_s\leq R \atop s=1,\cdots ,d }\tau_k \mcF^{-1}_{\R^d} (\psi_{\delta^{\frac{1}{5}}})(\cdot))\|_{L^{\infty}([-N_0,N_0]^d)}>\frac{1}{4J^2}. 
\end{equation}

Finally, direct computation shows that 
\begin{align*}
    &\ \ \ \|V_{\varepsilon}* (e^{2\pi i(y_i-y_j) \cdot }R^{-d}\sum_{1\leq k_s\leq R \atop s=1,\cdots ,d }\tau_k \mcF^{-1}_{\R^d} (\psi_{\delta^{\frac{1}{5}}})(\cdot))\|_{L^{\infty}([-N_0,N_0]^d)} \\
    &=\| \int V_{\varepsilon}(x-y) R^{-d}\sum_{1\leq k_s\leq R \atop s=1,\cdots ,d } e^{2\pi i (y_i-y_j)  y}  \mcF^{-1}_{\R^d} (\psi_{\delta^{\frac{1}{5}}})(y-k)dy\|_{L^{\infty}([-N_0,N_0]^d)} \\
    &= \| R^{-d}\sum_{1\leq k_s\leq R \atop s=1,\cdots ,d }  \int  V_{\varepsilon}(x-y-k) e^{2\pi i (y_i-y_j) k} \cdot \mcF^{-1}_{\R^d} (\psi_{\delta^{\frac{1}{5}}})(y) e^{2\pi i (y_i-y_j) y} dy \|_{L^{\infty}([-N_0,N_0]^d)} \\
    &=\| \left(\mcF^{-1}_{\R^d} (\psi_{\delta^{\frac{1}{5}}})(\cdot) e^{2\pi i (y_i-y_j)\cdot} \right)* \left(R^{-d}\sum_{k}e^{2\pi i (y_i-y_j) k }V_{\varepsilon}(\cdot-k)\right)\|_{L^{\infty}([-N_0,N_0]^d)}. 
\end{align*}
Again by \eqref{5.39} and the Young's inequality, we get  
\begin{align*}
  \frac{1}{4J^2} &< \|R^{-d}\sum_{k}e^{2\pi i (y_i-y_j) k }\tau_k V_{\varepsilon} \|_{L^{\infty}([-N_0,N_0]^d)} \cdot \| \mcF^{-1}_{\R^d} \psi\|_1 \\
      \notag &=\left\| \sum_{m\in \Z^d}\phi(x-m) \left(R^{-d}\sum_{k}e^{2\pi i(y_i-y_j)k}\varepsilon_{m-k}\right) \right\|_{L^{\infty}([-N_0,N_0]^d)} \\
      \notag & \leq \sup_{|m|\leq 2N_0} \left| R^{-d}\sum_{k}e^{2\pi i(y_i-y_j)k}\varepsilon_{m-k}   \right| +e^{-\frac{1}{2}N_0},
\end{align*}
where in  the last inequality,  we used $|\phi(x)|\sim e^{-|x|}$. Thus, if  \eqref{aim} fails,  we finally get 
\begin{equation}\label{final continuous event}
  \Omega_{N_0}=\left\{\varepsilon: \  \sup_{1\leq i,j \leq J \atop |m|\leq 2N_0} \left| \sum_{k}e^{2\pi i(y_i-y_j)k}\varepsilon_{m-k}   \right|>\frac{R^d}{8J^2}\right\}. 
\end{equation}
It is worthy to compare  \eqref{final continuous event} with \eqref{event initial scale},  \eqref{event initial scale, continuous version}. Indeed, they all show  that the non-uniqueness of maxima (or minima) of the symbol will cause the transitions of arguments  in the front  of random variables. 

Finally, by the standard probabilistic estimate as in  \eqref{2.43}, we have 
\begin{equation*}
  \P(\Omega) \leq 2e^{-C(J) \frac{R^d}{\log N_0}}\leq 2e^{-C(J) \frac{\delta^{-\frac{1}{10}}}{\log N_0}}
\end{equation*}
Taking  $\delta=(\log N_0)^{-10^3}$ and $N_0\gg 1$ gives 
\begin{equation*}
  \P(   \| R_{N_0} A_{\varepsilon} R_{N_0} \|\leq 1-\delta )>1-e^{-C(\log N_0)^3}
\end{equation*}
which is exactly the same estimate as in \cite[Section 4, (4.28)]{Bou04} for the initial scales.  

Furthermore, \textbf{(P2)} ensures that the integral kernel 
\begin{equation*}
  |(P+1)^{-\frac{1}{2}}(x,y)|\lesssim e^{-c|x-y|}
\end{equation*}
and so  $|A_{\varepsilon}(x,y)|\lesssim e^{-c|x-y|}$. This together with \eqref{aim} and Neumann series expansion  argument will lead to 
\begin{equation*}
  \| (1+R_{N_0}A_{\varepsilon}R_{N_0})^{-1}\|< \delta^{-1}\leq e^{N_1^{\frac{9}{10}}}
\end{equation*}
and 
\begin{equation*}
  |(1+R_{N_0}A_{\varepsilon}R_{N_0})^{-1}(x,y)|\leq e^{-\gamma_0|x-y|}\ {\rm for}\ |x-y|>\frac{N_0}{10}. 
\end{equation*}
These are the Green's function estimates  required for the long-range operator $A_{\varepsilon}$ in initial scales. 

The remaining  proofs on MSA  iteration and localization are again standard (cf.  \cite[Section 4]{Bou04}).

\section{Basic facts on the Floquet-Bloch theory}\label{Floquet-B}
Consider  the  periodic Schr\"odinger operator on $\ell^2(\Z^d)$
\[H=P+V,\]
where $P$ is a convolution type operator with symbol $h(x),\ x\in \T^d$ and $V$ is $\Lambda_N$-periodic potential, namely, 
\[V(n+l)=V(n)\ {\rm for}\  \forall n\in \Z^d, l\in [(2N+1)\Z]^d.  \]
We still define the Fourier transform $\mcF:\ \ell^2(\Z^d)\rightarrow L^2(\T^d)$ via  \eqref{Fourier transform lattice}. Then 
\begin{align}
  (\mcF H   \mcF^{-1} u)(x) =h(x)u(x)+\sum_{n\in \Z^d} V(n) e^{2\pi i n\cdot x} \int_{\T^d}e^{-2\pi i n\cdot y}u(y)dy,\ \T^d=\R^d/\Z^d. 
\end{align}
Now,  for any $n\in \Z^d$,  we have the unique representation $n=k+l,k\in \Lambda_N,l\in [(2N+1)\Z]^d,$ which leads to 
\[\Z^d = [(2N+1)\Z]^d +\Lambda_N.\]
We define the unitary isometry
\begin{align}\label{Floquet transform}
  U: \  \ell^2(\Z^d)&\to L^2\left((\frac{\T}{2N+1})^d\right)\otimes \ell^2(\Lambda_N):= \mcH, \\
  \notag  (u_n)_{n\in \Z^d} &\longmapsto (u_k(x))_{k\in \Lambda_N},\ u_{k}(x)=\sum_{l\in [(2N+1)\Z]^d } u_{k+l}\cdot  e^{2\pi i l\cdot x} .
\end{align}
It's easy to see  that $u_k(x)$ is  $\frac{1}{2N+1}$-periodic in $x$, and 
\[u(x)=\mathcal F[(u_n)_{n\in\Z^d}]=\sum_{k\in \Lambda_N}e^{2\pi i k \cdot x} u_k(x).\] 
By elementary calculation,  we get 
\begin{equation}\label{Floquet unitary conjungate}
  \bigg( U  H U^{-1} (u_j(\cdot))_{j\in\Lambda_N}\bigg)_k(x)=\sum_{j\in \Lambda_N} h_{k-j}(x)u_j(x)+V(k)u_k(x), 
\end{equation}
where $h_k(x)=(U  \mcF^{-1} h)_k(x)$  corresponds to vector with the  symbol $h(x)$ in $\mcH$. After fixing $x\in (\frac{\T}{2N+1})^d$, the operator \eqref{Floquet unitary conjungate}  becomes the fiber matrix 
\begin{equation}\label{Floquet matrix}
  M^N(x)=(h_{k-j}(x))_{k\in\Lambda_N,j\in\Lambda_N}+V|_{\Lambda_N}. 
\end{equation}
We  call $x$ the \textbf{Floquet quasi-momentum}. 

By  the \textbf{Floquet eigenvalue and Floquet eigenfunction} of $H$  at the  quasi-momentum $x$,  we mean  the eigenvalue and eigenfunction of $M^N(x)$. In other words,  they are the value $E$ and vector $u\in \ell^2(\Z^d)$ such that
\begin{equation*}
  \left\{
    \begin{aligned}
      &Hu=Eu, \\
      &u_{j+m}=e^{-2\pi i m\cdot x} u_j \ {\rm for} \ j\in \Z^d,m\in [(2N+1)\Z]^d.
    \end{aligned}\right.
\end{equation*} 
Denote by $\{E_s(x)\}_{s\in \Lambda_N}$ all eigenvalues of  $M^N(x)$.  Then    we have 
\begin{equation}\label{Floquet spectrum}
  \sigma(H)=\bigcup_{x\in (\frac{\T}{2N+1})^d} \sigma(M^N(x))=\left\{E_s(x):\ s\in \Lambda_N,x\in (\frac{\T}{2N+1})^d\right\}. 
\end{equation} 
$\qquad$\\

If $V=0$, then the Floquet eigenvalue and Floquet eigenfunction of $H$  at  $x$ are 
\begin{equation}\label{Floquet basis}
  E_s(x)=h(x+\frac{2\pi s}{2N+1}),\ \beta_s(x)=\frac{1}{(2N+1)^{\frac{d}{2}}}(e^{-2\pi i (x+\frac{s}{2N+1})\cdot k})_{k\in \Lambda_N} \in \mcH,
\end{equation}
where $s\in \Lambda_N.$
This set of  Floquet eigenfunctions forms the  orthonormal basis of $\ell^2({\Lambda_N})$, which is  called the \textbf{Floquet basis}. The standard  basis in $\ell^2(\Lambda_N)$ can be represented in  the Floquet basis as 
\begin{align*}
  \delta_{l}=\sum_{s\in \Lambda_N} \frac{1}{(2N+1)^{\frac{d}{2}}}e^{2\pi i (x+\frac{s}{2N+1})\cdot l}\cdot \beta_s(x).
\end{align*}
To eliminate  the dependence on  $x$ of coordinates, we  define the modified canonical basis 
\begin{equation}\label{modified canonical basis}
  v_l=e^{-2\pi i x\cdot l}\delta_{l} =\sum_{s\in \Lambda_N} \frac{1}{(2N+1)^{\frac{d}{2}}}e^{2\pi i \frac{s}{2N+1}\cdot l}\cdot \beta_s(x). 
\end{equation}

\section{A quantitative  uncertainty principle}\label{UP}
In this section, we introduce a quantitative  uncertainty principle of   Klopp \cite{Klopp02}. Consider first  the discrete Fourier transform on the finite Abelian group $\Z_{2N+1}^d$ 
\begin{align*}
  \mcF_N: \ \ell^2(\Z_{2N+1}^d)&\to \ell^2(\Z_{2N+1}^d), \\
     a= (a_n)_{n\in \Z_{2N+1}^d }& \longmapsto \mcF_N a  = \hat{a},
\end{align*}
where 
\begin{equation}\label{discrete Fourier transform}
  \hat{a}_l=(\mcF_N a)_l= \sum_{n\in \Z_{2N+1}^d } a_n\cdot  \frac{1}{(2N+1)^{\frac{d}{2}}}e^{-2\pi i \frac{n}{2N+1}\cdot l}. 
\end{equation}
The quantitative uncertainty principle indicates that, if $a$ is supported in a $K$-size block in $\Z^d_{2N+1}$, then $\hat{a}$ can be  nearly constant in a $\frac{N}{K}$-size block.  
More precisely, we have 
\begin{lem}[\cite{Klopp02}, Lemma 6.2]\label{discrete UP}
Assume $N, L, K, K', L'$ are positive integers such that
\begin{itemize}
    \item $2N + 1 = (2K + 1)(2L + 1) = (2K' + 1)(2L' + 1)$; 
    \item $K < K'$ and $L' < L$.
\end{itemize}
Let  $a = (a_n)_{n \in \mathbb{Z}_{2N+1}^d} \in \ell^2(\mathbb{Z}_{2N+1}^d)$  satisfy  $a_n = 0$ for $|n| > K$. Then there exists some $b \in \ell^2(\mathbb{Z}_{2N+1}^d)$ such that
\begin{enumerate}
    \item $\|a - b\|_{\ell^2(\mathbb{Z}_{2N+1}^d)} \leq C_{K, K'} \|a\|_{\ell^2(\mathbb{Z}_{2N+1}^d)},$  where $0<C_{K, K'} {\underset{K/K' \to 0}\sim} K/K'$;
    \item For $l' \in \mathbb{Z}_{2L'+1}^d$ and $k' \in \mathbb{Z}_{2K'+1}^d$, we have $\hat{b}_{l'+k'(2L'+1)} = \hat{b}_{k'(2L'+1)}$;
    \item $\|a\|_{\ell^2(\mathbb{Z}_{2N+1}^d)} = \|b \|_{\ell^2(\mathbb{Z}_{2N+1}^d)}$.
\end{enumerate}
\end{lem}

\section{The Dudley's  estimate}\label{Dudley}
In this section,  we introduce some standard probabilistic estimates  needed in the derivation of \eqref{2.42} and \eqref{2.43}.  Most of those  can be found in \cite{AAGM15}. 

Let $(\Omega,\mcF,\P)$ be a probability space, and let $f$ be  a random variable on it.  Denote 
\[ \psi_{\alpha}(t)=e^{t^{\alpha}}-1,\  t\in [0,\infty),\alpha\geq 1.\]
Define the $\psi_{\alpha}$-Orlicz norm of $f$ as  
\[\| f\|_{\psi_{\alpha}}:= \inf\left\{\lambda >0:\  \int_{\Omega} \psi_{\alpha} \left(\frac{|f|}{\lambda}\right) \ d\P\leq 1\right\},\]
and  the Orlicz space as
\[L^{\psi_{\alpha}}(\Omega,\P)=\{f:\ \|f\|_{\psi_{\alpha}}<\infty \}.\]
We have the  Chernoff  estimate 
\begin{thm}[Chernoff bound]\label{Chernoff bound}
  If $\| f\|_{\psi_{\alpha}}<\infty$, then 
\[\P(|f|\geq t)\leq 2e^{- (\frac{t}{\|f\|_{\psi_{\alpha}}})^{\alpha}}.\]
\end{thm}
\begin{proof}
  By the Chebyshev's  inequality, we have 
  \begin{align*}
    \P(|f|\geq t) & =\P\left(  e^{(\frac{|f|}{\|f\|_{\psi_{\alpha}}})^{\alpha} }  \geq e^{(\frac{t}{\|f\|_{\psi_{\alpha}}})^{\alpha} } \right) \\
   &\leq e^{-(\frac{t}{\|f\|_{\psi_{\alpha}}})^{\alpha} } \cdot \E \left( e^{(\frac{|f|}{\|f\|_{\psi_{\alpha}}})^{\alpha} }\right) \\
   &= 2 e^{-(\frac{t}{\|f\|_{\psi_{\alpha}}})^{\alpha} }. 
  \end{align*}
\end{proof}

Another useful estimate  in $L^{\psi_{\alpha}}(\Omega,\P)$ is an analogue of \cite[Proposition 3.5.8]{AAGM15}:
\begin{thm}[Dudley's $L^{\psi_{\alpha}}$-estimate]\label{Dudley inequality}
   Assume the  random variables $X_1,\cdots,X_N\in L^{\psi_{\alpha}}(\Omega,\P)$. Then 
  \[\left\| \max_{1\leq i\leq N} |X_i| \right\|_{\psi_{\alpha}} \lesssim_{\alpha} (\log N)^{\frac{1}{\alpha}} \cdot \max_{1\leq i\leq N} \|X_i\|_{\psi_{\alpha}}.\]
\end{thm}
\begin{proof}
  Denote $\max_{1\leq i\leq N} \|X_i\|_{\psi_{\alpha}}=b$.
  
  First, by Theorem \ref{Chernoff bound}, we have for $p\geq \alpha,1\leq i\leq N,$ 
  \begin{align*}
    \E(|X_i|^p)& =p\int_{0}^{\infty} t^{p-1} \P(|X_i|\geq t) dt \leq 2p\int_{0}^{\infty} t^{p-1} e^{-(\frac{t}{b})^{\alpha}} dt =2 b^p \Gamma(\frac{p}{\alpha}+1). 
  \end{align*}
  Applying Stirling's formula gives 
  \begin{equation}\label{p-norm}
    \E(|X_i|^p)\lesssim b^p \sqrt{\frac{p}{\alpha}} (\frac{p}{e\alpha})^{\frac{p}{\alpha}} \ {\rm for} \ p\geq \alpha ,1\leq i\leq N. 
  \end{equation}
  
Next, for any $p\geq \alpha,q\geq 1$, one has 
  \begin{align*}
    \E(\big|\max_{1\leq i\leq N} |X_i| \big|^p) \leq  \E( (\sum_{i=1}^{N}|X_i|^{qp})^{\frac{1}{q}}) \leq \left( \E(\sum_{i=1}^{N}|X_i|^{qp})\right)^{\frac{1}{q}}, 
  \end{align*}
  where in   the last inequality   we use the  H\"older inequality.  As a result, 
  \begin{equation*}
    \| \max_{1\leq i\leq N} |X_i| \|_p\leq \left(\sum_{i=1}^{N} \E (|X_i|^{qp})\right)^{\frac{1}{qp}}. 
  \end{equation*}
  Now applying  \eqref{p-norm} implies 
  \begin{align*}
    \| \max_{1\leq i\leq N} |X_i| \|_p & \lesssim \left(N b^{qp} \sqrt{\frac{qp}{\alpha}} (\frac{qp}{e\alpha})^{\frac{qp}{\alpha}} \right)^{\frac{1}{qp}} \\
     &\lesssim_{\alpha} N^{\frac{1}{qp}} b (qp)^{\frac{1}{\alpha}}. 
  \end{align*}
  By taking  $q=\log N$,  we get 
  \begin{equation}\label{C3}
    \sup_{p\geq \alpha} \frac{\| \max_{1\leq i\leq N} |X_i| \|_p}{p^{1/ {\alpha}}} \lesssim_{\alpha} b (\log N)^{\frac{1}{\alpha}}. 
  \end{equation}

  Finally, combining \cite[Lemma 3.5.5]{AAGM15}  and  \eqref{C3}  finishes  the proof.
\end{proof}

The last fact is the sub-orthogonal property of sub-Gaussian random variables. We say $X$ is sub-Gaussian if  $X\in L^{\psi_2}(\Omega,\P)$.
\begin{thm}[\cite{V18}, Theorem 2.6.1]\label{sub orthogonal} 
Assume $X_1,\cdots, X_N$ are independent {\bf mean-zero} sub-Gaussian random variables. Then 
  \[\|\sum_{i=1}^{N}X_i\|_{\psi_2}^2 \lesssim \sum_{i=1}^{N}\|X_i\|_{\psi_2}^2.\]
\end{thm}

\section*{Acknowledgement}
This work  is  supported by  the National Key R\&D Program
of China under Grant 2023YFA1008801.
 Y. Shi is supported by the NSFC (12271380) and  Z. Zhang is  supported by the NSFC (12288101).   

\section*{Data Availability}
		The manuscript has no associated data.
		\section*{Declarations}
		{\bf Conflicts of interest} \ The authors  state  that there is no conflict of interest.

\bibliographystyle{alpha}

\end{document}